\documentclass[12pt,reqno]{amsart}
\usepackage{amsmath}
\usepackage{amsfonts}
\usepackage{mathrsfs}
\usepackage {amssymb,euscript}
\usepackage {amsmath}
\usepackage {amsthm}
\usepackage {amscd}
\usepackage{geometry}
\usepackage{hyperref}
\usepackage{color}
\usepackage{epsfig}
\usepackage{caption}

\usepackage{tikz}

\usepackage{graphics}

\title[Trapped surfaces arising from mild incoming radiation]{Trapped surfaces in vacuum arising dynamically from mild incoming radiation}
\date{\today}

\author{Xinliang An}
\address{Department of Mathematics, University of Toronto, Toronto, Ontario M5S 2E4, Canada}
\email{xinliang.an@utoronto.ca}

\author{Jonathan Luk}
\address{Department of Mathematics, Stanford University, Palo Alto, CA 94305, USA}
\email{jluk@stanford.edu}

\theoremstyle{definition}
\newtheorem{lemma}{Lemma}[section]

\newtheorem{proposition}[lemma]{Proposition}
\newtheorem{theorem}[lemma]{Theorem}
\newtheorem{corollary}[lemma]{Corollary}

\newtheorem{remark}{Remark}

\numberwithin{equation}{section}
\begin{document}

\newcommand{\ub}{\underline{u}}
\newcommand{\Cb}{\underline{C}}
\newcommand{\Lb}{\underline{L}}
\newcommand{\Lh}{\hat{L}}
\newcommand{\Lbh}{\hat{\Lb}}
\newcommand{\phib}{\underline{\phi}}
\newcommand{\Phib}{\underline{\Phi}}
\newcommand{\Db}{\underline{D}}
\newcommand{\Dh}{\hat{D}}
\newcommand{\Dbh}{\hat{\Db}}
\newcommand{\omb}{\underline{\omega}}
\newcommand{\omh}{\hat{\omega}}
\newcommand{\ombh}{\hat{\omb}}
\newcommand{\Pb}{\underline{P}}
\newcommand{\chib}{\underline{\chi}}
\newcommand{\chih}{\hat{\chi}}
\newcommand{\chibh}{\hat{\chib}}

\newcommand{\alb}{\underline{\alpha}}
\newcommand{\zeb}{\underline{\zeta}}
\newcommand{\beb}{\underline{\beta}}
\newcommand{\etb}{\underline{\eta}}
\newcommand{\Mb}{\underline{M}}
\newcommand{\oth}{\hat{\otimes}}

%$\alb,\beb,\etb,\zeb,\omb,\chib,\chih,\chibh,\oth,\phib,\Phib$

\def\a {\alpha}
\def\b {\beta}
\def\ab {\alphab}
\def\bb {\betab}
\def\nab {\nabla}

%$\a,\b,\bb,\ab, \nab$

\def\ub {\underline{u}}
\def\th {\theta}
\def\Lb {\underline{L}}
\def\Hb {\underline{H}}
\def\chib {\underline{\chi}}
\def\chih {\hat{\chi}}
\def\chibh {\hat{\underline{\chi}}}
\def\omegab {\underline{\omega}}
\def\etab {\underline{\eta}}
\def\betab {\underline{\beta}}
\def\alphab {\underline{\alpha}}
\def\Psib {\underline{\Psi}}
\def\hot{\widehat{\otimes}}
\def\Phib {\underline{\Phi}}
\def\thb {\underline{\theta}}
\def\t {\tilde}
\def\st {\tilde{s}}

%%%%%%%%%
\def\rhoc{\check{\rho}}
\def\sigmac{\check{\sigma}}
\def\Psic{\check{\Psi}}
\def\kappab{\underline{\kappa}}
\def\betabc {\check{\underline{\beta}}}

%%%%%%%%%%%%
%%%%%%%%%%%%
\def\d {\delta}
\def\f {\frac}
\def\i {\infty}
\def\l {\bigg(}
\def\r {\bigg)}
\def\S {S_{u,\underline{u}}}
\def\o{\omega}
\def\O{\Omega}\
\def\be{\begin{equation}\begin{split}}
\def\en{\end{split}\end{equation}}
\def\at{a^{\frac{1}{2}}}
\def\af{a^{\frac{1}{4}}}
\def\od{\omega^{\dagger}}
\def\ombd{\underline{\omega}^{\dagger}}
\def\K{K-\frac{1}{|u|^2}}
\def\ut{\frac{1}{|u|^2}}
\def\s{\frac{\delta a^{\frac{1}{2}}}{|u|}}
\def\Kb{K-\frac{1}{(u+\underline{u})^2}}
%%%%%
\def\de{\delta}
\def\ls{\lesssim}
\def\om{\omega}
\def\Om{\Omega}
%%%%%%%%%%%%%%
%%%%%%%%%%%%%%

\newcommand{\e}{\epsilon}
\newcommand{\et} {\frac{\epsilon}{2}}
\newcommand{\ef} {\frac{\epsilon}{4}}
\newcommand{\LH} {L^2(H_u)}
\newcommand{\LHb} {L^2(\underline{H}_{\underline{u}})}
\newcommand{\M} {\mathcal}
\newcommand{\TM} {\tilde{\mathcal}}
\newcommand{\p}{\psi\hspace{1pt}}
\newcommand{\q}{\underline{\psi}\hspace{1pt}}
\newcommand{\Li}{_{L^{\infty}(S_{u,\underline{u}})}}
\newcommand{\Lt}{_{L^{2}(S)}}
\newcommand{\da}{\delta^{-\frac{\epsilon}{2}}}
\newcommand{\db}{\delta^{1-\frac{\epsilon}{2}}}
\newcommand{\D}{\Delta}

%%%%%%%%%%%%
%%%%%%%%%%%%%

\renewcommand{\div}{\mbox{div }}
\newcommand{\curl}{\mbox{curl }}
\newcommand{\trchb}{\mbox{tr} \chib}
\def\trch{\mbox{tr}\chi}
\newcommand{\tr}{\mbox{tr}}

\newcommand{\Ls}{{\mathcal L} \mkern-10mu /\,}
\newcommand{\eps}{{\epsilon} \mkern-8mu /\,}
%%%%%%%%%

\newcommand{\xib}{\underline{\xi}}
\newcommand{\psib}{\underline{\psi}}
\newcommand{\rhob}{\underline{\rho}}
\newcommand{\thetab}{\underline{\theta}}
\newcommand{\gammab}{\underline{\gamma}}
\newcommand{\nub}{\underline{\nu}}
\newcommand{\lb}{\underline{l}}
\newcommand{\mub}{\underline{\mu}}
\newcommand{\Xib}{\underline{\Xi}}
\newcommand{\Thetab}{\underline{\Theta}}
\newcommand{\Lambdab}{\underline{\Lambda}}
\newcommand{\vphb}{\underline{\varphi}}

\newcommand{\ih}{\hat{i}}

\newcommand{\tcL}{\widetilde{\mathscr{L}}}

\newcommand{\sRic}{Ric\mkern-19mu /\,\,\,\,}
\newcommand{\sL}{{\cal L}\mkern-10mu /}
\newcommand{\sLh}{\hat{\sL}}
\newcommand{\sg}{g\mkern-9mu /}
\newcommand{\seps}{\epsilon\mkern-8mu /}
\newcommand{\sd}{d\mkern-10mu /}
\newcommand{\sR}{R\mkern-10mu /}
\newcommand{\snab}{\nabla\mkern-13mu /}
\newcommand{\sdiv}{\mbox{div}\mkern-19mu /\,\,\,\,}
\newcommand{\scurl}{\mbox{curl}\mkern-19mu /\,\,\,\,}
\newcommand{\slap}{\mbox{$\triangle  \mkern-13mu / \,$}}
\newcommand{\sGamma}{\Gamma\mkern-10mu /}
\newcommand{\somega}{\omega\mkern-10mu /}
\newcommand{\somb}{\omb\mkern-10mu /}
\newcommand{\spi}{\pi\mkern-10mu /}
\newcommand{\sJ}{J\mkern-10mu /}
\renewcommand{\sp}{p\mkern-9mu /}
\newcommand{\su}{u\mkern-8mu /}

\begin{abstract}
In this paper, we study the ``minimal requirement'' on the incoming radiation that guarantees a trapped surface to form in vacuum. First, we extend the region of existence in Christodoulou's theorem on the formation of trapped surfaces and consequently show that the lower bound required to form a trapped surface can be relaxed. Second, we demonstrate that trapped surfaces form dynamically from a class of initial data which are large merely in a scaling-critical norm. This result is motivated in part by the scaling in Christodoulou's formation of trapped surfaces theorem for the Einstein-scalar field system in spherical symmetry.
\end{abstract}
\maketitle

\section{Introduction}

In this paper, we study the formation of trapped surfaces by the focusing of ``mild'' incoming radiation in a $(3+1)$-dimensional spacetime $(\mathcal M,g)$ satisfying the vacuum Einstein equations
\begin{equation}\label{vacuum}
\mbox{Ric}_{\mu\nu}=0.
\end{equation}
A trapped surface is a 2-surface such that both null expansions are negative on every point on the surface, i.e., the area element of a trapped surface is infinitesimally decreasing along both families of null generators emanating from this surface. Trapped surfaces play an important role in the study of the solutions to \eqref{vacuum} due to the following celebrated incompleteness theorem of Penrose \cite{Penrose}:
\begin{theorem}[Penrose \cite{Penrose}]\label{Penrose.thm}
A spacetime with a non-compact Cauchy hypersurface satisfying the vacuum Einstein equations \eqref{vacuum} and containing a compact trapped surface is future causally geodesically incomplete. 
\end{theorem}
This theorem, however, does not show that trapped surfaces can arise from regular data without trapped surfaces. Indeed, this latter problem requires an understanding of the dynamics of the vacuum Einstein equations~\eqref{vacuum} in some large data regime.

\begin{figure}
\centering
\begin{tikzpicture}
%\draw (1.5,-0.5) node[very near end, sloped,below]{$H_{u_{\infty}}(u=u_{\infty})$}--(2,0);
\draw [white](3,-1)-- node[midway, sloped, below,black]{$H_1(u=1)$}(4,0);

\draw [white](2,2)--node [midway,sloped,above,black] {$\Hb_{\delta}(\ub=\delta)$}(4,0);
\draw [white](1,1)--node [midway,sloped, below,black] {$\Hb_{0}(\ub=0)$}(3,-1);
\draw [dashed] (0, 4)--(0, -4);
\draw [dashed] (0, -4)--(4,0)--(0,4);
\draw [dashed] (0,0)--(2,2);
%\draw [dashed] (0,1)--(1.5,2.5);
\draw [dashed] (0,-4)--(2,-2);
\draw [dashed] (0,2)--(3,-1);
\draw [very thick] (1,1)--(3,-1)--(4,0)--(2,2)--(1,1);
\fill[black!50!white] (1,1)--(3,-1)--(4,0)--(2,2)--(1,1);
\draw [white](1,1)-- node[midway,sloped,above,black]{$H_{u_*}$}(2,2);
\draw [->] (3.3,-0.6)-- node[midway, sloped, above,black]{$e_4$}(3.6,-0.3);
\draw [->] (1.4,1.3)-- node[midway, sloped, below,black]{$e_4$}(1.7,1.6);
\draw [->] (3.3,0.6)-- node[midway, sloped, below,black]{$e_3$}(2.7,1.2);
\draw [->] (2.4,-0.3)-- node[midway, sloped, above,black]{$e_3$}(1.7,0.4);
\end{tikzpicture}
\caption{Setup in Theorem \ref{Chr.thm}.}
\end{figure}
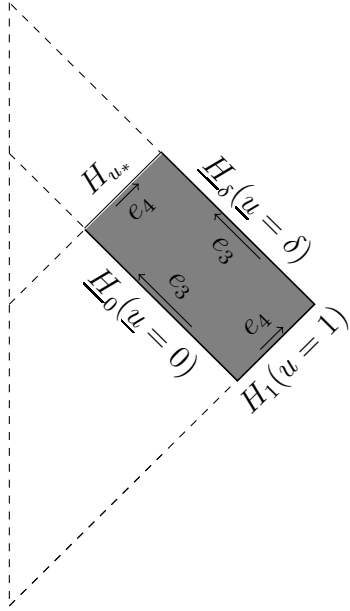

In a monumental work, Christodoulou \cite{Chr:book} showed that a trapped surface can be formed dynamically starting from regular initial data free of trapped surfaces. Christodoulou studied\footnote{Strictly speaking, the original theorem of Christodoulou is stated such that the initial data is at past null infinity and the radius of the final (trapped) sphere is of radius $\approx 1$. Nevertheless, a simple rescaling implies also a version of the theorem such that the initial inner sphere is of radius $1$ as in the case under consideration.}\label{fn1} the characteristic initial value problem with data posed on a truncated incoming cone $\Hb_0$ and a truncated outgoing cone $H_1$, which intersect at a 2-sphere $S_{1,0}$ (see Figure 1). The data on $\Hb_0$ are prescribed to coincide with a backward light cone in Minkowski space such that the sphere $S_{1,0}$ is the standard $2$-sphere with radius $1$. On the other hand, the data on $H_1$ are given in a region with a short characteristic length $\ub\leq \delta$ and such that the traceless part of the null second fundamental form $\chih$ is large in terms of $\delta$. This special form of initial data was termed a ``short pulse'' by Christodoulou. As a consequence of the short pulse ansatz, Christodoulou was able to consider a hierarchy of large and small quantities, parametrized by the smallness parameter $\delta$, whose sizes are preserved by the nonlinear evolution. Therefore, despite being a problem in a large data regime, a long time existence theorem can be established. Moreover, a sufficient condition, i.e., that the incoming radiation per unit solid angle is bounded uniformly below independent of $\de$, is identified such that a trapped surface is guaranteed to form in the causal future of the data, within the domain of existence. We summarize\footnote{We briefly explain the notation necessary to understand the following theorem but refer the readers to later sections for more details. We foliate the spacetime by a double null foliation $(u,\ub)$ and denote the intersections of the $u=\mbox{constant}$ and $\ub=\mbox{constant}$ hypersurfaces by $S_{u,\ub}$. Topologically, each $S_{u,\ub}$ is a 2-sphere. Associated to the double null foliation are the normalized null vectors $(e_3,e_4)$. In the statement of the theorem, all integrations are with respect to the natural volume form associated to the induced metric $\gamma$ on $S_{u,\ub}$. The differential operator $\nab$ is defined to be the Levi-Civita connection associated to $\gamma$ and $\nab_4$ is defined as the projection of the spacetime covariant differentiation in the $e_4$ direction to $TS_{u,\ub}$.} Christodoulou's result\footnote{Again, as mentioned in footnote \ref{fn1}, Christodoulou's original result allows the initial data to be posed at past null infinity. Here, we only mention a version in a finite region.} as follows:
\begin{theorem}[Christodoulou \cite{Chr:book}]\label{Chr.thm}
Consider the characteristic initial value problem for \eqref{vacuum} such that $\Hb_0$ coincides with a backwards light cone\footnote{Here, and in the remainder of this paper, we normalize the $u$ coordinate on the backwards light cone as follows. Let $C=\{(t,x_1,x_2,x_3):t\leq 0,\,t^2=x_1^2+x_2^2+x_3^2\}$ be the backward light cone in Minkowski space emanating from the origin. Define $r=\sqrt{x_1^2+x_2^2+x_3^2}$ and $u=\f12(-t+r)$. Notice in particular that $u=1$ on a standard sphere of radius $1$ and $u=0$ on the vertex.} in Minkowski space for $0\leq u\leq 1$. For every $B>0$ and $u_*\leq 1$, there exists $\de=\de(B,u_*)>0$ sufficiently small such that if the initial $\chih_0$, prescribed on $H_1$ for $0\leq \ub\leq \de$, satisfies
\begin{equation}\label{Chr.upper.bound}
\sum_{i\leq 5,\,j\leq 3}\de^{\frac 12+j}\|\nab^i\nab_4^j \chih_0\|_{L^\infty_{\ub}L^2(\S)}\leq B, 
\end{equation}
then the solution to \eqref{vacuum} remains regular in $u_*\leq u\leq 1$, $0\leq \ub\leq \de$. Moreover, if the initial data also verify the lower bound
\begin{equation}\label{Chr.lower.bound}
\inf_{\vartheta\in S_{1,0}} \int_0^{\de} |\chih_0|^2(\ub',\vartheta)\,d\ub' \geq M_* > 2(1-u_*),
\end{equation}
then, after choosing $\de$ to be smaller (depending on $B$, $u_*$ and $M_*$) if necessary, the sphere $S_{u_*,\de}$ is a trapped surface.
\end{theorem}

\begin{remark}
In Chapter 2 of \cite{Chr:book}, Christodoulou also constructed initial data satisfying both \eqref{Chr.upper.bound} and \eqref{Chr.lower.bound} simultaneously. Moreover, the initial data can be arranged to obey the additional bound 
$$\inf_{\vartheta\in S_{1,0}} \int_0^{\de} |\chih_0|^2(\ub',\vartheta)\,d\ub'<2$$
so that for $\de$ sufficiently small, it can be shown that the initial hypersurface $H_1$ indeed does not contain any trapped surfaces.
\end{remark}

In this paper, we extend Theorem \ref{Chr.thm} to show that trapped surfaces can arise dynamically from the focusing of ``milder'' incoming radiation. Our main result is Theorem \ref{main.thm.intro} below, which guarantees the formation of trapped surfaces when $\chih$ is much smaller than is required in Theorem \ref{Chr.thm}. We note here that the monumental theorem of Christodoulou-Klainerman \cite{Chr-Kl} on the stability of Minkowski spacetime states that the maximal Cauchy development of small initial data must be future causally geodesically complete and hence, by Theorem \ref{Penrose.thm}, is free of trapped surfaces. Therefore, any trapped surface formation result necessarily requires the data to be ``large'' in a certain sense. Our theorem below in particular allows the data for the metric to be large only in $H^{\f32}$ (as opposed to $H^1$ in Theorem \ref{Chr.thm}), which is a scaling-critical norm\footnote{See further discussions after Corollary \ref{Chr.cor}.} for the Einstein equations:

\begin{theorem} \label{main.thm.intro}
Consider the following characteristic initial value problem for \eqref{vacuum}: The initial incoming hypersurface $\Hb_0$ is required to coincide with a backwards light cone in Minkowski space with $0\leq u \leq 1$. On the initial outgoing hypersurface $H_1$, the initial $\chih$ satisfies
\begin{equation}\label{main.thm.intro.upper.bound}
\sum_{i\leq 7}\|\nab^i\chih_0\|_{L^\infty_{\ub}L^2(\S)}\leq a^{\frac 12} 
\end{equation}
for $0\leq \ub \leq \delta$. There exists a universal large constant $b_0$ such that if $b_0\leq b\leq a$ and $\de \at b< 1$, then the unique solution to \eqref{vacuum} remains regular in the region $\delta \at b\leq u \leq 1$, $ 0 \leq \ub \leq \delta$. Moreover, if the initial data also verify the lower bound
\begin{equation}\label{main.thm.intro.lower.bound}
\inf_{\vartheta\in S_{1,0}} \int_0^{\de} |\chih_0|^2(\ub',\vartheta)\,d\ub' \geq 4b\de\at,
\end{equation}
then the sphere $S_{b\de\at,\de}$ is a trapped surface.
\end{theorem}

\begin{figure}
\centering
\begin{tikzpicture}
%\draw (1.5,-0.5) node[very near end, sloped,below]{$H_{u_{\infty}}(u=u_{\infty})$}--(2,0);
\draw [white](3,-1)-- node[midway, sloped, below,black]{$H_1(u=1)$}(4,0);

\draw [white](0.5,1.5)-- node[midway,sloped,above,black]{$H_{b \delta \at}$}(1.5,2.5);
\draw [white](2,2)--node [midway,sloped,above,black] {$\Hb_{\delta}(\ub=\delta)$}(4,0);
\draw [white](1,1)--node [midway,sloped, below,black] {$\Hb_{0}(\ub=0)$}(3,-1);
\draw [dashed] (0, 4)--(0, -4);
\draw [dashed] (0, -4)--(4,0)--(0,4);
\draw [dashed] (0,0)--(2,2);
\draw [dashed] (0,1)--(1.5,2.5);
\draw [dashed] (0,-4)--(2,-2);
\draw [dashed] (0,2)--(3,-1);
\draw [very thick] (1,1)--(3,-1)--(4,0)--(2,2)--(1,1);
\draw [very thick] (1,1)--(0.5,1.5)--(1.5,2.5)--(2,2)--(1,1);
\fill[black!50!white] (1,1)--(3,-1)--(4,0)--(2,2)--(1,1);
\fill[black!20!white](1,1)--(0.5,1.5)--(1.5,2.5)--(2,2)--(1,1);
\draw [white](1,1)-- node[midway,sloped,above,black]{$H_{\delta a}$}(2,2);
\draw [->] (3.3,-0.6)-- node[midway, sloped, above,black]{$e_4$}(3.6,-0.3);
\draw [->] (1.4,1.3)-- node[midway, sloped, below,black]{$e_4$}(1.7,1.6);
\draw [->] (3.3,0.6)-- node[midway, sloped, below,black]{$e_3$}(2.7,1.2);
\draw [->] (2.4,-0.3)-- node[midway, sloped, above,black]{$e_3$}(1.7,0.4);
\end{tikzpicture}
\caption{Domain of Existence in Theorem \ref{main.thm.intro}.}
\end{figure}
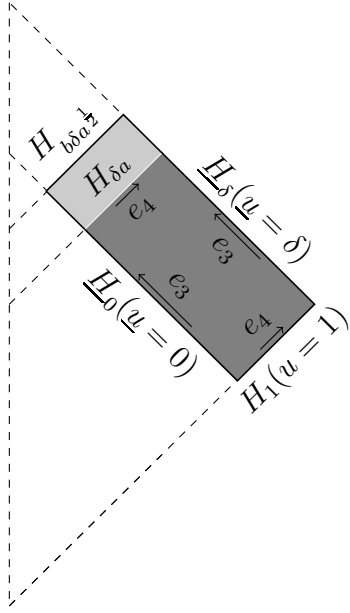   

The region of existence given by Theorem \ref{main.thm.intro} is depicted\footnote{Clearly, the depiction assumes that $b<\at$, which is not required in the statement of Theorem \ref{main.thm.intro}. However, as we will see below, this is the technically more difficult case and we will frequently restrict our attention to this scenario in the discussion in the introduction.} as the union of the darkly and lightly shaded regions in Figure 2. After choosing $a=B^2 \de^{-1}$ and $b=b_0$, we recover (a slightly weaker version\footnote{Strictly speaking, the numbers of derivatives required in Theorems \ref{Chr.thm} and \ref{main.thm.intro} are not the same. We nonetheless make a comparison of them since after choosing $a=B^2 \de^{-1}$ and $b=b_0$, the data have the same scaling in terms of $\delta$. In fact, in view of \cite{L-R:Interaction}, one expects that the condition \eqref{main.thm.intro.upper.bound} can be improved to requiring only $i\leq 5$ with better book-keeping.} of) Theorem \ref{Chr.thm} as a special case. This is because for fixed $b_0$, $B$, $u_*$ and $M_*$, we obviously have the inequalities $a=B^2\de^{-1}\geq b_0$, $u_*\geq b_0 B \de^{\frac 12}$ and $M_*>4 b_0 B \de^{\frac 12}$ after choosing $\de$ to be sufficiently small. Moreover, notice that even in this restricted case, Theorem \ref{Chr.thm} only guarantees that the solution remains regular in the darker shaded region in Figure 2, while Theorem \ref{main.thm.intro} proves that the solution also remains regular in the lightly shaded region. As a consequence, under the assumption of a similar \emph{upper bound} \eqref{Chr.upper.bound} in Theorem \ref{Chr.thm}, the \emph{lower bound}~\eqref{Chr.lower.bound} can be relaxed to 
\begin{equation*}
\inf_{\vartheta\in S_{1,0}} \int_0^{\de} |\chih_0|^2(\ub',\vartheta)\,d\ub' \geq 4 b_0 B \de^{\frac 12}.
\end{equation*}
We summarize this explicitly in the following corollary:
\begin{corollary}\label{Chr.cor}
Consider the following characteristic initial value problem for \eqref{vacuum}: The initial incoming hypersurface $\Hb_0$ is required to coincide with a backwards light cone in Minkowski space with $0\leq u \leq 1$. On the initial outgoing hypersurface $H_1$, the initial $\chih$ satisfies
\begin{equation*}
\sum_{i\leq 7}\de^{\f12}\|\nab^i\chih_0\|_{L^\infty_{\ub}L^2(\S)}\leq B
\end{equation*}
for $0\leq \ub \leq \delta$. Then there exists a universal large constant $b_0$ such that the solution to \eqref{vacuum} remains regular in $u_*\leq u\leq 1$, $0\leq \ub\leq \de$ for $u_*=b_0 B\de^{\f12}$. Moreover, if the initial data also verify the lower bound
\begin{equation*}
\inf_{\vartheta\in S_{1,0}} \int_0^{\de} |\chih_0|^2(\ub',\vartheta)\,d\ub' \geq 4 b_0 B \de^{\frac 12},
\end{equation*}
then the sphere $S_{u_*,\de}$ is a trapped surface.
\end{corollary}

More importantly, $a$ is allowed to be much smaller than $\de^{-1}$ in Theorem~\ref{main.thm.intro} and can be a large constant \emph{independent of $\de$}. In particular, if we think of 
$$\inf_{\vartheta\in S_{1,0}} \int_0^{\de} |\chih_0|^2(\ub',\vartheta)\,d\ub'$$
as a measure of the ``size'' of the incoming radiation, it can be of the same order of magnitude as the length scale $\de$. In terms of $L^2$ based Sobolev spaces\footnote{Here, we can understand fractional Sobolev ${H}^{s}$ spaces for instance in the given coordinate system $(\ub,\theta^1,\theta^2)$ (see Section \ref{coordinates}).}, there exist data satisfying the assumptions of Theorem \ref{main.thm.intro} with metric that is only large in $H^{\f 32}$ but can be small in $H^s$ for any $s<\f32$. Indeed, since we have $\f{\partial}{\partial\ub}\gamma_{AB} =2\chi_{AB}$ for the induced metrics $\gamma_{AB}$ on the $2$-sphere $S_{1,\ub}$ in the $(\ub,\theta^1,\theta^2)$ coordinate system, and moreover the $\ub$-interval has length $\de$, there exist data verifying \eqref{main.thm.intro.upper.bound} and \eqref{main.thm.intro.lower.bound} in Theorem \ref{main.thm.intro} such that
$$\|\gamma\|_{H^s}\sim a^{\f12} \de^{\f 32-s}.$$
Taking $a=b_0$ and $\de$ sufficiently small, it is easy to see that such data are small in the $H^s$ norm whenever $s<\f 32$. This is in contrast to the data in Theorem \ref{Chr.thm} which are large in $H^s$ for all $s>1$. The significance of the $H^{\frac 32}$ space is that it is a critical space for the Einstein equations according to scaling considerations.

One motivation for Theorem \ref{main.thm.intro} is that it can be considered as a non-spherically symmetric counterpart of an analogous theorem by Christodoulou \cite{Chr.1} in the setting of the Einstein-scalar field equations in \emph{spherically symmetry}. In that context, the formation of trapped surfaces theorem requires only\footnote{On the other hand, we note emphatically that the spherically symmetric result \emph{does not require any upper bounds} either on $H_1$ or $\Hb_0$, in stark contrast to our main theorem. See a more detailed discussion in Section \ref{sec.SSESF}.} a lower bound as in Theorem \ref{main.thm.intro}, i.e., one that is much weaker compared to that in Theorem \ref{Chr.thm}. We remark that this sharper version of the theorem was used crucially in Christodoulou's resolution of the weak cosmic censorship conjecture, showing that ``naked singularities''\footnote{In this context, these are singularities not preceded by a trapped surface.} are non-generic\footnote{Here, genericity is understood in the sense that the non-generic set of initial data has co-dimension at least two in the BV topology for the derivative of the scalar field. Moreover, Christodoulou constructed examples of naked singularities for this system in \cite{Chr.2}, showing that the \emph{non-generic} condition is indeed necessary.} for this system in spherical symmetry. We will make a more detailed comparison with the spherically symmetric case in Section \ref{sec.SSESF}, after a discussion of other known extensions of Christodoulou's theorem (Theorem \ref{Chr.thm}) in Section \ref{sec.known.results}

The main analytical difficulty that we face in the present work is that under the assumptions of Theorem \ref{main.thm.intro}, a trapped surface can only form near (in terms of $\delta$) the vertex where certain geometric quantities become large. We thus have to introduce weighted estimates to carefully track the growth of various geometric quantities as we approach the vertex. Moreover, in order to obtain the existence result in the region $|u|\geq \de\at b$, these weighted estimates are coupled with renormalized energy estimates introduced in \cite{L-R:Propagation}, \cite{L-R:Interaction} and \cite{L}. We will discuss the main ideas of the proof in Section \ref{sec.main.ideas}.

\subsection{Known results on the formation of trapped surfaces for the vacuum Einstein equations}\label{sec.known.results}

Returning to the problem of dynamical formation of trapped surfaces in vacuum, various extensions of Theorem \ref{Chr.thm} have been achieved since Christodoulou's breakthrough in \cite{Chr:book}. Here, we first briefly discuss the works \cite{KR:Trapped}, \cite{KR:Scarred} and \cite{L-R:Interaction}, in particular since some ideas in our present paper are drawn from them. 

In \cite{KR:Trapped} and \cite{KR:Scarred}, Klainerman-Rodnianski extended Theorem \ref{Chr.thm} by allowing the angular derivatives of the initial data to be large in terms of inverse powers of $\de$ in accordance with the parabolic scaling on null hypersurfaces for the Einstein equations. This is in contrast to Theorem \ref{Chr.thm} where the angular derivatives of $\chih_0$ are required to satisfy similar bounds as $\chih_0$ itself. The class of admissible initial data in \cite{KR:Scarred} is moreover critical with respect to the parabolic scaling on null hypersurface. Moreover, by considering a relaxed hierarchy of large and small quantities, Klainerman-Rodnianski showed that most of the geometric quantities obey scale-invariant estimates and this observation allowed them to greatly simplify the proof of \cite{Chr:book}.

Another extension of Theorem \ref{Chr.thm} was achieved in \cite{L-R:Interaction} in which the assumption on the triviality of data on $\Hb_0$ is replaced by assuming only some regularity conditions on the initial cone. This result was in fact a corollary of a more general existence theorem established in \cite{L-R:Interaction}, which was motivated by the problem of the interaction of impulsive gravitational waves. In these spacetimes, the Riemann curvature tensors admit delta singularities supported on transversely intersecting null hypersurfaces. It turns out that the estimates used to handle this type of singularities also can be applied to extend the formation of trapped surfaces theorem. As we will outline in Section \ref{sec.main.ideas}, these estimates will also play a crucial role in the analysis of this paper.

While the above theorems strengthened the existence part of\linebreak Christodoulou's theorem, a more recent work \cite{KLR} shows that the lower bound can also be relaxed. In fact, a lower bound is only necessary in the neighborhood of a \emph{single} null geodesic on the initial hypersurface $H_1$, i.e., the $\inf$ in \eqref{Chr.lower.bound} can be replaced by an $\sup$:

\begin{theorem}[Klainerman-Luk-Rodnianski \cite{KLR}]
Assume that the data for \eqref{vacuum} satisfy the condition \eqref{Chr.upper.bound} in Theorem \ref{Chr.thm}. If the initial data also verify the lower bound
\begin{equation}\label{KLR.lower.bound}
\sup_{\vartheta\in S_{1,0}} \int_0^{\de} |\chih_0|^2(\ub',\vartheta)\,d\ub' \geq M_* > 0,
\end{equation}
then, after choosing $\de$ smaller if necessary, a compact trapped surface can be guaranteed to formed to the future of the initial data, within the domain in which the solution remains regular.
\end{theorem}

Even with the improvement in \cite{KLR}, however, the lower bound for the $L^2$ integrals of $\chih_0$ along null generators is required to be independent of $\delta$, albeit only in a small set. By contrast, in our present work, we show that under a similar upper bound as \eqref{Chr.upper.bound}, the $L^2$ integrals of $\chih_0$ along null generators only have to be of size $\de^{\frac 12}$ in order to guarantee the formation of a trapped surface. Moreover, our result implies that as long as we also have an even better upper bound, the lower bound can be further relaxed. The techniques introduced in this work may possibly be combined with \cite{KLR} to obtain a trapped surface formation theorem for which the lower bound can be small in terms of $\de$ and is only required in a small set of angular directions, although we will not pursue this in the present work.

Regarding other related results, we also point out the treatments of \cite{An}, \cite{R-T} and \cite{Yu1}, as well as the work of Yu \cite{Yu2} on an analogous theorem for the Einstein-Maxwell system. While all of the aforementioned works posed initial data on null hypersurfaces, Li-Yu \cite{LY} combined Christodoulou's result with the Corvino-Schoen gluing method to construct a class of Cauchy data such that a trapped surface is guaranteed to form in the future. Finally, we refer the interested readers to the beautiful exposition \cite{Dafermos} for more background on the problem and a further discussion on the original work of Christodoulou.

\subsection{The Einstein-scalar field system in spherical symmetry}\label{sec.SSESF}

As mentioned previously, our present work is motivated in part by\linebreak Christodoulou's trapped surface formation theorem for the Einstein-scalar field system:
\begin{equation*}
\left\{
\begin{aligned}
\mbox{Ric}_{\mu\nu}-\f12 g_{\mu\nu} R&= 2T_{\mu\nu},\\
T_{\mu\nu}&=\partial_\mu\phi\partial_\nu\phi-\f12 g_{\mu\nu}(g^{-1})^{\a\beta}\partial_{\a}\phi\partial_{\beta}\phi,\\
\Box_g\phi&=0
\end{aligned}
\right.
\end{equation*}
in \emph{spherically symmetry}. Here, $\phi$ is a real valued function on the manifold $\mathcal M$. A special case of the result\footnote{The original theorem does not require the sphere defined as the intersection of the initial hypersurfaces to be of area radius $1$. We only state this special case for easy comparison with Theorem \ref{main.thm.intro}. The general case can be recovered by a scaling argument. Moreover, the original result gives a slightly sharper bounds in terms of the constants involved. We refer the readers to \cite{Chr.1} for the original statement.} of Christodoulou in \cite{Chr.1} can be repharsed as follows:

\begin{theorem}[Christodoulou \cite{Chr.1}]\label{Chr.thm.SS}
Consider the following characteristic initial value problem for the Einstein scalar field system \emph{with spherically symmetric data}. The initial hypersurface $\Hb_0$ is a cone such that the sphere $S_{1,0}$ has area radius $r=1$. The data on $\Hb_0$ are otherwise arbitrary. On the initial hypersurface $H_1$, after normalizing the outgoing null coordinate $\ub$ by the condition $\partial_{\ub} r=\f 12$, we prescribe $\phi$ in the region $\ub\in [0,\de]$. Then there exists a universal constant $C_0$ such that as long as the data on $H_1$ satisfy
$$\int_0^{\de} (\partial_{\ub}\phi)^2(\ub') d\ub' \geq C_0 \de\log(\f1{\de}), $$
a trapped surface must form in the causal domain of the initial data.
\end{theorem}

In contrast to any of the trapped surface formation theorems without symmetry assumptions that we have mentioned above, Theorem \ref{Chr.thm.SS} requires no upper bounds for the initial data, either on $H_1$ or $\Hb_0$. In fact, it is precisely because the data on $\Hb_0$ are allowed to be singular that Christodoulou was able to apply it to resolve the weak cosmic censorship conjecture in \cite{Chr.3}. Moreover, not only did Christodoulou show that a trapped surface must form in \cite{Chr.1}, he also gave a complete description of the maximal Cauchy development of the initial data, proving that the spacetime possesses a complete null infinity and has a black hole region with a spacelike singularity in its interior.

The proof of this theorem relies on a special monotonicity formula that holds for this system of equations in spherical symmetry. It immediately breaks down even if we restrict to a small perturbation of a spherically symmetric background.

Up to the logarithmic factor $\log(\frac 1{\de})$, Theorem \ref{Chr.thm.SS} is sharp. More precisely, there exists a positive constant $c$ such that for every $\de>0$ sufficiently small, there are examples of initial data satisfying
\begin{equation}\label{Chr.BV}
\int_0^{\de} (\partial_{\ub}\phi)^2(\ub') d\ub'\geq c \delta
\end{equation}
whose future development does not contain a trapped surface. That this is the case follows from Christodoulou's theorem\footnote{The original theorem in \cite{Chr.1.5} also gives quantitative estimates on the development of the data. We will only need that statement that it is free of trapped surfaces in order to contrast Theorem \ref{Chr.thm.SS}.} on the regularity of spacetimes with initial data of small BV norms:

\begin{theorem}[Christodoulou \cite{Chr.1.5}]\label{Chr.thm.BV}
Consider characteristic initial value problem for the Einstein scalar field system with spherically symmetric initial data given on a cone in $\ub\in [0,2]$ with $\partial_{\ub} r=\f 12$. Then there exists $\epsilon_0>0$ such that if 
$$\int_0^{\de} |\partial_{\ub}^2(r\phi)|(\ub') d\ub' \leq \epsilon_0,$$
then the causal domain of the initial data does not contain any trapped surfaces.
\end{theorem}

Using Theorem \ref{Chr.thm.BV}, in order to show that there exist initial data satisfying \eqref{Chr.BV} which give rise to spacetime without trapped surface, it suffices to show that for small but fixed $\epsilon_0>0$, there exists $c$ sufficiently small such that for every $\de>0$ we can find a function $\phi_{\de}:[0,\de]\to \mathbb R$ with the properties
\begin{equation*}
\int_0^{\de} (\partial_{\ub}\phi_{\delta})^2(\ub') d\ub'\geq c \delta
\end{equation*}
and
$$\int_0^{\de} |\partial_{\ub}^2(r\phi_{\delta})|(\ub') d\ub' \leq \epsilon_0.$$
This can be achieved by considering $\phi_{\de}(\ub)=\epsilon_0^2\de\xi(\f{\ub}{\de})$ for a fixed smooth function $\xi$ compactly supported in $[0,1]$ of $C^2$ norm $\approx 1$ and then taking $c$ to be sufficiently small.

Returning to the discussion of the result in this paper, notice that our main theorem (Theorem \ref{main.thm.intro}) can be thought of as a formal analogue of Theorem \ref{Chr.thm.SS} without symmetry assumptions in terms of scaling. In particular, if we formally compare the metric in Theorem \ref{main.thm.intro} with the scalar field in Theorem \ref{Chr.thm.SS}, we see that the initial data in both theorems are allowed to be small in $H^s$ for every $s<\f32$ and are large only in $H^{\f32}$. However, we emphasize that even though our result applies without any symmetry assumptions, it also requires much stronger assumptions on the initial data, including the very restrictive condition that the initial data have to be trivial on $\Hb_0$.

\subsection{Main ideas of the proof}\label{sec.main.ideas}

\subsubsection{General structure}
As in \cite{Chr:book}, \cite{KR:Trapped}, our proof is based on establishing a priori estimates for the vacuum Einstein equations in the double null foliation gauge. In this gauge, the Einstein equations exhibit a certain type of \emph{null structure}, in which large components are coupled with suitably small components, allowing all the estimates to be closed even in a large data regime.

More precisely, we foliate the spacetime under consideration by outgoing and incoming null hypersurfaces $H_u$ and $\Hb_{\ub}$ respectively. Associated to this double null foliation, we define a null frame $\{e_1,e_2,e_3,e_4\}$ such that $e_3$ (resp. $e_4$) is null and tangent to $\Hb_{\ub}$ (resp. $H_u$), and $\{e_1,e_2\}$ forms a frame tangent to the $2$-spheres $S_{u,\ub}$, defined to be the intersections of $H_u$ and $\Hb_{\ub}$. We then define the following Ricci coefficients
\begin{equation}\label{Ricci.def.intro}
\begin{split}
&\chi_{AB}=g(D_A e_4,e_B),\, \,\, \quad \chib_{AB}=g(D_A e_3,e_B),\\
&\eta_A=-\frac 12 g(D_3 e_A,e_4),\quad \etab_A=-\frac 12 g(D_4 e_A,e_3),\\
&\omega=-\frac 14 g(D_4 e_3,e_4),\quad\,\,\, \omegab=-\frac 14 g(D_3 e_4,e_3),\\
&\zeta_A=\frac 1 2 g(D_A e_4,e_3)
\end{split}
\end{equation}
and curvature components\footnote{Here, $^*R$ is the Hodge dual of $R$.}
 \begin{equation}\label{curv.def.intro}
\begin{aligned}
\a_{AB}&=R(e_A, e_4, e_B, e_4),\quad \, \,\,   \ab_{AB}=R(e_A, e_3, e_B, e_3),\\
\beta_A&= \frac 1 2 R(e_A,  e_4, e_3, e_4) ,\quad \bb_A =\frac 1 2 R(e_A,  e_3,  e_3, e_4),\\
\rho&=\frac 1 4 R(e_4,e_3, e_4,  e_3),\quad \sigma=\frac 1 4  \,^*R(e_4,e_3, e_4,  e_3)
\end{aligned}
\end{equation}
with respect to this null frame, where $A,B\in\{1,2\}$.

The general strategy of our proof is to obtain $L^2$ energy estimates for the curvature components and their derivatives and to derive bounds for the Ricci coefficients using transport equations and elliptic estimates. These estimates are highly coupled: Ricci coefficients and their derivatives arise as error terms in the energy estimates and curvature terms come up as a source for the transport equations for the Ricci coefficients. A bootstrap argument is therefore set up in order to obtain all these bounds. 

As mentioned earlier, one of the challenges of our problem at hand is that since the data is small in terms of $\delta$, a trapped surface can only form near (in terms of $\delta$) the vertex of the cone. On the other hand, various geometric quantities necessarily grow as the vertex is approached. It is therefore necessary to prove precise estimates on the growth of the geometric quantities. As a result, as we will describe below, all the estimates for the Ricci coefficients and curvature components are appropriately weighted in $u$. 

These weighted estimates are reminiscent of those used in the setting where $u$ is large, i.e., in a neighborhood of null infinity. In fact, they form an important part of the monumental proof of the nonlinear stability of Minkowski spacetime by Christodoulou-Klainerman \cite{Chr-Kl} (see also the work of Klainerman-Nicolo \cite{KNI:book}), as well as Christodoulou's result on the formation of trapped surfaces from past null infinity \cite{Chr:book}. We show in the present paper that these estimates can also be used near the vertex, where $u$ does not measure decay, but instead measures \emph{growth} of the geometric quantities.

In the following subsections, we will discuss these weighted estimates in the context of transport estimates for Ricci coefficients, energy estimates for the curvature components and also the elliptic estimates for the highest order derivatives for the Ricci coefficients. We will highlight some difficulties of the problem and explain the special structure of the Einstein equations that we use to overcome them.

\subsubsection{Weighted estimates for the Ricci coefficients}

As mentioned above, the Ricci coefficients are estimated in function spaces with weights in $u$. The Ricci coefficients are controlled using transport equations and the weights in the bounds are ultimately dictated by the non-integrability of the null mean curvature $\trchb$ of the incoming hypersurfaces along incoming null generators. More precisely, on the initial Minkowskian hypersurface $\Hb_0$, we have $\trchb=-\frac {2}{|u|}$, which is non-integrable in $u$. Consider for example the following transport equation for $\chih$:
$$\nab_3\chih+\f12 \trchb\chih=\cdots$$
Suppose the right hand side of the equation can be controlled and that $\trchb$ is close to its value on $\Hb_0$, i.e., $\trchb\approx -\f 2{|u|}$. Then the equation implies that we have
$$\chih\sim\f{\at}{|u|},$$
which $\to\infty$ as $u\to 0$. We therefore have to content with norms for $\chih$ which are weighted in $u$. Moreover, these terms enter into the equations for all the other Ricci coefficients and every geometric quantity has to be estimated in an appropriate weighted space.

We therefore apply the following strategy: Assume as one of the bootstrap assumptions that the value of $\trchb$ remains close to its Minkowskian value throughout the region of existence. We then prove weighted estimates for all the Ricci coefficients satisfying the $\nab_3$ equation with weights in $u$ that are determined by the $\trchb$ term. The Ricci coefficients satisfying $\nab_4$ equations will in turn have weighted that are dictated by the source terms.

We now describe the weighted estimates in more detail. Recalling the definitions for the Ricci coefficients in \eqref{Ricci.def.intro}, we introduce the following schematic notation:
$$\p\in\{\trch,\chih,\om\},\quad\q\in\left\{\trchb+\frac 2{|u|},\chibh,\eta,\etab,\omb\right\}.$$
They obey the following slightly schematic null structure equations
\begin{gather}\label{intro.eqn.1}
\nab_3\p+\f 1{|u|}\p =\q\,\q+\p\q+\nab\q+\mbox{curvature},\\
\label{intro.eqn.2}
\nab_4\q=\q\,\q+\p\q+\frac{1}{|u|}\p+\nab\p+\nab\q+\mbox{curvature}.
\end{gather}
At this point, in order to ease the exposition, we will first not discuss the terms with curvature and derivatives of Ricci coefficients on the right hand side.\footnote{Directly integrating these terms will of course lead to a loss of derivative and in order to tackle these terms, we need to combine with energy estimates and elliptic estimates which we will describe later.} To derive estimates from the schematic null structure equations, we first assume that $\q$ has small $u$-integral. This implies that the growth rate of $\p$ is determined by the non-integrable coefficient $\f 1{|u|}$.

Therefore, \eqref{intro.eqn.1} implies that
\begin{equation}\label{intro.weight.est.1}
\psi\sim \f{\at}{|u|}.
\end{equation}
Substituting this bound for $\p$ into the equation \eqref{intro.eqn.2}, we obtain the following estimates:
\begin{equation}\label{intro.weight.est.2}
\q\sim\f{\de\at}{|u|^2}.
\end{equation}
It is easy to see that these estimates are sharp, at least for some of the Ricci coefficients $\p$ and $\q$. For instance, as described above, the bounds for $\chih$ are dictated by linear theory and indeed can be shown to be sharp. At the same time, by considering the equation for $\chibh$
$$\nab_4\chibh=\trchb\chih+\cdots,$$
where $\cdots$ denotes terms that behave better, it is easy to show the estimates for $\chibh$ are also optimal. 

Observe that since we only work in the region $|u|\geq \de\at b$, the weighted estimates \eqref{intro.weight.est.1} and \eqref{intro.weight.est.2} above imply in particular that $\q$ is smaller than $\p$. On the other hand, as $u\to \delta\at b$, $\q$ can be in fact be large with estimates $\|\q\|_{L^\infty}\sim \frac{1}{\de\at b^2}$. Nevertheless, while $\q$ may be large in sup norms, \emph{its integral in $u$ is small}, as can be seen by the direct computation
$$\de\at\int_u^1 \frac{du'}{|u'|^2}\ls \frac{\de\at}{|u|}\ls \frac 1b. $$
In particular, the assumption above that $\q$ has small $u$-integral is valid and we can justify the heuristic argument above. Indeed, one sees that if $\q\sim \f{\de\at}{|u|^2}$, then according to the equation \eqref{intro.eqn.1}, $\p$ does satisfy the estimate $\p\sim\f{\at}{|u|}$ as predicted by the linear theory.

\subsubsection{Estimates for higher derivatives of the Ricci coefficients}

In order to close our estimates, we also need to control the higher derivatives of the Ricci coefficients with appropriate weights. Since the area of the 2-spheres $S_{u,\ub}$ scales like $|u|^2$, one expects that for every additional angular derivative on the Ricci coefficients, the bounds gets worse at least by a factor of $\frac{1}{|u|}$. We will show that in fact such estimates can be proved. In other words, we have a \emph{commutation principle} similar to that in \cite{DHR}: {\bf$|u|^i\nab^i\p$ and $|u|^i\nab^i\q$ obey estimates similar to those for $\p$ and $\q$.} More precisely, we have
$$|u|^i\nab^i\psi\sim \f{\at}{|u|},\quad |u|^i\nab^i\q\sim\f{\de\at}{|u|^2},$$
for $i\leq 2$. Moreover, for $i=3,4$, we also have $L^2$ estimates for $\nab^i\p$ and $\nab^i\q$ that are consistent with the above scaling. 

The observation that allows us to obtain these estimates is that the commutator $[\nab_3,\nab]$ takes the form
$$[\nab_3,\nab]=-\frac 12\trchb\nab+\cdots$$
where $\cdots$ denotes terms that behave better. Recall now that $\trchb\approx -\f2{|u|}$\linebreak and therefore the linear part of the commuted equation for $\p$ takes the form
$$\nab_3\nab^i\p-\frac{i+1}{u}\nab^i\p=\cdots$$
Recall moreover that the initial data obey the bound $|\nab^i \p|\ls \at$. Therefore, $\nab^i \p$ verifies the estimate
$$\nab^i\p\sim \f{\at}{|u|^{i+1}},$$
consistent with the commutation principle.

In order to obtain the appropriate weights for the angular derivatives of the Ricci coefficient $\q$ satisfying a $\nab_4$ equation, we notice that their weights are dictated by the source terms, which include the terms $\nab^i\p$ whose weights we have determined. In particular, this shows that $\nab^i\q$ also obeys estimates in accordance with the commutation principle.

\subsubsection{Reductive structure and improved estimates}

As we can already see in the above argument, some of the terms are borderline and there are no extra smallness in the error terms. Nevertheless, similar to \cite{Chr:book}, \cite{L:local} and \cite{L-R:Interaction}, there is a \emph{reductive structure} that allow us derive estimates in a sequence of steps, each of which involves error terms that either come with sufficient smallness or have already been controlled in the previous step. More precisely, we first prove the bounds for $\nab^i\p$. All the error terms on the right hand side behave at least $b^{-\f12}$ better than is necessary and we can obtain the desired bounds for $\nab^i\p$ under the bootstrap assumptions. We then turn to the equation for $\nab^i\q$. There are error terms in this equation without additional smallness, but we can use the estimates that we have just obtained for $\nab^i\p$ instead of using the bootstrap assumptions to control them.

Beyond the estimates described above, we also need to obtain improved estimates for the components $\trch-\f{2}{|u|}$, $\trchb+\f{2}{|u|}$, $\om$ and\footnote{$\mu$ is defined by $\mu=-\div\eta+K-\f{1}{|u|^2}$. See discussions in latter sections.} $\mu$ compared to what is suggested by scaling. They are needed in the elliptic estimates and the energy estimates since otherwise we would encounter $\f{1}{|u|}$ terms that we need to integrate in $u$, resulting in logarithmic losses. To explain this further, we consider the term $\trch-\f2{|u|}$. According to the previous discussion, scaling considerations suggest that 
$$|\trch| \ls \f{\at}{|u|}.$$
However, this will be insufficient to close the estimates and we need to obtain a bound
$$|\trch-\f2{|u|}|\ls \f{\de a}{|u|^2}.$$
Notice that when $|u|=\de\at b$, this gives
$$|\trch-\f2{|u|}|\ls \f{\de \at}{b|u|}$$
and is only better than the previous estimate by a constant factor $\f 1b$. Nevertheless, the fact that this is smaller for a large range of $u$ allows us to integrate in $u$ in some of the error terms to avoid logarithmic divergences. However, we note that the improved estimates for $\trch-\f{2}{|u|}$, $\trchb+\f{2}{|u|}$, $\om$ and $\mu$ are coupled and the proof of them requires another reductive structure. We refer the readers to the text for more details on this more refined reductive structure.

\subsubsection{Weighted energy estimates I: Renormalization}

In order to close our estimates, we need to obtain $L^2$ energy estimates for the curvature components in addition to the bounds for the Ricci coefficients described above. In particular, we need to prove sufficiently strong estimates for the curvature terms arising in the transport equations \eqref{intro.eqn.1} and \eqref{intro.eqn.2} for the Ricci coefficients.

These energy estimates are also weighted in $u$ in order to capture the growth near the vertex. However, unlike for the Ricci coefficients for which we can obtain estimates for each component separately, the energy estimates for different curvature components have to be derived at the same time. For instance, when trying to obtain the energy estimates for $\beta$ (recall the notation from \eqref{curv.def.intro}) and its derivative along outgoing null hypersurfaces, typically we need to prove at the same time the energy estimates for $\rho$ and $\sigma$ and their derivatives on incoming null hypersurfaces.

Now, the Codazzi equation on the $2$-spheres $S_{u,\ub}$ reads
$$\div\chih=\frac 12 \nabla \trch - \frac 12 (\eta-\etab)\cdot \left(\chih -\frac 1 2 \trch\right) -\beta.$$
By the estimates for $\p$ and $\q$ and the commutation principle, we hope to prove a bound consistent with
$$||u|^{i+2}\nab^i\beta|\ls \at.$$
Recalling that $\mbox{Area}(S_{u,\ub})\approx |u|^2$, we would like to prove the following $L^2(H_u)$ bound for $\nab^i\beta$:
$$\sum_{i\leq 4}\||u|^{i+1}\nab^i\beta\|_{L^2_{\ub}L^2(\S)}\ls \delta^{\f12}\at.$$
Now, if we were to prove estimates for $\nab^i\rho$ and $\nab^i\sigma$, in order to obtain the above $L^2(H_u)$ estimates for $\nab^i\beta$, we also need to prove at the same time that
\begin{equation}\label{rho.sigma.false}
\sum_{i\leq 4}\||u|^{i+1}\nab^i(\rho,\sigma)\|_{L^2_uL^2(\S)}\ls \delta^{\f12}\at.
\end{equation}
These estimates, however, are inconsistent with those dictated by scaling considerations and cannot be expected. To see this, notice the equation
$$\curl\eta =\sigma +\frac 1 2\chibh \wedge\chih.$$
By the estimates for $\chih$, $\chibh$, $\eta$ and the commutation principle, $\sigma$ obeys the pointwise bound
$$|\sigma|\ls \f{\de a}{|u|^3},$$
which is only consistent with \eqref{rho.sigma.false} in the region $|u|\geq \de a$, but cannot be proved in the full region $|u|\geq \de \at b$ in the case $b<\at$.

Instead, we perform renormalized energy estimates as in \cite{L-R:Propagation}, \cite{L-R:Interaction} and \cite{L}. We make the observation that the following \emph{renormalized curvature components}:
$$K-\f1{|u|^2}=-\rho+\frac 12\chih\cdot\chibh-\frac 14\trch\trchb-\f1{|u|^2},\quad\sigmac=\sigma+\frac 12\chibh\wedge\chih$$
behave better than the spacetime curvature components $\rho$ and $\sigma$. More precisely, we have
$$\left|\left(\K,\sigmac\right)\right|\ls \f{\de a^{\f12}}{|u|^3}$$
and we can therefore use the $(\beta,\K,\sigmac)$ equations
\begin{equation*}
\begin{split}
\nab_3\beta+\trchb\beta+\nab \left(\K\right)-^*\nab\sigma&=\cdots\\
\nab_4\left(\K\right)+\div\beta&=\cdots\\
\nab_4\sigmac+\div ^*\beta&=\cdots
\end{split}
\end{equation*}
to get the following desired estimates:
$$\sum_{i\leq 4}\bigg(\|u^{i+1}\nab^i\beta\|_{L^2_{\ub}L^2(\S)}+\left\|u^{i+1}\nab^i\left(\K,\sigmac\right)\right\|_{L^2_{\ub}L^2(\S)}\bigg)\ls \de^{\f12}\at.$$
In particular, we need to obtain energy estimates directly from the Bianchi equations as in \cite{Holzegel} instead of using the Bel-Robinson tensor (see \cite{Chr:book}, \cite{Chr-Kl}, \cite{KNI:book}, \cite{KR:Trapped}).

The reason that we can prove better estimates is that when we compare the equation for $(\K,\sigmac)$ to that for $(\rho,\sigma)$, we see that the most singular term $\chibh\alpha \sim \f{a}{|u|^3}$ drops off. In fact, as noted in \cite{L-R:Interaction}, by performing renormalized energy estimates, the estimates for $\alpha$ completely decouple and we do not need to prove any bounds for $\alpha$ component, which is very large in this setting. At the same time, we also do not need any information on $\alphab$, which allows the proof to be simplified.

\subsubsection{Weighted energy estimate II: Additional cancellations}
\hspace{-1.5ex}
Even after using the renormalization, we need to exploit additional cancellations in the error terms in order to close the energy estimates. One of these instances is that the renormalization introduces an error term in the equation for $K$ of the form
$$\nab_3 K+\trchb K=-\frac 12\trchb \div\eta+\cdots$$
At the highest level of derivatives $\nab \eta$ has to be retrieved in $L^2$ from the bounds for $K-\frac{1}{|u|^2}$ via elliptic estimates. Since $\trchb\sim-\frac 2{|u|}$, it would seem that we have an estimate of the type
\begin{align*}
&\quad\ \left\|u\left(K-\frac 1{|u|^2}\right)\right\|_{L^2_{\ub}L^2(\S)}^2(u)\\
&\ls \mbox{Data}+\int_u^1 \frac{1}{u'}\left\|u'\left(K-\frac 1{|u'|^2}\right)\right\|_{L^2_{\ub}L^2(S_{u',\ub})}^2 \,du'+\cdots
\end{align*}
Since $\frac{1}{u}$ is non-integrable, Gronwall's inequality would imply that this term grows as $u\to 0$ and the bound will be too weak to close the estimates. Instead, we note that there is a more subtle structure (for more details, we refer the readers to the proof of Proposition \ref{EE.2}) and we can rewrite the equation as
\begin{align}\label{K.rewrite}
&\quad\ \nab_3 K+\trchb K+\frac 12\trchb \div\eta\\
\notag &=\nab_3\left(K-\frac{1}{|u|^2}\right)+\frac 32\trchb \left(K-\frac{1}{|u|^2}\right)-\frac 12\trchb\mu+\cdots 
\end{align}
where $\mu$ is the mass aspect function given by 
$$\mu=-\div\eta+K-\frac{1}{|u|^2}$$
and $\cdots $ denotes terms that are under control.
The crucial observation is that we can prove a better estimate for \emph{the derivatives of} $\mu$ than what scaling naively suggests. More precisely, for derivatives of $\eta$, we expect using the heuristics in the previous section that 
$$|\nab^{i+1}\eta |\ls \frac{\de\at}{|u|^{i+3}}\quad \left(\mbox{or }\|\nab^{i+1}\eta \|_{L^2(\S)}\ls \frac{\de\at}{|u|^{i+2}}\right),$$
while for $\mu$, which is a special combination of derivative of $\eta$ and curvature, we have
$$|\nab^i\mu |\ls \frac{\de^2 a^{\frac 54} b^{\frac 14}}{|u|^{i+4}}\quad \quad \left(\mbox{or }\|\nab^{i}\mu \|_{L^2(\S)}\ls \frac{\de^2 a^{\frac 54} b^{\frac 14}}{|u|^{i+3}}\right),$$
as long as $i\geq 1$. Notice that in order to close the energy estimates, we only need to use this improvement for $i\geq 1$ and the extra power of $u$ allows us to close the estimates using Gronwall's inequality.

There is yet another miraculous structure in the Einstein equations that allows us to prove the energy estimates with the desired weights in $u$. The linear parts of the equations for $\K$ and $\sigmac$ from which we prove energy estimates in $L^2_{\ub}L^2(\S)$ take the form
\begin{equation*}
\begin{split}
\nab_3\left(\K\right)+\f32 \trchb \left(\K\right) &= \cdots ,\\
\nab_3\sigmac+\f32 \trchb \sigmac &= \cdots 
\end{split}
\end{equation*}
Therefore, the $\f32 \trchb$ factors dictate that we control $(K-\frac{1}{|u|^2},\sigmac)$ at least in the weighted space
$$\left\|u^2\left(K-\frac{1}{|u|^2},\sigmac\right)\right\|_{L^2_{\ub}L^2(\S)}.$$
On the other hand, the $L^2_{\ub}L^2(\S)$ energy estimates for $(K-\frac{1}{|u|^2})$ and $\sigmac$ are coupled with the $L^2_{u}L^2(\S)$ of $\betab$. However, naive scaling considerations may suggest that $\betab$ cannot be controlled in the weighted space $\|u^2\betab\|_{L^2_u L^2(\S)}$. One important observation is that there is a cancellation which results in $\betab$ being slightly better behaved than expected. More precisely, using the Codazzi equation, we can express $\betab$ in terms of the following Ricci coefficients:
$$\betab=\div\chibh -\frac 12 \nabla \trchb + \frac 12 (\eta-\etab)\cdot \left(\chibh-\frac 1 2   \trchb\right).$$
As we observed previously, the terms $\div\chibh$, $\trchb\eta$ and $\trchb\,\etab$ behave no better than $\frac{\de \at}{|u|^3}$in $L^\infty$. On the other hand, notice that the contributions from $\div\chibh$ and $\frac 14 (\eta-\etab)\trchb$ in the Codazzi equation actually cancel! Thus $\betab$ in fact behaves better and can be controlled in the desired weighted $L^2$ space. On a more technical level, in the $\nab_4\betab$ equation, the potentially deadly term $\trchb\nab\chih$ does not appear! Instead, in this equation, the only quadratic term that contains $\trchb$ is the term $\trchb\nab\trch$. As mentioned above, we have improve estimates for $\nab\trch=\nab(\trch-\f2{|u|})$ which are stronger than that for $\nab\chih$. This then allows us to close all the energy estimates.

\subsubsection{Elliptic estimates}

We now turn to the final technical difficulties in the proof of the main theorem. Notice that yet another consequence of using the renormalized energy estimates is that they introduce error terms with highest order derivatives in the Ricci coefficients. As a result, it is necessary to control those terms via elliptic estimates. While the procedure of deriving highest order bounds for the Ricci coefficients from the energy estimates via elliptic estimates is by now standard (see \cite{Chr:book}, \cite{Chr-Kl}, \cite{KNI:book} and \cite{KR:Trapped}), many technical difficulties arise when coupled with the scale-invariant weight estimates. In particular, we also need to obtain improved estimates for certain components using elliptic estimates.

One of the difficulties is that the estimate we obtain for $\nab^5\etab$ is weaker than that for $\nab^5\eta$. These terms would otherwise appear to be ``similar'' from the point of view of scaling and indeed obey the same lower order estimate. On the other hand, for the highest order bound, the derivation of the $L^2_uL^2(\S)$ estimate for $\nab^5\etab$ would give rise to a logarithmic loss. As a result, we give up the $L^2_uL^2(\S)$ estimate for $\nab^5\etab$ and content with an estimate in $L^2_{\ub}L^2(\S)$. (Notice that $\nab^5\eta$ obeys both $L^2_uL^2(\S)$ and $L^2_{\ub}L^2(\S)$ bounds.) It is a remarkable fact regarding the structure of the Einstein equations that while the stronger estimate for $\nab^5\eta$ is crucially used to obtain the bounds for $\nab^4\chih$ and $\nab^4\trch$, we can close the whole argument without a similar estimate for $\nab^5\etab$!

This concludes the discussions of the main ideas and difficulties of the proof.

\subsection{Outline of the paper}
We end the introduction by a brief outline of the remainder of the paper. We will describe the setting of the double null foliation gauge and explain the notations used in this paper in Section \ref{sec.setting}. This will allow us to state the main estimates that we will prove in Section \ref{sec.main.thm}. The main a priori estimates will then be proved in Sections \ref{secoutline}--\ref{seccurv}. We refer the readers to Section \ref{secoutline} for a more detailed outline of the proof in Sections \ref{secbasic}--\ref{seccurv}. Finally, in Section \ref{sec.formation}, we show that a trapped surface indeed forms given the bounds we have derived in the previous sections.

\subsection{Acknowledgements}
The authors thank Mihalis Dafermos, Sergiu Klainerman, Sung-Jin Oh, Igor Rodnianski and J\'er\'emie Szeftel for valuable discussions. We also thank Mihalis Dafermos and Sung-Jin Oh for helpful suggestions on an earlier version of the manuscript. Most of this work was carried out when both authors were at Princeton University. J. Luk is support by the NSF Postdoctoral Fellowship DMS-1204493.

\section{Setting, equations and notations} \label{sec.setting}

In this section, we will introduce the geometric setup and the double null foliation gauge. We then write the Einstein equations as a system of equations for the Ricci coefficients and curvature components adapted to this gauge. After that we introduce the necessary notations and the norms that we will use.

\subsection{Double null foliation}\label{secdnf}

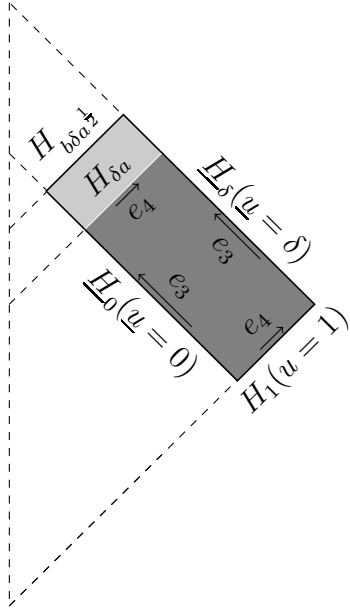
\begin{figure}
\centering
\begin{tikzpicture}

\draw [white](3,-1)-- node[midway, sloped, below,black]{$H_1(u=1)$}(4,0);

\draw [white](0.5,1.5)-- node[midway,sloped,above,black]{$H_{b \delta \at}$}(1.5,2.5);
\draw [white](2,2)--node [midway,sloped,above,black] {$\Hb_{\delta}(\ub=\delta)$}(4,0);
\draw [white](1,1)--node [midway,sloped, below,black] {$\Hb_{0}(\ub=0)$}(3,-1);
\draw [dashed] (0, 4)--(0, -4);
\draw [dashed] (0, -4)--(4,0)--(0,4);
\draw [dashed] (0,0)--(2,2);
\draw [dashed] (0,1)--(1.5,2.5);
\draw [dashed] (0,-4)--(2,-2);
\draw [dashed] (0,2)--(3,-1);
\draw [very thick] (1,1)--(3,-1)--(4,0)--(2,2)--(1,1);
\draw [very thick] (1,1)--(0.5,1.5)--(1.5,2.5)--(2,2)--(1,1);
\fill[black!50!white] (1,1)--(3,-1)--(4,0)--(2,2)--(1,1);
\fill[black!20!white](1,1)--(0.5,1.5)--(1.5,2.5)--(2,2)--(1,1);
\draw [white](1,1)-- node[midway,sloped,above,black]{$H_{\delta a}$}(2,2);
\draw [->] (3.3,-0.6)-- node[midway, sloped, above,black]{$e_4$}(3.6,-0.3);
\draw [->] (1.4,1.3)-- node[midway, sloped, below,black]{$e_4$}(1.7,1.6);
\draw [->] (3.3,0.6)-- node[midway, sloped, below,black]{$e_3$}(2.7,1.2);
\draw [->] (2.4,-0.3)-- node[midway, sloped, above,black]{$e_3$}(1.7,0.4);
\end{tikzpicture}
\caption{Basic Setup}
\end{figure}

\hspace{0.05\textwidth} 

Given a spacetime solution with initial data as in Theorem \ref{main.thm.intro}, we define a double null foliation by solving the eikonal equations
$$(g^{-1})^{\mu\nu}\partial_\mu u\partial_\nu u=0,\quad (g^{-1})^{\mu\nu}\partial_\mu\ub\partial_\nu \ub=0,$$
for $u$ and $\ub$ such that $u=1$ on $H_1$ and $\ub=0$ on $\Hb_0$. Note that $\ub$ is increasing towards the future
while $u$ is decreasing towards the future.

Let
$$L'^\mu=-2(g^{-1})^{\mu\nu}\partial_\nu u,\quad \Lb'^\mu=2(g^{-1})^{\mu\nu}\partial_\nu \ub.$$ 
be future directed, null geodesic vector fields and define
$$2\Omega^{-2}=-g(L',\Lb').$$
Let
$$e_3=\Omega\Lb', \quad e_4=\Omega L'$$
such that 
$$g(e_3,e_4)=-2.$$
These are the frames that we will use to decompose the Ricci coefficients and curvature components. Define also
$$\Lb=\Omega^2\Lb',\quad  L=\Omega^2 L'$$
to be the equivariant vector fields.

We fix the gauge on the initial hypersurfaces such that 
$$\Omega=1,\quad\mbox{on $H_1$ and $\Hb_0$}.$$

Let $H_u$ be the level sets of $u$ and $\Hb_{\ub}$ be the level sets of $\ub$. By the eikonal equations, $H_u$ and $\Hb_{\ub}$ are null hypersurface. The intersections of the hypersurfaces $H_u$ and $\Hb_{\ub}$ are topologically 2-spheres. We will denote them by $S_{u,\ub}$. 

\subsection{The coordinate system}\label{coordinates}

We define a coordinate system $(u,\ub,\theta^1,\theta^2)$ in the spacetime as follows:
On the standard sphere $S_{1,0}$, define a coordinate system $(\theta^1,\theta^2)$ such that on each coordinate patch the metric $\gamma$ is smooth, bounded and positive definite. We then define the coordinates on the initial hypersurfaces by requiring $\theta^A$ to 
be constant along null generators of the initial hypersurface. In the spacetime, we define $u$ and $\ub$ to be solutions to the eikonal equations as described in the previous subsection.
Moreover, define $\theta^1, \theta^2$ by
$$\Ls_L \theta^A=0,$$ 
where $\Ls_L$ denote the restriction of the Lie derivative to $TS_{u,\ub}$ (See \cite{Chr:book}, Chapter 1).
Relative to the coordinate system $(u,\ub,\theta^1,\theta^2)$, $e_3$ and $e_4$ can be expressed as
$$e_3=\Omega^{-1}\left(\frac{\partial}{\partial u}+d^A\frac{\partial}{\partial \theta^A}\right), \quad e_4=\Omega^{-1}\frac{\partial}{\partial \ub},$$
for some $d^A$ such that $d^A=0$ on $\Hb_0$ and the metric $g$ takes the form
$$g=-2\Omega^2(du\otimes d\ub+d\ub\otimes du)+\gamma_{AB}(d\theta^A-d^Adu)\otimes (d\theta^B-d^Bdu).$$

\subsection{Equations}\label{seceqn}
We decompose the Ricci coefficients and curvature components with respect to a null frame $e_3$, $e_4$ defined above and an frame ${e_1,e_2}$ tangent to the 2-spheres $S_{u,\ub}$. Using the indices $A,B\in\{1,2\}$, we define the Ricci coefficients relative to the null fame:
 \begin{equation}
\begin{split}
&\chi_{AB}=g(D_A e_4,e_B),\, \,\, \quad \chib_{AB}=g(D_A e_3,e_B),\\
&\eta_A=-\frac 12 g(D_3 e_A,e_4),\quad \etab_A=-\frac 12 g(D_4 e_A,e_3),\\
&\omega=-\frac 14 g(D_4 e_3,e_4),\quad\,\,\, \omegab=-\frac 14 g(D_3 e_4,e_3),\\
&\zeta_A=\frac 1 2 g(D_A e_4,e_3)
\end{split}
\end{equation}
where $D_A=D_{e_{(A)}}$; and also the  null curvature components,
 \begin{equation}
\begin{split}
\a_{AB}&=R(e_A, e_4, e_B, e_4),\quad \, \,\,   \ab_{AB}=R(e_A, e_3, e_B, e_3),\\
\beta_A&= \frac 1 2 R(e_A,  e_4, e_3, e_4) ,\quad \bb_A =\frac 1 2 R(e_A,  e_3,  e_3, e_4),\\
\rho&=\frac 1 4 R(e_4,e_3, e_4,  e_3),\quad \sigma=\frac 1 4  \,^*R(e_4,e_3, e_4,  e_3).
\end{split}
\end{equation}
Here $\, ^*R$ denotes the Hodge dual of $R$.  Let $\nab$ be the 
induced covariant derivative operator on $S_{u,\ub}$ and $\nab_3$, $\nab_4$ be
the projections to $S_{u,\ub}$ of the covariant derivatives $D_3$, $D_4$ (see
precise definitions in \cite{KNI:book}). 

Notice that,
\begin{equation}
\begin{split}
&\omega=-\frac 12 \nab_4 (\log\Omega),\qquad \omegab=-\frac 12 \nab_3 (\log\Omega),\\
&\eta_A=\zeta_A +\nab_A (\log\Omega),\quad \etab_A=-\zeta_A+\nab_A (\log\Omega).
\end{split}
\end{equation}

Define the following contractions of the tensor product $\phi^{(1)}$ and $\phi^{(2)}$ with respect to the metric $\gamma$. For two symmetric $2$-tensors $\phi^{(1)}_{AB}$, $\phi^{(2)}_{AB}$, define
\begin{align*}
\phi^{(1)}\cdot\phi^{(2)} &:= (\gamma^{-1})^{AC}(\gamma^{-1})^{BD}\phi^{(1)}_{AB}\phi^{(2)}_{CD},\\
\phi^{(1)}\wedge\phi^{(2)} &:= \eps^{AB}(\gamma^{-1})^{CD}\phi^{(1)}_{AC}\phi^{(2)}_{BD},
\end{align*}
where $\eps$ is the volume form associated to the metric $\gamma$.
For two $1$-forms $\phi^{(1)}_{A}$, $\phi^{(2)}_{A}$, define
\begin{align*}
\phi^{(1)}\cdot\phi^{(2)}& :=(\gamma^{-1})^{AB}\phi^{(1)}_{A}\phi^{(2)}_{B} ,\\
\phi^{(1)}\wedge \phi^{(2)} & := \eps^{AB} \phi^{(1)}_A \phi^{(2)}_B, \\
(\phi^{(1)}\hot\phi^{(2)})_{AB}&:=\phi^{(1)}_A\phi^{(2)}_B+\phi^{(1)}_B\phi^{(2)}_A-\gamma_{AB}(\phi^{(1)}\cdot\phi^{(2)}).
\end{align*}
For a symmetric $2$-tensor $\phi^{(1)}_{AB}$ and a $1$-form $\phi^{(2)}_{A}$, define
$$(\phi^{(1)}\cdot\phi^{(2)})_A:=(\gamma^{-1})^{BC}\phi^{(1)}_{AB}\phi^{(2)}_{C}.$$
We also define by $^*$ for $1$-forms and symmetric $2$-tensors respectively as follows (note that on $1$-forms this is the Hodge dual on $S_{u,\ub}$):
\begin{align*}
^*\phi_A := & \gamma_{AC} \eps^{CB} \phi_B, \\
^*\phi_{AB} := & \gamma_{BD} \eps^{DC} \phi_{AC}.
\end{align*}
For totally symmetric tensors, the $\div$ and $\curl$ operators are defined by the formulas
\begin{align*}
(\div\phi)_{A_1\cdots A_r}&:=\nabla^B\phi_{BA_1\cdots A_r},\\
(\curl\phi)_{A_1\cdots A_r}&:=\eps^{BC}\nabla_B\phi_{CA_1\cdots A_r}.
\end{align*}
Define the operator $\nab\widehat{\otimes}$ on a $1$-form $\phi_{A}$ by
$$(\nab\widehat{\otimes}\phi)_{AB} :=  \nab_A \phi_B + \nab_B \phi_A - \gamma_{AB} \div \phi.$$
Also, define the trace to be
$$(\mbox{tr}\phi)_{A_1\cdots A_{r-1}}:=(\gamma^{-1})^{BC}\phi_{BCA_1\cdots A_{r-1}}.$$
Let $\chih$ and $\chibh$ be the traceless parts of $\chi$ and $\chib$ respectively. Then $\chi$ and $\chib$ satisfy the following null structure equations:
\begin{equation}
\label{null.str1}
\begin{split}
\nab_4 \trch+\frac 12 (\trch)^2&=-|\chih|^2-2\omega \trch,\\
\nab_4\chih+\trch \chih&=-2 \omega \chih-\alpha,\\
\nab_3 \trchb+\frac 12 (\trchb)^2&=-2\omegab \trchb-|\chibh|^2,\\
\nab_3\chibh + \trchb\,  \chibh&= -2\omegab \chibh -\alphab,\\
\nab_4 \trchb+\frac1 2 \trch \trchb &=2\omega \trchb +2\rho- \chih\cdot\chibh +2\div \etab +2|\etab|^2,\\
\nab_4\chibh +\frac 1 2 \trch \chibh&=\nab\widehat{\otimes} \etab+2\omega \chibh-\frac 12 \trchb \chih +\etab\widehat{\otimes} \etab,\\
\nab_3 \trch+\frac1 2 \trchb \trch &=2\omegab \trch+2\rho- \chih\cdot\chibh+2\div \eta+2|\eta|^2,\\
\nab_3\chih+\frac 1 2 \trchb \chih&=\nab\widehat{\otimes} \eta+2\omegab \chih-\frac 12 \trch \chibh +\eta\widehat{\otimes} \eta.
\end{split}
\end{equation}
The remaining Ricci coefficients satisfy the following null structure equations:
\begin{equation}
\label{null.str2}
\begin{split}
\nabla_4\eta&=-\chi\cdot(\eta-\etab)-\beta,\\
\nabla_3\etab &=-\chib\cdot (\etab-\eta)+\bb,\\
\nabla_4\omegab&=2\omega\omegab-\eta\cdot\etb+\f12|\eta|^2+\frac 12 \rho,\\
\nabla_3\omega&=2\omega\omegab-\eta\cdot\etb+\f12|\eta|^2+\frac 12 \rho.
\end{split}
\end{equation}
The Ricci coefficients also satisfy the following constraint equations:
\begin{equation}
\label{null.str3}
\begin{split}
\div\chih&=\frac 12 \nabla \trch - \frac 12 (\eta-\etab)\cdot (\chih -\frac 1 2 \trch) -\beta,\\
\div\chibh&=\frac 12 \nabla \trchb + \frac 12 (\eta-\etab)\cdot (\chibh-\frac 1 2   \trchb) +\betab,\\
\curl\eta &=-\curl\etab=\sigma +\frac 1 2\chibh \wedge\chih,\\
K&=-\rho+\frac 1 2 \chih\cdot\chibh-\frac 1 4 \trch \trchb,
\end{split}
\end{equation}
with $K$ the Gauss curvature of the spheres $S_{u,\ub}$.
The curvature components verify the following null Bianchi equations:
\begin{equation}
\label{eq:null.Bianchi}
\begin{split}
&\nab_3\alpha+\frac 12 \trchb \alpha=\nabla\hot \beta+ 4\omegab\alpha-3(\chih\rho+^*\chih\sigma)+
(\zeta+4\eta)\hot\beta,\\
&\nab_4\beta+2\trch\beta = \div\alpha - 2\omega\beta +  (2\zeta+\etab)\cdot \alpha,\\
&\nab_3\beta+\trchb\beta=\nabla\rho + 2\omegab \beta +^*\nabla\sigma +2\chih\cdot\betab+3(\eta\rho+^*\eta\sigma),\\
&\nab_4\sigma+\frac 32\trch\sigma=-\div^*\beta+\frac 12\chibh \wedge \alpha-\zeta\wedge\beta-2\etab\wedge\beta,\\
&\nab_3\sigma+\frac 32\trchb\sigma=-\div ^*\betab-\frac 12\chih \wedge \alphab+\zeta\wedge\betab-2\eta\wedge\betab,\\
&\nab_4\rho+\frac 32\trch\rho=\div\beta-\frac 12\chibh\cdot\alpha+\zeta\cdot\beta+2\etab\cdot\beta,\\
&\nab_3\rho+\frac 32\trchb\rho=-\div\betab- \frac 12\chih\cdot\alphab+\zeta\cdot\betab-2\eta\cdot\betab,\\
&\nab_4\betab+\trch\betab=-\nabla\rho +^*\nabla\sigma+ 2\omega\betab +2\chibh\cdot\beta-3(\etab\rho-^*\etab\sigma),\\
&\nab_3\betab+2\trchb\betab=-\div\alphab-2\omegab\betab-(-2\zeta+\eta) \cdot\alphab,\\
&\nab_4\alphab+\frac 12 \trch\alphab=-\nabla\hot \betab+ 4\omega\alphab-3(\chibh\rho-^*\chibh\sigma)+
(\zeta-4\etab)\hot \betab.
\end{split}
\end{equation}
Defining
$$\sigmac=\sigma+\frac 12 \chibh\wedge\chih,$$
the Bianchi equations can be expressed in terms of $K$ and $\sigmac$ instead of $\rho$ and $\sigma$ are as follows:
\begin{equation}
\label{eq:null.Bianchi2}
\begin{split}
\nab_3\beta+\trchb\beta&=-\nabla K  +^*\nabla\sigmac + 2\omegab \beta+2\chih\cdot\betab-3(\eta K-^*\eta\sigmac)\\
&\quad +\frac 1 2(\nabla(\chih\cdot\chibh)+^*\nabla(\chih\wedge\chibh))-\frac 34 \eta\trch\trchb\\
&\quad +\frac 3 2(\eta\chih\cdot\chibh+^*\eta\chih\wedge\chibh)-\frac 14 (\nab\trch \trchb+\trch\nab\trchb),\\
\nab_4\sigmac+\frac 32\trch\sigmac &= -\div^*\beta-\zeta\wedge\beta-2\etab\wedge
\beta-\frac 12 \chih\wedge(\nab\widehat{\otimes}\etab)\\
&\quad -\frac 12 \chih\wedge(\etab\widehat{\otimes}\etab),\\
\nab_4 K+\trch K &= -\div\beta-\zeta\cdot\beta-2\etab\cdot\beta+\frac 12 \chih\cdot\nab\widehat{\otimes}\etab+\frac 12 \chih\cdot(\etab\widehat{\otimes}\etab)\\
&\quad -\frac 12 \trch\div\etab-\frac 12\trch |\etab|^2,\\
\nab_3\sigmac+\frac 32\trchb\sigmac &= -\div ^*\betab+\zeta\wedge\betab-2\eta\wedge
\betab+\frac 12 \chibh\wedge(\nab\widehat{\otimes}\eta)\\
&\quad +\frac 12 \chibh\wedge(\eta\widehat{\otimes}\eta),\\
\nab_3 K+\trchb K&=  \div\betab-\zeta\cdot\betab+2\eta\cdot\betab+\frac 12 \chibh\cdot\nab\widehat{\otimes}\eta+\frac 12 \chibh\cdot(\eta\widehat{\otimes}\eta)\\
&\quad -\frac 12 \trchb\div\eta-\frac 12 \trchb |\eta|^2,\\
\nab_4\betab+\trch\betab&=\nabla K +^*\nabla\sigmac+ 2\omega\betab +2\chibh\cdot\beta+3(-\etab K+^*\etab\sigmac)\\
&\quad -\frac 1 2(\nabla(\chih\cdot\chibh)-^*\nabla(\chih\wedge\chibh))\\
&\quad +\frac 14 (\nab\trch \trchb+\trch\nab\trchb)\\
&\quad -\frac32 (\etb\chih\cdot\chibh-^*\etb\chih\wedge\chibh)+\frac 34 \etab\trch\trchb.
\end{split}
\end{equation}

In the remainder of the paper, we will use the convention that capital Latin letters $A\in \{1,2\}$ are used as indices on the spheres $S_{u,\ub}$ and while Greek letters $\mu\in\{1,2,3,4\}$ are used as indices in the whole spacetime.

\subsection{Schematic notation}\label{secsche}
We introduce a schematic notation as follow: Let $\phi$ denote an arbitrary tensorfield. For the Ricci coefficients, we use the notation
\begin{equation}\label{schpsi}
\psi\in\{\chih,\tr\chi,\omega\},\quad \q\in\left\{\eta,\etb, \chibh, \tr\chib+\frac{2}{u}, \omb \right\}.
\end{equation}
The set of all Ricci coefficients can therefore be represented by either $\p,\q$ or $\tr \chib$.

We will simply write $\psi\psi$ (or $\psi\q$, $\psi\beta$, etc.) to denote arbitrary contractions with respect to the metric $\gamma$. $\nab$ will be used to denote an arbitrary angular covariant derivative. We will only use the schematic notation when the precise nature of the contraction is not important to the argument. Moreover, under this schematic notation, all constant factors will be neglect.

When writing an equation, we use the following convention. On the left hand side of the equations, all of the terms are written with exact coefficients, while on the right hand side of the equations, terms are only written schematically. In particular, as mentioned above, we will neglect constant factors.

Another convention we introduce is that brackets are used to denote terms with any one of the components in the brackets. For example, $\psi(\p,\q)$ is used to denote either $\p\p$ or $\p\q$.

Finally, $\nab^i\q^j$ will be used to denote angular derivatives of products of $\q$. More precisely, $\nab^i\q^j$ denotes the sum of all terms which are products of $j$ factors, where each factor is $\nab^{i_k}\q$ and that the sum of all $i_k$'s is $i$, i.e., 
$$\nab^i\q^j=\displaystyle\sum_{i_1+i_2+\cdots +i_j=i}\underbrace{\nab^{i_1}\q\nab^{i_2}\q\cdots \nab^{i_j}\q}_\text{j factors}.$$

\subsection{Integration}

Given a function $\phi$, the integration on $S_{u,\ub}$, i.e., $\int_{S_{u,\ub}} \phi$, is defined with respect to the volume form induced by $\gamma$. The spacetime integration is defined with respect to the volume form induced by the spacetime metric $g$.
Since there are no canonical volume forms on $H_u$ and $\Hb_{\ub}$, we define integration by
$$\int_{H_{u}} \phi :=2\int_0^{\ub} \left(\int_{S_{u,\ub'}}\Omega\phi\right) d\ub' $$
and
$$\int_{H_{\ub}} \phi :=2\int_{u}^{1} \left(\int_{S_{u',\ub}}\Omega\phi\right) du'.$$
Likewise, the norms $L^p(\S)$, $L^p(H_u)$ and $L^p(\Hb_{\ub})$ are defined using the volume forms above.

We will also use mixed norms defined by
$$\|\phi\|_{L^p_{\ub}L^q_u L^r(S)}=\left(\int_0^{\ub}\left(\int_u^1\|\phi\|_{L^r(S_{u',\ub'})}^q  du'\right)^{\f pq} d\ub'\right)^{\f1p},$$
$$\|\phi\|_{L^p_{u}L^q_{\ub} L^r(S)}=\left(\int_u^1\left(\int_0^{\ub}\|\nabla^i\phi\|_{L^r(S_{u',\ub'})}^q d\ub'\right)^{\f pq}du'\right)^{\frac 1p}.$$
with appropriate modifications if $p=\infty$ or $q=\infty$. Notice that it is implicit that the $L^p$ norms are taken over the spacetime region given in coordinates by $\{(u',\ub',\theta^1,\theta^2):u\leq u'\leq 1,0\leq \ub'\leq \ub\}$. In particular, the size of these norms can depend on $u$ and $\ub$.

With the above definition, $\|\phi\|_{L^2_uL^2(S_{u,\ub})}$ and $\|\phi\|_{L^2(\Hb_{\ub})}$ (similarly for $\|\phi\|_{L^2_{\ub}L^2(S_{u,\ub})}$ and $\|\phi\|_{L^2(H_u)}$) differ by a factor of $\Omega$.  Nevertheless, in view of Proposition \ref{Omega} below, these norms are equivalent up to a factor of $2$.

\subsection{Norms}
We now define the norms that we will work with. First we define the norms for the curvature components:
\begin{equation*}
\begin{split}
\mathcal R&=\sum_{i\leq 4}\Bigg(\sup_{u}\bigg(\frac{1}{\delta^{\frac12}\at}\|u^{i+1}\nabla^i\beta\|_{L^2(H_u)}\bigg)\\
&\qquad\quad + \sup_{\ub}\bigg(\frac{1}{\delta^{\frac12}\at}\left\|u^{i+1}\nabla^i\left(K-\f{1}{|u|^2},\sigmac\right)\right\|_{L^2(\Hb_{\ub})}\bigg)\Bigg)\\
&\quad +\sum_{1\leq i\leq 4}\Bigg(\sup_{u}\bigg(\frac{1}{\delta^{\frac32}a^{\f34}}\left\|u^{i+2}\nabla^i\left(K-\f{1}{|u|^2},\sigmac\right)\right\|_{L^2(H_u)}\bigg)\\
&\qquad\qquad\quad +\sup_{\ub}\bigg(\frac{1}{\delta^{\frac32}a^{\f34}}\|u^{i+2}\nabla^i\beb\|_{L^2(\Hb_{\ub})}\bigg)\Bigg)\\
&\quad +\sup_{u}\bigg(\frac{1}{\delta^{\frac32}a^{\f34}}\left\|u^2\left(K-\f{1}{|u|^2}\right)\right\|_{L^2(H_u)}\bigg).
\end{split}
\end{equation*}
We then define the norms for the Ricci coefficients. We begin with those for the highest order derivatives:
\begin{equation*}
\begin{split}
\tilde{\mathcal O}_{5,2}&=\sup_{u}\bigg(\frac{1}{\delta^{\frac12}\at}\|u^{5}\nabla^{5}(\chih,\tr\chi,\omega)\|_{L^2(H_u)}\bigg)+ \sup_{\ub}\bigg(\frac{1}{\delta^{\frac12}\at}\|u^{5}\nabla^{5}\eta\|_{L^2(\Hb_{\ub})}\bigg)\\
&\quad +\sup_{u}\bigg(\frac{1}{\delta^{\frac32}a^{\f34}}\|u^{6}\nabla^{5}(\eta,\etb)\|_{L^2(H_u)}\bigg)\\
&\quad +\sup_{u,\ub}\bigg(\frac{|u|^{\f12}}{\delta a^{\f12}}\|u^{5}\nabla^{5}(\trchb,\chibh,\omb)\|_{L^2(\Hb_{\ub})}\bigg).
\end{split}
\end{equation*}
For $i\leq 4$, we define the following $L^2$ norms:
\begin{equation*}
\begin{split}
{\mathcal O}_{i,2}=\sup_{u,\ub}& \bigg ( \f{1}{\at}\|u^i\nab^i(\chih,\omega, \tr\chi)\|_{L^{2}(\S)}\\
&+\frac{|u|}{\delta \at}\left\|u^i\nab^i\left(\eta,\etb,\nab\log\Omega, \chibh, \tr\chib+\f {2}{u},\omb\right)\right\|_{L^{2}(\S)} \bigg ),
\end{split}
\end{equation*}
and for $i\leq 2$, we define the following $L^\infty$ norms:
\begin{equation*}
\begin{split}
{\mathcal O}_{i,\infty}= \sup_{u,\ub}&\bigg(\frac{1}{\at}|u|\|u^i\nab^i(\chih,\omega, \tr\chi)\|_{L^{\infty}(\S)}\\
&+\frac{|u|^2}{\delta\at}\left\|u^i\nab^i\left(\eta,\etb,\chibh, \tr\chib+\f{2}{u},\omb\right)\right\|_{L^{\infty}(\S)}\bigg).
\end{split}
\end{equation*}
As a shorthand, we will also denote 
$$\mathcal O=\sum_{i\leq 2}\mathcal O_{i,\infty}+\sum_{i\leq 4}\mathcal O_{i,2}.$$

\section{Quantitative statement of main theorem}\label{sec.main.thm}

We now state the a priori estimates that we will prove in this paper. The existence result in Theorem \ref{main.thm.intro} follows from the a priori estimates using standard arguments (see, for example, \cite{Chr:book}). We will omit the details.

\begin{theorem} \label{main.thm}
Consider the following characteristic initial value problem for the Einstein vacuum equations. The initial incoming hypersurface $\Hb_0$ is required to coincide with a backwards light cone in Minkowski space with $0\leq u \leq 1$. On the initial outgoing hypersurface $H_1$, the data are smooth and the initial shear satisfies
$$\sum_{i\leq 7}\|\nab^i\chih_0\|_{L^\infty_{\ub}L^2(\S)}\leq a^{\frac 12} $$
for $0\leq \ub \leq \delta$.

Then there exists a universal large constant $b_0$ such that if $b_0\leq b\leq a$ and $\de \at b< 1$, the unique solution to the Einstein vacuum equations obeys the following estimates in the region $\delta \at b\leq u \leq 1$, $ 0 \leq \ub \leq \delta$:
$$\mathcal O, \tilde{\M O}_{5,2}, \M R\ls 1,$$
where the implicit constant is universal and independent of $a$, $b$ and $\delta$.
\end{theorem}

\begin{remark}
Following \cite{Chr:book}, one can solve the constraint ODEs and obtain bounds for the initial data on $H_1$ from that of the initial shear. In particular, under the assumption of Theorem \ref{main.thm}, we have the following initial bounds for the Ricci coefficients
\begin{equation*}
\begin{split}
&\sum_{i\leq 5}\bigg(\f{1}{\at}\|\nab^i(\chih,\omega, \tr\chi)\|_{L^{2}(S_{1,\ub})}+\frac{1}{\delta \at}\|\nab^i(\eta,\etb,\chibh, \tr\chib+2,\omb)\|_{L^{2}(S_{1,\ub})}\bigg)\ls 1,
\end{split}
\end{equation*}
and the following initial bounds for the curvature components
\begin{equation*}
\begin{split}
\sum_{i\leq 4}\bigg(\f{1}{\at}\|\nab^i\beta\|_{L^\infty_{\ub} L^2(S_{1,\ub})}+\frac{1}{\delta\at}\|\nab^i(K-1, \sigmac, \beb)\|_{L^\infty_{\ub}L^2(S_{1,\ub})}\bigg)\ls 1.
\end{split}
\end{equation*}
\end{remark}

Once the existence theorem is established, the actual formation of\linebreak trapped surfaces follows from a simple ODE argument as in \cite{Chr:book}:

\begin{theorem}\label{trapped.thm}
If, moreover,
the data on $H_1$ obey
\begin{equation}\label{main.lower.bd}
\int_0^{\delta}|\chih_0|^2(\ub') d\ub'\geq 4b \delta\at
\end{equation}
along outgoing characteristics in every direction $\vartheta$, then the $2$-sphere defined by
$S_{b\delta\at,\delta}:=\{(u,\ub,\theta^1,\theta^2): u=b\delta\at\,,\ub=\delta\}$ is a trapped surface. 
\end{theorem}

Theorems \ref{main.thm} and \ref{trapped.thm} together imply Theorem~\ref{main.thm.intro}. The proof of Theorem~\ref{main.thm} will take up most of the remainder of the paper in Sections \ref{secoutline}--\ref{seccurv}. We will then turn to the proof of Theorem \ref{trapped.thm} in the final section.

\section{Bootstrap assumptions and the outline of the proof of Theorem \ref{main.thm}}\label{secoutline}

We now introduce the bootstrap assumptions that we use to prove Theorem \ref{main.thm}. We make the following bootstrap assumptions on the first four derivatives of the Ricci coefficients:
\begin{equation}\label{BA.1}
\sum_{i\leq 4}\frac{1}{\delta a^{\frac 12}}\|u^{i+1}\nab^i\q\|_{L^2(S_{u,\ub})}+\sum_{i\leq 2}\frac{1}{\delta a^{\frac 12}}\|u^{i+2}\nab^i\q\|_{L^{\infty}(S_{u,\ub})} \leq b^{\frac 14}
\end{equation}
and
\begin{equation}\label{BA.2}
\sum_{i\leq 4}\frac{1}{a^{\frac 12}}\|u^{i}\nab^i\p\|_{L^2(S_{u,\ub})}+\sum_{i\leq 2}\frac{1}{a^{\frac 12}}\|u^{i+1}\nab^i\p\|_{L^{\infty}(S_{u,\ub})} \leq b^{\frac 14}.
\end{equation}
We will also make the bootstrap assumptions on $\tilde{\mathcal O}_{5,2}$ and $\mathcal R$:
\begin{equation}\label{BA.3}
\tilde{\mathcal O}_{5,2}+\mathcal R\leq b^{\frac 14}.
\end{equation}
Finally, we make a bootstrap assumption on $K$ and its derivatives, which will be useful for elliptic estimates:
\begin{equation}\label{BA.4}
\sum_{i\leq 3}\left\|u^{i+1}\nab^i\left(K-\f{1}{|u|^2}\right)\right\|_{L^\infty_u L^\infty_{\ub} L^2(\S)}\leq 1.
\end{equation}

In the following sections, we will show that given the bootstrap assumptions \eqref{BA.1}-\eqref{BA.4}, the bounds in fact hold with a better constant. More precisely, we will show
\begin{gather}\label{BA.r1}
\sum_{i\leq 4}\frac{1}{\delta a^{\frac 12}}\|u^{i+1}\nab^i\q\|_{L^2(S_{u,\ub})}+\sum_{i\leq 2}\frac{1}{\delta a^{\frac 12}}\|u^{i+2}\nab^i\q\|_{L^{\infty}(S_{u,\ub})} \ls 1,\\
\label{BA.r2}
\sum_{i\leq 4}\frac{1}{a^{\frac 12}}\|u^{i}\nab^i\p\|_{L^2(S_{u,\ub})}+\sum_{i\leq 2}\frac{1}{a^{\frac 12}}\|u^{i+1}\nab^i\p\|_{L^{\infty}(S_{u,\ub})} \leq 1,\\
\label{BA.r3}
\tilde{\mathcal O}_{5,2}+\mathcal R\ls 1 
\end{gather}
and
\begin{equation}\label{BA.r4}
\sum_{i\leq 3}\left\|u^{i+1}\nab^i\left(K-\f{1}{|u|^2}\right)\right\|_{L^\infty_u L^\infty_{\ub} L^2(\S)}\ls \f{1}{b^{\f34}}.
\end{equation}
{\bf Here, and in the rest of the paper, we use the convention that $A\ls B$ denotes the inequality $A\leq C B$ for some universal constant $C$ that is independent of $\de$, $a$ and $b$.} The bounds that we derive will therefore improve over the bootstrap assumptions \eqref{BA.1}, \eqref{BA.2}, \eqref{BA.3} and \eqref{BA.4} after choosing $b_0$ to be sufficiently large.

We now give a brief outline of the proof of Theorem \ref{main.thm} in Sections \ref{secbasic}--\ref{seccurv}.
\begin{itemize}
\item In Section \ref{secbasic}, we prove some preliminary estimates. These include the bounds for the metric components $\gamma$ and $\Omega$. A particular consequence of these bounds is a Sobolev embedding theorem on the $2$-spheres $\S$. We also derive propositions from which we can obtain bounds from general transport equations and elliptic systems.
\item In Section \ref{secRicci}, we use transport equations for the Ricci coefficients to obtain
\begin{equation*}
\begin{split}
\sum_{i\leq 4}\frac{1}{\delta a^{\frac 12}}\|u^{i+1}\nab^i\q\|_{L^2(S_{u,\ub})}+\sum_{i\leq 2}\frac{1}{\delta a^{\frac 12}}\|u^{i+2}\nab^i\q\|_{L^{\infty}(S_{u,\ub})}
 \ls 1+\tilde{\mathcal O}_{5,2}+\mathcal R
\end{split}
\end{equation*}
and
$$\sum_{i\leq 4}\frac{1}{a^{\frac 12}}\|u^{i}\nab^i\p\|_{L^2(S_{u,\ub})}+\sum_{i\leq 2}\frac{1}{a^{\frac 12}}\|u^{i+1}\nab^i\p\|_{L^{\infty}(S_{u,\ub})} \leq 1.$$
In other words, we obtain estimates that would imply \eqref{BA.r1} and \eqref{BA.r2} once \eqref{BA.r3} is also proved. In the same section, we also prove some more refined estimates for some of the Ricci coefficients. Moreover, we obtain the bound \eqref{BA.r4}.
\item In Section \ref{secelliptic}, we use elliptic estimates to show that 
$$\tilde{\mathcal O}_{5,2}\ls 1+\mathcal R.$$
\item In Section \ref{seccurv}, we use energy estimates to prove that
$$\mathcal R\ls 1.$$
Combining this with the bounds above, we will have thus obtained \eqref{BA.r1}--\eqref{BA.r4} as desired.

\end{itemize}

\section{The preliminary estimates}\label{secbasic}

\subsection{Estimates for metric components}\label{metric}
We first show that we can control $\Omega$ under the bootstrap assumptions:
\begin{proposition}\label{Omega}
Under the assumptions of Theorem \ref{main.thm} and the bootstrap assumptions \eqref{BA.1}, \eqref{BA.2}, \eqref{BA.3} and \eqref{BA.4}, we have
$$\|\Omega^{-1}-1\|_{L^\infty(\S)}\ls \f{\de\at b^{\f14}}{|u|}.$$
\end{proposition}
\begin{proof}
Consider the equation
\begin{equation}\label{Omegatransport}
 \omega=-\frac{1}{2}\nabla_4\log\Omega=\frac{1}{2}\Omega\nabla_4\Omega^{-1}=\frac{1}{2}\frac{\partial}{\partial \ub}\Omega^{-1}.
\end{equation}
Fix $\ub$. Notice that both $\omega$ and $\Omega$ are scalars and therefore the $L^\infty$ norm is independent of the metric. We can integrate equation (\ref{Omegatransport}) using the fact that $\Omega^{-1}=1$ on $\Hb_0$ to obtain
$$\|\Omega^{-1}-1\|_{L^\infty(S_{u,\ub})}\ls \int_0^{\ub}\|\omega\|_{L^\infty(S_{u,\ub'})}d\ub'\ls \frac{\delta \at b^{\f14}}{|u|},$$
where we have used the bootstrap assumption \eqref{BA.2}.
\end{proof}

We then show that we can control $\gamma$ under the bootstrap assumptions. This follows from an argument similar to that in \cite{Chr:book}.
\begin{proposition}\label{gamma}
We continue to work under the assumptions of Theorem~\ref{main.thm} and the bootstrap assumptions \eqref{BA.1}, \eqref{BA.2}, \eqref{BA.3} and \eqref{BA.4}. Fix a point $(u,\vartheta)$ on the initial hypersurface $\Hb_0$. Along the outgoing characteristic emanating from $(u,\vartheta)$, define $\Lambda(\ub)$ and $\lambda(\ub)$ to be the larger and smaller eigenvalue of $\gamma^{-1}(u,\ub=0,\vartheta)\gamma(u,\ub,\vartheta)$. Then
$$|\Lambda(\ub)-1|+|\lambda(\ub)-1|\ls \f{\de\at b^{\f14}}{|u|}$$
for every $\ub\in[0,\de]$. As a consequence, we also have
$$|\xi(\ub)-1|\ls \f{\de\at b^{\f14}}{|u|}$$
for every $\ub\in[0,\de]$, where
$$\xi(\ub)=\f{dvol_{\gamma_{\ub}}}{dvol_{\gamma_0}}.$$
\end{proposition}
\begin{proof}
The first variation formula states that
$$\Ls_L\gamma=2\Omega\chi,$$
which translates to
\begin{equation}\label{1st.var}
\frac{\partial}{\partial \ub}\gamma_{AB}=2\Omega\chi_{AB}
\end{equation}
in coordinates. From this we derive that 
$$\frac{\partial}{\partial \ub}\log \xi=\Omega\trch.$$
Since by definition $\xi(0)=1$, we have
\begin{equation}\label{0.mu.est}
|\xi(\ub)-1|\ls \f{\de\at b^{\f14}}{|u|},
\end{equation}
using the bound for $\trch$ in the bootstrap assumption \eqref{BA.2} and the estimate for $\Omega$ from Proposition \ref{Omega}.

Define $\nu(\ub)=\f{\Lambda(\ub)}{\xi(\ub)}=\sqrt{\f{\Lambda(\ub)}{\lambda(\ub)}}$.
Following the derivation of (5.93) in \cite{Chr:book}, we can use \eqref{1st.var} to derive the estimate
$$\nu(\ub)\leq 1+\int_0^{\ub} |(\Omega\chih)(\ub')|_{\gamma} \nu(\ub') d\ub'.$$
This implies via Gronwall's inequality that 
\begin{equation}\label{0.nu.est}
|\nu(\ub)-1|\ls \f{\de\at b^{\f14}}{|u|}.
\end{equation}
The desired conclusion follows from \eqref{0.mu.est} and \eqref{0.nu.est}.
\end{proof}

A direct consequence of the previous proposition is an estimate on the surface area of the two sphere $S_{u,\ub}$.
\begin{proposition}\label{area}
Under the assumptions of Theorem \ref{main.thm} and the bootstrap assumptions \eqref{BA.1}, \eqref{BA.2}, \eqref{BA.3} and \eqref{BA.4}, we have
$$\sup_{u,\ub}|\mbox{Area}(S_{u,\ub})-\mbox{Area}(S_{u,0})|\ls \f{\delta\at b^{\f14}}{|u|}.$$
\end{proposition}

\subsection{Estimates for transport equations}\label{transportsec}

In latter sections of the paper, we will derive the estimates for the Ricci coefficients and the null curvature components from the null structure equations and the null Bianchi equations respectively. These will be viewed as transport equations and we will need a way to obtain estimates from the covariant null transport equations. For the transport equation in the $e_4$ direction, we will need the smallness of $\|\trch\|_{L^\infty_u L^1_{\ub}L^\infty(\S)}$, which is a consequence of our bootstrap assumption \eqref{BA.2}. More precisely, we have

\begin{proposition}\label{transport}
Under the assumptions of Theorem \ref{main.thm} and the bootstrap assumptions \eqref{BA.1}, \eqref{BA.2}, \eqref{BA.3} and \eqref{BA.4}, we have
\[
 \|\phi\|_{L^2(S_{u,\ub})}\ls \|\phi\|_{L^2(S_{u,\ub'})}+\int_{\ub'}^{\ub} \|\nabla_4\phi\|_{L^2(S_{u,\ub''})}d{\ub''}
\]
for an $S_{u,\ub}$ tangent tensor $\phi$ of arbitrary rank.
\end{proposition}

\begin{proof}
We first note that the following identity holds for any scalar $f$:
\[
 \frac{d}{d\ub}\int_{\S} f=\int_{\S} \left(\frac{df}{d\ub}+\Omega \trch f\right)=\int_{\S} \Omega\left(e_4(f)+ \trch f\right).
\]
Hence, taking $f=|\phi|_{\gamma}^2$, we have
\begin{equation}\label{Lptransport}
\begin{split}
 \|\phi\|^2_{L^2(S_{u,\ub})}=\|\phi\|^2_{L^2(S_{u,\ub'})}+\int_{\ub'}^{\ub}\int_{S_{u,\ub''}} 2\Omega\left(<\phi,\nabla_4\phi>_\gamma+ \frac{1}{2}\trch |\phi|^2_{\gamma}\right)d{\ub''}.
\end{split}
\end{equation}
The proposition can be concluded using the Cauchy-Schwarz inequality on the sphere and the $L^\infty$ bounds for $\Omega$ and $\trch$  which are provided by Proposition \ref{Omega} and the bootstrap assumption \eqref{BA.2} respectively.
\end{proof}
On the other hand, in order to use the $\nab_3$ equation, we need to incorporate the weights in the norms. These weights depend on the coefficients in front of the linear term with a $\trchb$ factor. The main observation is that under the bootstrap assumption \eqref{BA.1}, $\trchb$ can be viewed essentially as $-\f2{|u|}$. More precisely, we have
\begin{proposition}\label{evolution lemma}
We continue to work under the assumptions of Theorem~\ref{main.thm} and the bootstrap assumptions \eqref{BA.1}, \eqref{BA.2}, \eqref{BA.3} and \eqref{BA.4}. Let $\phi$ and $F$ be $\S$-tangent tensor fields of rank $k$ satisfying the following transport equation:
\begin{equation*}
\nab_3 \phi_{A_1\cdots A_k}+\lambda_0{\tr\underline{\chi}}\phi_{A_1\cdots A_k}=F_{A_1\cdots A_k},
\end{equation*}
Denoting $\lambda_1=2(\lambda_0-\frac{1}{2})$, we have
\begin{equation*}
|u|^{\lambda_1}\|\phi\|_{L^2(\S)}\ls
\|\phi\|_{L^2(S_{1,\underline{u}})}+\int_u^1|u'|^{\lambda_1}\|F\|_{L^2(S_{u',\underline{u}})}du',
\end{equation*}
where the implicit constant is allowed to depend on $\lambda_0$.
\end{proposition}
\begin{proof}
To begin, we have the following identity for any scalar function $f$:
\[
 -\frac{d}{du}\int_{\S} f=\int_{\S} \left(-\frac{df}{du}+\Omega \trchb f\right)=\int_{\S} \Omega\left(e_3(f)+ \trchb f\right).
\]
Using this identity, we obtain
\begin{align}\label{evolution.id}
&\quad -\frac{d}{du}(\int_{\S}|u|^{2\lambda_1}|\phi|^2)\\
\notag &=\int_{\S}\!\!\Omega\bigg(\! 2\lambda_1|u|^{2\lambda_1-1}(e_3 u)|\phi|^2+2|u|^{2\lambda_1}<\phi,\nab_3\phi>+\tr\underline{\chi}|u|^{2\lambda_1}|\phi|^2\!\bigg)\\
\notag &= \int_{\S}\Omega \bigg( 2|u|^{2\lambda_1}<\phi, \nab_3\phi+\lambda_0\trchb\phi>\bigg)\\
\notag &\quad +\int_{\S}\bigg( |u|^{2\lambda_1} \Omega \left(\f{2\lambda_1 (e_3u)}{|u|}+(1-2\lambda_0)\trchb\right)|\phi|^2\bigg).
\end{align}
Observe that\footnote{Note that in the following formula we need the exact cancellation and thus we do not use the schematic notation for this equation.} we have
\begin{equation*}
\begin{split}
&\quad\ \f{2\lambda_1 (e_3u)}{|u|}+(1-2\lambda_0)\trchb\\
&= -\f{2\lambda_1-4\lambda_0+2}{|u|}-\f{2\lambda_1(\Omega^{-1}-1)}{|u|}+(1-2\lambda_0)\left(\trchb+\f2{|u|}\right)\\
&\ls \f{\delta\at b^{\f14}}{|u|^2}
\end{split}
\end{equation*}
using Proposition \ref{Omega} and the bootstrap assumption \eqref{BA.1}, since we have chosen the parameters to satisfy $2\lambda_1-4\lambda_0+2=0$.

Therefore, 
\begin{equation*}
\begin{split}
\left|-\frac{d}{du}\left(\int_{\S}|u|^{2\lambda_1}|\phi|^2\right)\right|
\ls &\int_{\S}\bigg( 2|u|^{2\lambda_1}|\phi| |F| +|u|^{2\lambda_1-2} \delta\at b^{\f14}|\phi|^2\bigg).
\end{split}
\end{equation*}
Using Cauchy-Schwarz for the first term and applying Gronwall's inequality for the second term, we obtain
\begin{equation*}
\begin{split}
&\quad\ |u|^{\lambda_1}\|\phi\|_{L^2(\S)}\\
&\ls e^{\de\at b^{\f14}\|u^{-2}\|_{L^1_u}}\left(\|\phi\|_{L^2(S_{1,\underline{u}})}+\int_u^1 |u'|^{\lambda_1}\|F\|_{L^2(S_{u',\underline{u}})}du'\right)\\
&\ls \|\phi\|_{L^2(S_{1,\underline{u}})}+\int_u^1 |u'|^{\lambda_1}\|F\|_{L^2(S_{u',\underline{u}})}du'.
\end{split}
\end{equation*}
since $\de\at b^{\f14}\|u^{-2}\|_{L^1_u}\ls \f{\de\at b^{\f14}}{|u|}\ls \f{1}{b^{\f34}}$.
\end{proof}

\subsection{Sobolev embedding}\label{Embedding}
Using the estimates for the metric $\gamma$ given in Proposition \ref{gamma}, we can follow \cite{Chr:book} to obtain a bound on the isoperimetric constant
$$I(S)=\sup_{\substack{U\\\partial U \in C^1}} \f{\min\{\mbox{Area}(U),\mbox{Area}(U^c)\}}{(\mbox{Perimeter}(\partial U))^2},$$
where $S$ is one of the $2$-spheres $\S$ adapted to the double null foliation. 
This will then allow us to obtain the necessary Sobolev embedding theorems. We first have the following estimate:
\begin{proposition}\label{isoperimetric}
Under the assumptions of Theorem \ref{main.thm} and the bootstrap assumptions \eqref{BA.1}, \eqref{BA.2}, \eqref{BA.3} and \eqref{BA.4}, the isoperimetric constant obeys the bound
$$I(\S)\leq \f1{\pi}$$
for $0\leq \ub\leq \de$ and $\de\at b\leq u\leq 1$.
\end{proposition}
\begin{proof}
Fix $u$. For a domain $U_{\ub}\subset S_{u,\ub}$, define $U_0\subset S_{u,0}$ to be the image of $U_{\ub}$ under the diffeomorphism generated
by the equivariant vector field $L$.
Using the notations introduced in Proposition \ref{gamma} and its proof, we have
$$\f{\mbox{Perimeter}(\partial U_{\ub})}{\mbox{Perimeter}(\partial U_0)}\geq \sqrt{\inf_{S_{u,0}} \lambda(\ub)} $$
and
$$\f{\mbox{Area}(U_{\ub})}{\mbox{Area}(U_0)}\leq \sup_{S_{u,0}} \xi(\ub),\quad\f{\mbox{Area}(U^c_{\ub})}{\mbox{Area}(U^c_0)}\leq \sup_{S_{u,0}} \xi(\ub).$$
Now, the conclusion follows from the fact that $I(S_{u,0})=\f{1}{2\pi}$ and the bounds in Proposition \ref{gamma}.
\end{proof}
We will only need an $L^2-L^\infty$ Sobolev embedding in this paper. In order to derive it, we will use the following two propositions, quoted directly from \cite{Chr:book}. The first is an $L^2-L^p$ embedding theorem:
\begin{proposition}[\cite{Chr:book}, Lemma 5.1]\label{Lp}
For any Riemannian $2$-manifold $(S,\gamma)$, we have the estimate
\begin{align*}
 (\mbox{Area}(S))^{-\f1p}\|\phi\|_{L^p(S)}&\leq C_p\sqrt{\max\{I(S),1\}}\\
 &\quad\times(\|\nab\phi\|_{L^2(S)}+(\mbox{Area}(S))^{-\f12}\|\phi\|_{L^2(S)})
\end{align*}
for $2<p<\infty$ and for any tensor $\phi$.
\end{proposition}
The second is an $L^p-L^\infty$ embedding:
\begin{proposition}[\cite{Chr:book}, Lemma 5.2]\label{Linfty}
For any Riemannian $2$-manifold $(S,\gamma)$, we have the estimate
\begin{align*}
\|\phi\|_{L^\infty(S)}&\leq C_p\sqrt{\max\{I(S),1\}}\\
&\quad\times(\mbox{Area}(S))^{\f12-\f1p}(\|\nab\phi\|_{L^p(S)}+(\mbox{Area}(S))^{-\f12}\|\phi\|_{L^p(S)})
\end{align*}
for $p>2$ and for any tensor $\phi$.
\end{proposition}
Recall from Proposition \ref{area} that $\mbox{Area}(\S)\sim |u|^2$. We can now combine Propositions \ref{isoperimetric}, \ref{Lp} and \ref{Linfty} to obtain 
\begin{proposition}\label{Sobolev}
Under the assumptions of Theorem \ref{main.thm} and the bootstrap assumptions \eqref{BA.1}, \eqref{BA.2}, \eqref{BA.3} and \eqref{BA.4}, we have
\begin{equation}
\begin{split}
\|\phi\|_{L^\infty(S_{u,\ub})} \ls & \sum_{i\leq 2}\|u^{i-1}\nab^i\phi\|_{L^2(\S)}.
\end{split}
\end{equation}
\end{proposition}

\subsection{Commutation formula}\label{commutation}
In this section, we derive general commutation formulae. We have the following formula from \cite{KNI:book}:
\begin{proposition}
The commutator $[\nabla_4,\nabla]$ acting on an $(0,r)$ S-tensor is given by
\begin{equation*}
 \begin{split}
[\nabla_4,\nabla_B]\phi_{A_1\cdots A_r} &= [D_4,D_B]\phi_{A_1\cdots A_r}+(\nabla_B\log\Omega)\nabla_4\phi_{A_1\cdots A_r}\\
&\quad -(\gamma^{-1})^{CD}\chi_{BD}\nabla_C\phi_{A_1\cdots A_r}\\
&\quad -\sum_{i=1}^r (\gamma^{-1})^{CD}\chi_{BD}\etab_{A_i}\phi_{A_1\cdots \hat{A_i}C\cdots A_r}\\
&\quad +\sum_{i=1}^r (\gamma^{-1})^{CD}\chi_{A_iB}\etab_{D}\phi_{A_1\cdots \hat{A_i}C\cdots A_r}.
 \end{split}
\end{equation*}
\end{proposition}

\begin{proposition}
The commutator $[\nabla_3,\nabla]$ acting on an $(0,r)$ S-tensor is given by
\begin{equation*}
 \begin{split}
[\nabla_3,\nabla_B]\phi_{A_1\cdots A_r} &= [D_3,D_B]\phi_{A_1\cdots A_r}+(\nabla_B\log\Omega)\nabla_3\phi_{A_1\cdots A_r}\\
&\quad  -(\gamma^{-1})^{CD}\chib_{BD}\nabla_C\phi_{A_1\cdots A_r} \\
&\quad -\sum_{i=1}^r (\gamma^{-1})^{CD}\chib_{BD}\eta_{A_i}\phi_{A_1\cdots \hat{A_i}C\cdots A_r}\\
&\quad +\sum_{i=1}^r (\gamma^{-1})^{CD}\chib_{A_iB}\eta_{D}\phi_{A_1\cdots \hat{A_i}C\cdots A_r}.
 \end{split}
\end{equation*}
\end{proposition}

By induction, we get the following schematic formula for repeated commutations (see \cite{{L-R:Propagation}}):
\begin{proposition}\label{commuteeqn}
Suppose $\nabla_4\phi=F_0$. Let $\nabla_4\nabla^i\phi=F_i$.
Then
\begin{equation*}
\begin{split}
F_i &= \sum_{i_1+i_2+i_3=i}\nabla^{i_1}(\eta+\underline{\eta})^{i_2}\nabla^{i_3} F_0\\
&\quad +\sum_{i_1+i_2+i_3+i_4=i}\nabla^{i_1}(\eta+\underline{\eta})^{i_2}\nabla^{i_3}\chi\nabla^{i_4} \phi\\
&\quad +\sum_{i_1+i_2+i_3+i_4=i-1} \nabla^{i_1}(\eta+\underline{\eta})^{i_2}\nabla^{i_3}\beta\nabla^{i_4} \phi.
\end{split}
\end{equation*}
where by $\nabla^{i_1}(\eta+\underline{\eta})^{i_2}$ we mean the sum of all terms which is a product of $i_2$ factors, each factor being $\nabla^j (\eta+\underline{\eta})$ for some $j$ and that the sum of all $j$'s is $i_1$, i.e., $\nabla^{i_1}(\eta+\underline{\eta})^{i_2}=\displaystyle\sum_{j_1+\cdots +j_{i_2}=i_1}\nabla^{j_1}(\eta+\underline{\eta})\cdots \nabla^{j_{i_2}}(\eta+\underline{\eta})$. Similarly, suppose $\nabla_3\phi=G_{0}$. Let $\nabla_3\nabla^i\phi=G_{i}$.
Then
\begin{equation*}
\begin{split}
G_{i}-\frac{i}{2}\tr\chib \nab^i \phi &= \sum_{i_1+i_2+i_3=i}\nabla^{i_1}(\eta+\underline{\eta})^{i_2}\nabla^{i_3} G_{0}\\
&\quad +\sum_{i_1+i_2+i_3+i_4=i}\nabla^{i_1}(\eta+\underline{\eta})^{i_2}\nabla^{i_3}\left(\chibh,\tr\chib+\f{2}{u}\right)\nabla^{i_4} \phi\\
&\quad +\sum_{i_1+i_2+i_3+i_4=i-1} \nabla^{i_1}(\eta+\underline{\eta})^{i_2}\nabla^{i_3}\underline{\beta}\nabla^{i_4} \phi.
\end{split}
\end{equation*}

\end{proposition}
The following further simplified version is useful in the latter sections:
\begin{proposition}
Suppose $\nabla_4\phi=F_0$. Let $\nabla_4\nabla^i\phi=F_i$.
Then
\begin{equation*}
\begin{split}
F_i= \sum_{i_1+i_2+i_3=i}\nabla^{i_1}\q^{i_2}\nabla^{i_3} F_0+\sum_{i_1+i_2+i_3+i_4=i}\nabla^{i_1}\q^{i_2}\nabla^{i_3}\p\nabla^{i_4} \phi.
\end{split}
\end{equation*}
Similarly, suppose $\nabla_3\phi=G_{0}$. Let $\nabla_3\nabla^i\phi=G_{i}$.
Then
\begin{equation*}
\begin{split}
G_{i}-\f {i}{2} \tr\chib\nab^i\phi &= \sum_{i_1+i_2+i_3=i}\nabla^{i_1}\q^{i_2}\nabla^{i_3} G_{0}+\sum_{i_1+i_2+i_3=i}\nabla^{i_1}\q^{i_2} \nab^{i_3}\phi\\
&\quad +\sum_{i_1+i_2+i_3=i, i_3\leq i-1} \tr\chib\nabla^{i_1}\q^{i_2}\nabla^{i_3} \phi.
\end{split}
\end{equation*}
\end{proposition}
\begin{proof}
We first replace $\beta$ and $\betab$ using the schematic Codazzi equations:
\begin{gather*}
\beta=\nabla\p+\p\q,\\
\beb=\nabla\q+\q (\tr\chib+\q).
\end{gather*}
We then replace $\trch$ and $\chih$ by $\p$, as well as substitute $\eta$, $\etab$, $\chibh$ and $\trchb+\f2u$ with $\q$.
\end{proof}

\subsection{General elliptic estimates for Hodge systems}\label{elliptic}

In this subsection, we prove elliptic estimates for general Hodge systems. To this end, we recall the definition of the divergence and curl of a symmetric covariant tensor of an arbitrary rank:
\begin{align*}
(\div\phi)_{A_1\cdots A_r}&=\nabla^B\phi_{BA_1\cdots A_r},\\
(\curl\phi)_{A_1\cdots A_r}&=\eps^{BC}\nabla_B\phi_{CA_1\cdots A_r},
\end{align*}
where $\eps$ is the volume form associated to the metric $\gamma$.
Recall also that the trace is defined to be
$$(\mbox{tr }\phi)_{A_1\cdots A_{r-1}}=(\gamma^{-1})^{BC}\phi_{BCA_1\cdots A_{r-1}}.$$
The following is the main $L^2(\S)$ elliptic estimate that we will use:
\begin{proposition}\label{ellipticthm}
We continue to work under the assumptions of Theorem \ref{main.thm} and the bootstrap assumptions \eqref{BA.1}, \eqref{BA.2}, \eqref{BA.3} and \eqref{BA.4}. Let $\phi$ be a totally symmetric $r+1$ covariant tensorfield on a 2-sphere $(\mathbb S^2,\gamma)$ satisfying
$$\div\phi=f,\quad \curl\phi=g,\quad \mbox{tr}\phi=h.$$
Then, for $1\leq i\leq 4$, we have
\begin{equation*}
 \begin{split}
 \|u^i\nabla^{i}\phi\|_{L^2(\S)}
\ls \sum_{j=0}^{i-1}\big(&\|u^{j+1}\nabla^{j}(f,g)\|_{L^2(\S)}\\
&+\|u^j\nab^j h\|_{L^2(\S)}+\|u^j\nab^j\phi\|_{L^2(\S)}\big).
 \end{split}
\end{equation*}
\end{proposition}
\begin{proof}
Recall the following identity from Chapter 7 in \cite{Chr:book} that for $\phi$, $f$, $g$ and $h$ as above
\begin{equation}\label{basic.L2}
\int_{\S} \bigg( |\nab\phi|^2+(r+1) K |\phi|^2\bigg) = \int_{\S} \bigg(|f|^2+|g|^2+K |h|^2\bigg).
\end{equation}
Notice that $\| K\|_{L^\infty_u L^\infty_{\ub} L^\infty(\S)}\ls \f1{|u|^2}$ by the bootstrap assumption \eqref{BA.4} and the Sobolev embedding theorem (Proposition \ref{Sobolev}). This implies the conclusion for $i=1$ after multiplying \eqref{basic.L2} by $u^2$. For $i>1$, we recall again from \cite{Chr:book} that the symmetrized angular derivative of $\phi$ defined by 
$$(\nab\phi)^s_{BA_1\cdots A_{r+1}}:= \f{1}{r+2}\left(\nab_B\phi_{A_1\cdots A_r}+\sum_{i=1}^{r+1}\nab_{A_i}\phi_{A_1\cdots <A_i>B\cdots A_{r+1}}\right)$$
satisfies the div-curl system
\begin{equation*}
\begin{split}
\div(\nab\phi)^s&=(\nab f)^s-\f{1}{r+2}(^*\nab g)^s+(r+1)K\phi-\f{2K}{r+1}(\gamma\otimes^s h)\\
\curl(\nab\phi)^s &=\f{r+1}{r+2}(\nab g)^s+(r+1)K(^*\phi)^s\\
\mbox{tr }(\nab\phi)^s &=\f{2}{r+2}f+\f{r}{r+2}(\nab h)^s,
\end{split}
\end{equation*}
where
$$(\gamma\otimes^s h)_{A_1\cdots A_{r+1}}:=\gamma_{A_iA_j}\sum_{i<j=1,\cdots ,r+1}h_{A_1\cdots <A_i>\cdots <A_j>\cdots A_{r+1}}$$
and
$$(^*\phi)^s_{A_1\cdots A_{r+1}}:=\f{1}{r+1}\sum_{i=1}^{r+1}\eps_{A_i}{ }^B\phi_{A_1\cdots <A_i>B\cdots A_r}.$$
Using \eqref{basic.L2}, we therefore obtain that for $i=2$, we have
\begin{equation*}
\begin{split}
\|\nab^2 \phi\|_{L^2(\S)}^2 &\ls \|\nab f\|_{L^2(\S)}^2+\|\nab g\|_{L^2(\S)}^2\\
&\quad +\|K(|\nab\phi|^2+|f|^2+|\nab h|^2)\|_{L^1(\S)}\\
&\quad +\|K\phi\|_{L^2(\S)}^2+\|K h\|_{L^2(\S)}^2.
\end{split}
\end{equation*}
Using again $\|K\|_{L^\infty_uL^\infty_{\ub}L^\infty(\S)}\ls \f{1}{|u|^2}$, we have
\begin{equation*}
\begin{split}
\|\nab^2 \phi\|_{L^2(\S)}^2 &\ls \|\nab (f,g)\|_{L^2(\S)}^2
+\f1{|u|^2}(\|\nab(\phi,h)\|_{L^2(\S)}^2+\|f\|_{L^2(\S)}^2)\\
&\quad +\f{1}{|u|^4}\|(\phi,h)\|_{L^2(\S)}^2,
\end{split}
\end{equation*}
which implies, after multiplying by $|u|^4$, that
\begin{equation*}
\begin{split}
\|u^2\nab^2 \phi\|_{L^2(\S)}^2 &\ls \|u^2\nab (f,g)\|_{L^2(\S)}^2+\|u\nab(\phi,h)\|_{L^2(\S)}^2\\
&\quad +\|u f\|_{L^2(\S)}^2+\|(\phi,h)\|_{L^2(\S)}^2\\
&\ls \sum_{j\leq 1}(\|u^{j+1}\nab^j(f,g)\|_{L^2(\S)}^2+\|u^j \nab^j(\phi,g)\|_{L^2(\S)}).
\end{split}
\end{equation*}
Iterating the procedure and applying \eqref{basic.L2}, we thus obtain that for $i\leq 4$,
\begin{equation*}
\begin{split}
\|\nab^i \phi\|_{L^2(\S)}^2 &\ls \|\nab^{i-1} (f,g)\|_{L^2(\S)}^2\\
&\quad +\|K(|\nab^{i-1}(\phi,h)|^2 +|\nab^{i-2}(f,g)|^2)\|_{L^1(\S)}\\
&\quad +\bigg\|K \bigg(\sum_{i_1+i_2=i-3}\nab^{i_1}K\nab^{i_2}(\phi,h)\bigg)^2\bigg\|_{L^1(\S)}\\
&\quad + \bigg\|K \bigg(\sum_{i_1+i_2=i-4}\nab^{i_1}K\nab^{i_2}f\bigg)^2\bigg\|_{L^1(\S)}\\ 
&\quad +\sum_{i_1+i_2=i-2}\|\nab^{i_1}K\nab^{i_2}(\phi,h)\|_{L^2(\S)}^2\\
&\quad +\sum_{i_1+i_2=i-3}\|\nab^{i_1}K \nab^{i_2}(f,g)\|_{L^2(\S)}^2\\
&\quad +\sum_{i_1+i_2=i-4}\|K\nab^{i_1}K\nab^{i_2}(\phi,h)\|_{L^2(\S)}^2
\end{split}
\end{equation*}
where we have used the convention that $\sum_{i\leq -1}=0$. By the bootstrap assumption \eqref{BA.4} and Sobolev embedding (Proposition \ref{Sobolev}), we have 
$$\sum_{i\leq 2}\|u^i\nab^i K\|_{L^\infty(\S)}\ls 1.$$
Therefore, for $i\leq 4$, we have 
$$\|u^i\nab^i\phi\|_{L^2(\S)}^2 \ls \sum_{j\leq i-1}\|u^{j+1}\nab^j(f,g)\|_{L^2(\S)}^2+\|u^j\nab^j(\phi,h)\|_{L^2(\S)}^2.\vspace{-2em}$$
\end{proof}

For the special case that $\phi$ a symmetric traceless 2-tensor, we only need to know its divergence:
\begin{proposition}\label{elliptictraceless}
Suppose $\phi$ is a symmetric traceless 2-tensor satisfying
$$\div\phi=f.$$
Then, under the assumptions of Theorem \ref{main.thm} and the bootstrap assumptions \eqref{BA.1}, \eqref{BA.2}, \eqref{BA.3} and \eqref{BA.4}, for $1\leq i\leq 4$, we have
\begin{equation*}
 \begin{split}
 \|u^i\nabla^{i}\phi\|_{L^2(\S)}
\ls \sum_{j=0}^{i-1}(\|u^{j+1}\nabla^{j}f\|_{L^2(\S)}+\|u^j\nab^j\phi\|_{L^2(\S)}).
 \end{split}
\end{equation*}

\end{proposition}
\begin{proof}
In view of Proposition \ref{ellipticthm}, this Proposition follows from
$$\curl\phi=^*f.$$
This is a straightforward computation using the fact that $\phi$ is both symmetric and traceless.
\end{proof}

\section{$L^2(S_{u,\protect\underline{u}})$ estimates for Ricci coefficients}\label{secRicci}

In this section, we prove estimates for the Ricci coefficients and their first four angular derivatives in $L^2(S)$. 

Before we proceed to prove estimates for the Ricci coefficients, we first make a preliminary observation regarding the bounds for $\q$ and its derivatives and products. These follow directly from the bootstrap assumption~\eqref{BA.1}.
\begin{proposition}\label{product}
Under the assumptions of Theorem \ref{main.thm} and the bootstrap assumptions \eqref{BA.1}, \eqref{BA.2}, \eqref{BA.3} and \eqref{BA.4}, we have
\begin{equation*}
\begin{split}
\sum_{i_1+i_2\leq 4}\|u^{i_1+i_2+1}\nabla^{i_1}\q^{i_2+1}\|_{L^{2}(\S)}\ls &\sum_{i_1\leq 4}\|u^{i_1+1}\nab^{i_1} \q\|_{L^2(\S)}
\end{split}
\end{equation*}
and
\begin{equation*}
\begin{split}
\sum_{i_1+i_2\leq 2}\|u^{i_1+i_2+2}\nabla^{i_1}\q^{i_2+1}\|_{L^{\infty}(\S)}
\ls &\sum_{i_1\leq 2}\|u^{i_1+2}\nab^{i_1} \q\|_{L^\infty(\S)}.
\end{split}
\end{equation*}
In particular, by \eqref{BA.1}, we have
\begin{align*}
&\quad \sum_{i_1+i_2\leq 4}\|u^{i_1+i_2+1}\nabla^{i_1}\q^{i_2+1}\|_{L^{2}(\S)}+\sum_{i_1+i_2\leq 2}\|u^{i_1+i_2+2}\nabla^{i_1}\q^{i_2+1}\|_{L^{\infty}(\S)}\\
&\ls \de \at b^{\f14}.
\end{align*}
\end{proposition}
\begin{proof}
For the $L^2(\S)$ estimates, we have
\begin{equation*}
\begin{split}
&\quad\ \sum_{i_1+i_2\leq 4}\|u^{i_1+i_2+1}\nabla^{i_1}\q^{i_2+1}\|_{L^{2}(\S)}\\
&\ls \bigg(\sum_{i_1\leq 4}\|u^{i_1+1}\nab^{i_1} \q\|_{L^2(\S)}\bigg)\bigg(\sum_{i_2\leq 4}\sum_{i_3\leq 2}\|u^{i_3+1}\nab^{i_3}\q\|_{L^\infty(\S)}^{i_2}\bigg)\\
&\ls \sum_{i_1\leq 4}\|u^{i_1+1}\nab^{i_1} \q\|_{L^2(\S)},
\end{split}
\end{equation*}
since by \eqref{BA.1} and Sobolev embedding, we have
$$\sum_{i_3\leq 2}\|u^{i_3+1}\nab^{i_3}\q\|_{L^\infty(\S)}\ls \frac{\delta\at b^{\f14}}{|u|}\ls \frac{1}{b^{\f34}}.$$
Similarly, we have
\begin{equation*}
\begin{split}
&\quad\ \sum_{i_1+i_2\leq 2}\|u^{i_1+i_2+2}\nabla^{i_1}\q^{i_2+1}\|_{L^{\infty}(\S)}\\
&\ls \bigg(\sum_{i_1\leq 2}\|u^{i_1+2}\nab^{i_1} \q\|_{L^\infty(\S)}\bigg)\bigg(\sum_{i_2\leq 2}\sum_{i_3\leq 2}\|u^{i_3+1}\nab^{i_3}\q\|_{L^\infty(\S)}^{i_2}\bigg)\\
&\ls \sum_{i_1\leq 2}\|u^{i_1+2}\nab^{i_1} \q\|_{L^\infty(\S)}.\\[-1.5em]
\end{split}
\end{equation*}
\end{proof}

We now proceed to the estimates for the Ricci coefficients and their derivatives. We will first bound the terms we denote as $\p$, first with $\chih$ (Proposition \ref{chih.bd}), then $\trch$ (Proposition \ref{trch.bd}) and $\om$ (Proposition \ref{om.bd}). We will then turn to the estimates for $\q$ in Proposition \ref{q.bd}. We begin with $\chih$:

\begin{proposition}\label{chih.bd}
Under the assumptions of Theorem \ref{main.thm} and the bootstrap assumptions \eqref{BA.1}, \eqref{BA.2}, \eqref{BA.3} and \eqref{BA.4}, we have
\[
 \sum_{i\leq 4}\|u^i\nab^i\chih\|_{L^2(S_{u,\ub})} \ls a^{\frac 12}.
\]
\end{proposition}

\begin{proof}
We use the null structure equation
$$\nab_3\chih+\f12 \tr\chib\chih=\nab\eta+\p\q+\q\q.$$ 
Commuting this equation with $i$ angular derivatives, we have
\begin{equation*}
\begin{split}
&\quad\ \nab_3 \nab^i\chih +\frac{i+1}{2}\tr\chib\nab^i\chih\\
&= \nab^{i+1}\eta+\sum_{i_1+i_2=i} \nab^{i_1}\q^{i_2+2}+\sum_{i_1+i_2+i_3=i}\nabla^{i_1}\q^{i_2+1}\nab^{i_3}\p\\
&\quad +\sum_{i_1+i_2+i_3=i-1}\frac 1u \nabla^{i_1}\q^{i_2+1}\nab^{i_3}\p.
\end{split}
\end{equation*}

We apply Proposition \ref{evolution lemma} with $\lambda_0=\frac{i+1}{2}$, which shows that the quantity $\|u^i\nab^{i}\chih\|_{L_u^{\infty}L_{\ub}^{\infty}L^{2}(\S)}$ can be controlled by the sum of its
initial value and the $\|u^i \cdot\|_{L_{\ub}^{\infty}L_{u}^{1}L^{2}(\S)}$ norm of the right hand side. We now estimate each of the terms on the right hand side of the equation for $i\leq 4$. We first control the linear term in $\eta$ for $i\leq 3$: 
\begin{equation*}
\begin{split}
\sum_{i\leq 3}\|u^i\nab^{i+1}\eta\|_{L_{u}^{1}L^{2}(\S)} &\leq \sum_{i\leq 3}\left\|\frac{1}{|u|^2}\right\|_{L^1_u}\|u^{i+2}\nab^{i+1}\q\|_{L_{\ub}^{\infty}L^{2}(\S)}\\
&\ls \f{\de\at b^{\frac 14}}{|u|}.
\end{split}
\end{equation*}
using the bootstrap assumption \eqref{BA.1}.

For the highest order derivative, i.e., when $i=4$, we have
\begin{equation*}
\begin{split}
\|u^4\nab^5\eta\|_{L_{u}^{1}L^{2}(\S)} &\leq \left\|\frac{1}{|u|}\right\|_{L^2_u}\|u^5\nab^5\eta\|_{L_{u}^{2}L^{2}(\S)}\\
&\ls \frac{\delta^{\f12}a^{\f12}}{|u|^{\f12}} \tilde{\M {O}}_{5,2}\ls \f{\de^{\f12} a^{\f12} b^{\f14}}{|u|^{\f12}}.
\end{split}
\end{equation*}

We then control the second and third terms together. Here, we use the estimates derived in Proposition \ref{product}.
\begin{align}\label{chih.3}
&\sum_{i\leq 4}\left\|\sum_{i_1+i_2+i_3=i}u^i\nabla^{i_1}\q^{i_2+1}\nab^{i_3}(\p,\q)\right\|_{L_{u}^{1}L^{2}(\S)}\\
\notag &\ls  \sum_{i_1+i_2\leq 4}\|u^{i_1+i_2-1}\nabla^{i_1}\q^{i_2+1}\|_{L_{u}^{1}L^{2}(\S)}\\
\notag &\quad\times\sum_{i_3\leq 2}\|u^{i_3+1}\nab^{i_3}(\p,\q)\|_{L_{u}^{\infty}L^{\infty}(\S)}\\
\notag &\quad + \sum_{i_1+ i_2 \leq 2}\|u^{i_1+i_2-1}\nabla^{i_1}\q^{i_2+1}\|_{L_{u}^{1}L^{\infty}(\S)}\\
\notag &\quad\times \sum_{i_3\leq 4}\|u^{i_3+1}\nab^{i_3}(\p,\q)\|_{L_{u}^{\infty}L^{2}(\S)}\\
\notag &\ls \frac{\de\at b^{\f14}}{|u|}\left(\at b^{\frac 14}+\frac{\de\at b^{\f14}}{|u|}\right) \\
\notag &\ls \frac{\de a b^{\f12}}{|u|}.
\end{align}
Finally, we control the last term by
\begin{align}\label{chih.4}
&\quad \sum_{i\leq 4}\left\|\sum_{i_1+i_2+i_3=i-1}u^{i-1}\nabla^{i_1}\q^{i_2+1}\nab^{i_3}\p\right\|_{L_{u}^{1}L^{2}(\S)}\\
\notag &\ls  \sum_{i_1+ i_2\leq 3}\|u^{i_1+i_2-1}\nabla^{i_1}\q^{i_2+1}\|_{L_{u}^{1}L^{2}(\S)}\sum_{i_3\leq 1}\|u^{i_3+1}\nab^{i_3}\p\|_{L_{u}^{\infty}L^{\infty}(\S)}\\
\notag &\quad + \sum_{i_1+i_2 \leq 1}\|u^{i_1+i_2-1}\nabla^{i_1}\q^{i_2+1}\|_{L_{u}^{1}L^{\infty}(\S)}\sum_{i_3\leq 3}\|u^{i_3+1}\nab^{i_3}\p\|_{L_{u}^{\infty}L^{2}(\S)}\\
\notag &\ls\frac{\de\at b^{\f14}}{|u|}\left(\at b^{\frac 14}+\frac{\de\at b^{\f14}}{|u|}\right) \\
\notag &\ls \frac{\de a b^{\f12}}{|u|}.
\end{align}
We now apply the condition $|u|\geq \de \at b$. For $b$ sufficiently large and $a\geq b$, it is easy to see that all the terms above are of size
$$\ls \f{\de^{\f12}a^{\f34}}{|u|^{\f12}}\ls \f{a^{\f12}}{b^{\f12}}\ls \at.$$
Recalling that initially we have
$$\sum_{i\leq 4}\|\nab^i \chih_0\|_{L^\infty_{\ub}L^2(S_{0,\ub})}\leq \at$$
and using the above estimates, we get 
$$\sum_{i\leq 4}\|u^i\nab^i\chih\|_{L^{\infty}_{u}L^2(\S)}\ls a^{\frac 12}.\vspace{-2em}$$
\end{proof}
We next prove the estimates for $\trch$ and its derivatives. It will be useful to show not only that $\nab^i\trch$ obey the required estimates for $\nab^i\p$ but also to obtain a slightly more refined bound as follows:

\begin{proposition}\label{trch.bd}
Under the assumptions of Theorem \ref{main.thm} and the bootstrap assumptions \eqref{BA.1}, \eqref{BA.2}, \eqref{BA.3} and \eqref{BA.4}, we have
$$\sum_{i\leq 4} \left\|u^{i+1}\nab^i\left(\trch-\frac 2{|u|}\right)\right\|_{L^2(\S)}\ls \de a.$$
\end{proposition}

\begin{proof}
Using the null structure equation, we have 
$$\nab_4 \left(\tr\chi-\frac{2}{|u|}\right)=\chih\chih+\om\trch+\trch\trch.$$
Notice that we have put in the term $-\frac{2}{|u|}$ in the equation for $\nab_4\trch$. The $\nab_4$ derivative of this term is $0$. We put in this additional term because the initial value for $\tr\chi-\frac{2}{|u|}$ on $\Hb_0$ is $0$.
Commuting this equation with $i$ angular derivatives, we have
\begin{equation*}
\begin{split}
 \nab_4 \nab^i\left(\trch-\frac{2}{|u|}\right) &= \sum_{i_1+i_2+i_3+i_4=i}\nab^{i_1}\q^{i_2}\nab^{i_3}\chih\nab^{i_4}\chih\\
 &\quad +\sum_{i_1+i_2+i_3+i_4=i}\nab^{i_1}\q^{i_2}\nab^{i_3}\p\nab^{i_4}\trch\\
&\quad +\sum_{i_1+i_2+i_3=i}\frac{1}{|u|}\nab^{i_1}\q^{i_2}\nab^{i_3}\p.
\end{split}
\end{equation*}
We now use Proposition \ref{transport} to bound $\|u^{i+1}\nab^i(\trch-\f{2}{|u|})\|_{L^\infty_uL^\infty_{\ub}L^2(\S)}$ by controlling the right hand side in the $\|u^{i+1}\cdot\|_{L^{\infty}_uL^1_{\ub}L^2(\S)}$ norm.

We first bound each of the three terms in the case where $i_2=0$. We first consider the contribution from $\sum_{i_1+i_2+i_3+i_4=i}\nab^{i_1}\q^{i_2}\nab^{i_3}\chih\nab^{i_4}\chih$ where $i_2=0$ (and as a consequence $i_1=0$). In this term, we will see estimates that are ``borderline'', but we will use the fact that the bound derived for $\nab^i\chih$ in Proposition \ref{chih.bd} is independent of the bootstrap assumption \eqref{BA.1} and does not suffer a loss of $b^{\f14}$. More precisely, we have the bound
\begin{equation*}
\begin{split}
&\quad\ \sum_{i\leq 4}\left\|\sum_{i_3+i_4=i}u^{i+1}\nab^{i_3}\chih\nab^{i_4}\chih\right\|_{L_{\ub}^{1}L^{2}(\S)}\\
&\ls \de\sum_{i_3\leq 2}\|u^{i_3+1}\nab^{i_3}\chih\|_{L_{u}^{\infty} L^{\infty}_{\ub} L^{\infty}(\S)}\sum_{i_4\leq 4}\|u^{i_4}\nab^{i_4}\chih\|_{L_{u}^\infty L^{\infty}_{u} L^2(\S)}\\
&\ls  \de a.
\end{split}
\end{equation*}
We now consider the contribution from $\sum_{i_1+i_2+i_3+i_4=i}\nab^{i_1}\q^{i_2}\nab^{i_3}\p\nab^{i_4}\trch$ with $i_2=0$.

\begin{equation*}
\begin{split}
&\quad\ \sum_{i\leq 4}\left\|\sum_{i_3+i_4=i}u^{i+1}\nab^{i_3}\p\nab^{i_4}\trch\right\|_{L_{\ub}^{1}L^{2}(\S)}\\
&\ls \de\sum_{i_3\leq 2}\|u^{i_3}\nab^{i_3}\p\|_{L_{u}^{\infty} L^{\infty}_{\ub} L^{\infty}(\S)}\sum_{i_4\leq 4}\|u^{i_4+1}\nab^{i_4}\trch\|_{L_{u}^\infty L^{\infty}_{u} L^2(\S)}\\
&\quad +\de\sum_{i_3\leq 4}\|u^{i_3+1}\nab^{i_3}\p\|_{L_{u}^{\infty} L^{\infty}_{\ub} L^{2}(\S)}\sum_{i_4\leq 4}\|u^{i_4}\nab^{i_4}\trch\|_{L_{u}^\infty L^{\infty}_{u} L^{\infty}(\S)}\\
&\ls  \frac{\de \at b^{\f14}}{|u|}\sum_{i\leq 4}\left\|u^{i+1}\nab^{i}\left(\trch-\frac {2}{|u|}\right)\right\|_{L_{u}^\infty L^{\infty}_{u} L^2(\S)}+\de\at b^{\frac 14},
\end{split}
\end{equation*}
where we have used the bootstrap assumption \eqref{BA.2} together with Sobolev embedding in Proposition \ref{Sobolev}. 

The contribution of the third term, i.e., $\sum_{i_1+i_2+i_3=i}\frac{1}{|u|}\nab^{i_1}\q^{i_2}\nab^{i_3}\p$,\linebreak where $i_2=0$ can be controlled by
\begin{equation*}
\sum_{i\leq 4}\|u^{i}\nab^{i}\p\|_{L_{\ub}^{1}L^{2}(\S)}
\ls \de\sum_{i\leq 4}\|u^{i}\nab^{i}\p\|_{L_{u}^\infty L^{\infty}_{u} L^2(\S)}
\ls \de \at b^{\frac 14},
\end{equation*}
using the bootstrap assumption \eqref{BA.2}.

We now move to the terms where $i_2\geq 1$. It turns out that the $\q$ factors provide extra smallness and we can use the bootstrap assumption \eqref{BA.2} together with Proposition \ref{product} to deal with the first two terms\footnote{Notice that we have relabeled $i_2$ to simplify the notation. We have also used the schematic notation to write $\chih$ and $\trch$ as $\p$.} to get
\begin{equation*}
\begin{split}
&\quad\ \sum_{i\leq 4}\left\|\sum_{i_1+i_2+i_3+i_4=i-1}u^{i+1}\nab^{i_1}\q^{i_2+1}\nab^{i_3}\p\nab^{i_4}\p\right\|_{L_{\ub}^{1}L^{2}(\S)}\\
&\ls \f{\de}{|u|}\!\sum_{\substack{i_1+i_2\leq 3\\i_1\leq 2}}\|u^{i_1+i_2+2}\nab^{i_1}\q^{i_2+1} \|_{L^\infty_u L^{\infty}_{\ub} L^\infty(\S)}\!\sum_{i_3\leq 2}\|u^{i_3+1}\nab^{i_3}\p\|_{L_{u}^{\infty} L^{\infty}_{\ub} L^{\infty}(\S)}\\
&\quad\times\sum_{i_4\leq 3}\|u^{i_4}\nab^{i_4}\p\|_{L_{u}^\infty L^{\infty}_{u} L^2(\S)}\\
&\quad + \f{\de}{|u|} \|u^4\nab^3\q \|_{L^\infty_u L^{\infty}_{\ub} L^2(\S)}\|u \p\|_{L_{u}^{\infty} L^{\infty}_{\ub} L^{\infty}(\S)}\|u \p\|_{L_{u}^{\infty} L^{\infty}_{\ub} L^{\infty}(\S)}\\
&\ls \frac{\de^2 a^{\f32}b^{\f34}}{|u|}.
\end{split}
\end{equation*}
The third term can be handled in a similar way in the case $i_2\geq 1$ using Proposition \ref{product} and the bootstrap assumption \eqref{BA.2}
\begin{equation*}
\begin{split}
&\quad\ \sum_{i\leq 4}\left\|\sum_{i_1+i_2+i_3=i-1}u^{i}\nab^{i_1}\q^{i_2+1}\nab^{i_3}\p\right\|_{L_{\ub}^{1}L^{2}(\S)}\\
&\ls \f{\de}{|u|}\sum_{\substack{i_1+i_2\leq 3\\i_1\leq 2}}\|u^{i_1+i_2+2}\nab^{i_1}\q^{i_2+1} \|_{L^\infty_u L^{\infty}_{\ub} L^\infty(\S)}\sum_{i_3\leq 3}\|u^{i_3}\nab^{i_3}\p\|_{L_{u}^\infty L^{\infty}_{u} L^2(\S)}\\
&\quad + \f{\de}{|u|} \|u^4\nab^3\q \|_{L^\infty_u L^{\infty}_{\ub} L^2(\S)}\|u \p\|_{L_{u}^{\infty} L^{\infty}_{\ub} L^{\infty}(\S)}\\
&\ls  \frac{\de^2 a b^{\f12}}{|u|}.
\end{split}
\end{equation*}

Recall that the initial data for $\trch-\f{2}{|u|}$ is vanishing. Therefore, by Proposition \ref{transport} and using the condition $|u|\geq \de \at b$, we have
\begin{equation*}
\begin{split}
&\quad\ \sum_{i\leq 4} \left\|u^{i+1}\nab^i\left(\trch-\f2{|u|}\right)\right\|_{L^{\infty}_u L^\infty_{\ub}L^2(\S)}\\
&\ls \de a+\f{1}{b^{\f34}}\sum_{i\leq 4} \left\|u^{i+1}\nab^i\left(\trch-\f2{|u|}\right)\right\|_{L^{\infty}_ uL^\infty_{\ub}L^2(\S)}.
\end{split}
\end{equation*}
For $b$ sufficiently large, we can subtract the last term from both sides to get
$$\sum_{i\leq 4} \left\|u^{i+1}\nab^i\left(\trch-\f2{|u|}\right)\right\|_{L^{\infty}_u L^\infty_{\ub}L^2(\S)}\ls \de a.\vspace{-1em}$$
\end{proof}

We next prove the desired estimates for $\om$ and its derivatives. We will show that it obeys a slightly better estimate than a general component $\p$. Moreover, the bound improves if we take at least one angular derivative.
\begin{proposition}\label{om.bd}
Under the assumptions of Theorem \ref{main.thm} and the bootstrap assumptions \eqref{BA.1}, \eqref{BA.2}, \eqref{BA.3} and \eqref{BA.4}, we have
$$\sum_{i\leq 4} \|u^i\nab^i\om\|_{L^2(\S)}\ls 1+\frac{\delta^{\f12}a^{\f34}}{|u|^{\f12}}.$$
If $i\neq 0$, we get the improved bound
$$\sum_{1\leq i\leq 4} \|u^i\nab^i\om\|_{L^2(\S)}\ls \frac{\delta^{\f12}a^{\f34}}{|u|^{\f12}}.$$
\end{proposition}

\begin{proof}
We use the following schematic null structure equation for $\omega$:
$$\nab_3\omega=K+\p\q+\q\q+ \tr\chi\tr\chib.$$ 
Commuting it with angular derivative for $i$ times, we have
\begin{equation*}
\begin{split}
&\quad\ \nab_3 \nab^i\omega +\frac i2 \trchb\nab^i\omega\\
&= \sum_{i_1+i_2+i_3\leq i}\nab^{i_1}\q^{i_2}\nab^{i_3}K+\sum_{i_1+i_2=i} \nab^{i_1}\q^{i_2+2}+\sum_{i_1+i_2+i_3=i}\nabla^{i_1}\q^{i_2+1}\nab^{i_3}\p\\
&\quad +\sum_{i_1+i_2+i_3=i-1}\frac 1u \nabla^{i_1}\q^{i_2+1}\nab^{i_3}\p+\frac 1u \nab^i\trch.
\end{split}
\end{equation*}

We apply Proposition \ref{evolution lemma} with $\lambda_0=\frac i2$. In particular, since the initial data for $\omega$ vanishes, we can estimate $\|u^{i-1}\nab^{i}\omega\|_{L_u^{\infty}L_{\ub}^{\infty}L^{2}(\S)}$ by the $\|u^{i-1}\cdot\|_{L_{\ub}^{\infty}L_{u}^{1}L^{2}(\S)}$ norm of the right hand side. To estimate each of the terms in the equation, we note that all terms except for the $K$ term and the $\frac 1u \nab^i\trch$ term have been controlled in the proof of Proposition \ref{chih.bd} and \ref{trch.bd}. More precisely, by \eqref{chih.3} and \eqref{chih.4}, we have
\begin{equation*}
\begin{split}
\sum_{i\leq 4} \|u^{i-1}F_i\|_{L^{\infty}_{\ub}L^1_uL^2(\S)}\ls\sum_{i\leq 4} \frac{1}{|u|}\|u^{i}F_i\|_{L^{\infty}_{\ub}L^1_uL^2(\S)}\ls \frac{\de^{\f12}\at b^{\f14}}{|u|^{\f32}}\ls \frac{\de^{\f12} a^{\f34}}{|u|^{\f32}},
\end{split}
\end{equation*}
where $F_i$ is defined to be
\begin{equation*}
\begin{split}
F_i&=\sum_{i_1+i_2=i} \nab^{i_1}\q^{i_2+2}+\sum_{i_1+i_2+i_3=i}\nabla^{i_1}\q^{i_2+1}\nab^{i_3}\p\\
&\quad +\sum_{i_1+i_2+i_3=i-1}\frac 1u \nabla^{i_1}\q^{i_2+1}\nab^{i_3}\p.
\end{split}
\end{equation*}
There are two remaining terms: the term with $K$ and the term $\f{1}{u}\nab^i\trch$. We first estimate the term containing the Gauss curvature $K$. We split up the term into
$$K\ls \left(\K\right)+\frac 1{|u|^2}.$$
For the term with $(\K)$, if $i_2=0$, we have
\begin{equation*}
\begin{split}
&\quad\ \sum_{i\leq 4}\left\|u^{i-1}\nab^{i}\left(\K\right)\right\|_{L_{u}^{1}L^{2}(\S)}\\
&\ls \sum_{i\leq 4}\left\|u^{i+1}\nab^{i}\left(\K\right)\right\|_{L_{u}^2 L^{2}(\S)}\|u^{-2}\|_{L^2_u}\\
&\ls \frac{\de^{\f12}\at b^{\f14}}{|u|^{\f32}}\ls \frac{\de^{\f12} a^{\f34}}{|u|^{\f32}},
\end{split}
\end{equation*}
where in the second inequality, we have used the bootstrap assumption \eqref{BA.3}. For $i_2\geq 1$, we get\footnote{Notice that we have relabeled $i_2$} the better estimate
\begin{equation*}
\begin{split}
&\quad\ \sum_{i\leq 4}\bigg\|\sum_{i_1+i_2+i_3=i-1}u^{i-1}\nab^{i_1}\q^{i_2+1}\nab^{i_3}\left(\K\right)\bigg\|_{L_{u}^{1}L^{2}(\S)}\\
&\ls  \sum_{i_1+i_2\leq 4}\|u^{i_1+i_2+1}\nabla^{i_1}\q^{i_2+1}\|_{L_{u}^{\infty}L^{2}(\S)}\\
&\qquad\times\sum_{i_3\leq 2}\bigg\|u^{i_3+2}\nab^{i_3}\left(\K\right)\bigg\|_{L_{u}^2 L^{\infty}(\S)}\|u^{-3}\|_{L^2_u}\\
&\quad + \sum_{i_1+ i_2 \leq 2}\|u^{i_1+i_2+2}\nabla^{i_1}\q^{i_2+1}\|_{L_{u}^{\infty}L^{\infty}(\S)}\\
&\qquad\times\sum_{i_3\leq 4}\bigg\|u^{i_3+1}\nab^{i_3}\left(\K\right)\bigg\|_{L_{u}^{2}L^{2}(\S)}\|u^{-3}\|_{L^2_u}\\
&\ls \sum_{i_1+i_2\leq 4}\|u^{i_1+i_2+1}\nabla^{i_1}\q^{i_2+1}\|_{L_{u}^{\infty}L^{2}(\S)}\\
&\quad\times\sum_{i_3\leq 4}\bigg\|u^{i_3+1}\nab^{i_3}\left(\K\right)\bigg\|_{L_{u}^2 L^{2}(\S)}\|u^{-3}\|_{L^2_u}\\
&\ls \frac{\de\at b^{\f14}}{|u|^{\frac 52}}\cdot\de^{\f12}\at\M R\ls \frac{\de^{\f12}a^{\f12}}{|u|^{\f32}},
\end{split}
\end{equation*}
where in the second inequality, we have used the Sobolev embedding in Proposition \ref{Sobolev}; in the third inequality, we have applied Proposition \ref{product} and in the last estimate, we have used the bootstrap assumption \eqref{BA.3} and $|u|\geq \de \at b$. 

We then move to the contribution arising from $\frac{1}{|u|^2}$. Notice that for this term, the only possibility for having $i_2=0$ is when $i=0$. In this case, the term can be controlled by
\begin{equation}\label{om.0}
\begin{split}
\left\|u^{-1}\frac{1}{|u|^2}\right\|_{L_{u}^{1}L^{2}(\S)}
\ls &\frac 1{|u|}.
\end{split}
\end{equation}
For $i_2\geq 1$, we have
\begin{equation*}
\begin{split}
&\quad\ \sum_{i\leq 4}\left\|\sum_{i_1+i_2=i-1}u^{i-1}\nab^{i_1}\q^{i_2+1}\frac{1}{|u|^2}\right\|_{L_{u}^{1}L^{2}(\S)}\\
&\ls \sum_{i_1+i_2\leq 3}\|u^{i_1+i_2+1}\nab^{i_1} \q^{i_2+1}\|_{L^{\infty}_uL^2(\S)}\|u^{-3} \|_{L^1_u}\\
&\ls \frac{\de\at b^{\f14}}{|u|^2}\ls \frac{\de a^{\f34}}{|u|^2}, 
\end{split}
\end{equation*}
where in the second inequality, we have used the bootstrap assumption \eqref{BA.1}. 
Combining all the above estimates, we have
\begin{align}\label{om.1}
&\quad\ \sum_{i\leq 4}\left\|\sum_{i_1+i_2+i_3=i-1}u^{i-1}\nab^{i_1}\q^{i_2+1}\nab^{i_3}\left(\K\right)\right\|_{L_{u}^{1}L^{2}(\S)}\\
\notag &\ls \frac{1}{|u|}+\frac{\de^{\frac 12}a^{\frac 34}}{|u|^{\frac 32}}.
\end{align}
Notice moreover that the appearance of the term $\frac{1}{|u|}$ is only due to the contribution of \eqref{om.0} and is only present when $i=0$.
We now estimate the remaining term by:
\begin{equation}\label{om.2}
\begin{split}
\sum_{i\leq 4}\|u^{i-2}\nab^{i}\trch\|_{L_{u}^{1}L^{2}(\S)}\ls \frac{\de a}{|u|^2}+\frac 1{|u|}\ls \f{\delta^{\f12} a^{\f34}}{|u|^{\f32}}+\f1{|u|},
\end{split}
\end{equation}
using the bound for $\nab^i\trch$ in Proposition \ref{trch.bd}. Moreover, for $i\geq 1$, since $\nab^i\trch=\nab^i(\trch-\f2{|u|})$, we have the improved bound
\begin{equation*}
\begin{split}
\sum_{1\leq i\leq 4}\|u^{i-2}\nab^{i}\trch\|_{L_{u}^{1}L^{2}(\S)}\ls \f{\delta^{\f12} a^{\f34}}{|u|^{\f32}}
\end{split}
\end{equation*}
using Proposition \ref{trch.bd}.

Collecting the above estimates, we have
$$\sum_{i\leq 4}\|u^{i-1}\nab^i \om\|_{L^2(\S)}\ls  \frac 1{|u|}+\f{\delta^{\f12} a^{\f34}}{|u|^{\f32}}.$$
Multiplying by $|u|$, we get
$$\sum_{i\leq 4}\|u^{i}\nab^i \om\|_{L^2(\S)}\ls  1+\f{\delta^{\f12} a^{\f34}}{|u|^{\f12}}.$$
Moreover, if $i\neq 0$, the worst terms in \eqref{om.1} and \eqref{om.2} contributing to the bound $1$ are absent. Therefore, we have
$$\sum_{1\leq i\leq 4}\|u^{i}\nab^i \om\|_{L^2(\S)}\ls \f{\delta^{\f12} a^{\f34}}{|u|^{\f12}}.\vspace{-2em}$$
\end{proof}

We summarize the estimates that we have already proved for $\p$ and its derivatives:

\begin{proposition} \label{p.bd}
Under the assumptions of Theorem \ref{main.thm} and the bootstrap assumptions \eqref{BA.1}, \eqref{BA.2}, \eqref{BA.3} and \eqref{BA.4}, we have
$$\sum_{i\leq 4} \frac{1}{a^{\f12}}\|u^i\nab^i\p\|_{L^2(\S)}\ls 1.$$
\end{proposition}

We now estimate the $L^2(\S)$ norms of the first four derivatives of the remaining Ricci coefficients, i.e., the Ricci coefficients that we call $\q$: 
\begin{proposition} \label{q.bd}
Under the assumptions of Theorem \ref{main.thm} and the bootstrap assumptions \eqref{BA.1}, \eqref{BA.2}, \eqref{BA.3} and \eqref{BA.4}, we have
\[
 \sum_{i\leq 4}\frac{1}{\de\at}\|u^{i+1}\nab^i\q\|_{L^2(\S)} \ls 1+\tilde{\M {O}}_{5,2}+\mathcal R.
\]
We also have the more precise bound
\[
 \sum_{i\leq 4}\frac{1}{\de\at}\|u^{i+1}\nab^i\q\|_{L^2(\S)} \ls 1+\f{1}{\de^{\f12}a^{\f12}}\|u^5\nab^5\om\|_{L^\infty_uL^2_{\ub}L^2(\S)}+\mathcal R,
\]
i.e., the only dependence on $\tilde{\M {O}}_{5,2}$ is through the term 
$$\f{1}{\de^{\f12}a^{\f12}}\|u^5\nab^5\om\|_{L^\infty_uL^2_{\ub}L^2(\S)}.$$
\end{proposition}

\begin{proof}
In order to unify the exposition, we will not consider the equation $\nab_3\etab$, but will instead use the equation $\nab_4\nab\log\Omega$. Since $\etab=-\eta+2\nab\log\Omega$, the desired estimate for $\etab$ can be recovered from the bounds for $\eta$ and $\nab\log\Omega$. 

For $\q\in\{\eta,\nab\log\Omega,\trchb+\frac 2u,\chibh,\omb\}$, we have the schematic transport\linebreak  equation
$$\nab_4\q=\beta+\nab\om+K+\nab\etab+\p\q+\q\q+\f1u\p.$$
Notice moreover that by the assumption of Minkowskian data on the initial incoming cone, all of these quantities $\q$ are initially $0$. 

Commuting the above equation with angular derivatives, we get
\begin{equation*}
\begin{split}
\nab_4 \nab^i\q &= \nab^i\beta+\nab^{i+1}(\om,\etab)+\sum_{i_1+i_2+i_3=i} \nabla^{i_1}\q^{i_2}\nab^{i_3}K\\
&\quad +\sum_{i_1+i_2+i_3=i}\nabla^{i_1}\q^{i_2+1}\nab^{i_3}(\p,\q)
+\sum_{i_1+i_2+i_3=i}\frac 1u\nabla^{i_1}\q^{i_2}\nab^{i_3}\p.
\end{split}
\end{equation*}
By Proposition \ref{transport} and the triviality of $\q$ on $\Hb_0$, in order to estimate the quantity $\|u^i\nab^i\q\|_{L_{\ub}^{\infty}L^{2}(\S)}$, it suffices to bound the $\|u^i\cdot\|_{L_{\ub}^{1}L^{2}(\S)}$ norm of the right hand side. We now estimate each of the terms
in the equation. We first control the $\beta$ term 
\begin{equation}\label{q.1}
\sum_{i\leq 4}\|u^i\nab^i\beta\|_{L_{\ub}^{1}L^{2}(\S)}\leq \frac{\delta^{\f12}}{|u|}\sum_{i\leq 4}\|u^{i+1}\nab^i\beta\|_{L_{\ub}^{2}L^{2}(\S)}\ls \f{\delta\at}{|u|}\M R.
\end{equation}
We now turn to the term $\nab^{i+1}\om$. For $i\leq 3$, we apply the estimates in Proposition \ref{p.bd}, while for $i=4$, we use $\tilde{\mathcal O}_{5,2}$ to get
\begin{align}\label{q.2}
&\quad\ \sum_{i\leq 4}\|u^i\nab^{i+1}\om\|_{L_{\ub}^{1}L^{2}(\S)}\\
\notag &\ls \frac{\delta}{|u|}\sum_{i\leq 3}\|u^{i+1}\nab^{i+1}\p\|_{L_{\ub}^{\infty}L^{2}(\S)}+\frac{\delta^{\f12}}{|u|}\|u^5\nab^5\om\|_{L_{\ub}^{2}L^{2}(\S)}\\
\notag &\ls \frac{\delta a^{\frac 12}}{|u|}(1+\tilde{\M {O}}_{5,2}).
\end{align}
For the $\nab^{i+1}\etab$ term, we have
\begin{align}\label{q.3}
&\quad\ \sum_{i\leq 4}\|u^i\nab^{i+1}\etab\|_{L_{\ub}^{1}L^{2}(\S)}\\
\notag &\ls \frac{\delta}{|u|^2}\sum_{i\leq 3}\|u^{i+2}\nab^{i+1}\q\|_{L_{\ub}^{\infty}L^{2}(\S)}+\frac{\delta^{\f12}}{|u|^2}\|u^6\nab^5\etab\|_{L_{\ub}^{2}L^{2}(\S)}\\
\notag &\ls \frac{\delta^2 a^{\frac 12} b^{\f14}}{|u|^2}+\f{\delta^2 a^{\f34} b^{\f14}}{|u|^2}\ls \f{\delta^2 a^{\f34} b^{\f14}}{|u|^2},
\end{align}
where we have used the bootstrap assumptions \eqref{BA.1} and \eqref{BA.3}.

We control the term containing the Gauss curvature $K$ as follows:
\begin{align}\label{q.4}
&\quad\ \sum_{i\leq 4}\left\|\sum_{i_1+i_2+i_3=i} u^i\nab^{i_1}\q^{i_2}\nab^{i_3}K\right\|_{L_{\ub}^{1}L^{2}(\S)}\\
\notag &\ls \frac{\delta^{\f12}}{|u|^2}\sum_{i\leq 4}\left\|u^{i+2}\nab^{i}\left(\K\right)\right\|_{L_{\ub}^2L^{2}(\S)}+\de\|u^{-2}\|_{L^2(\S)}\\
\notag &\quad +\f{\de^{\f12}}{|u|^3}\sum_{i_1+i_2\leq 2}\|u^{i_1+i_2+2}\nab^{i_1}\q^{i_2+1}\|_{L^\infty_{\ub}L^{\infty}(\S)}\\
\notag &\qquad\times\sum_{i_3\leq 4}\bigg\|u^{i_3+2}\nab^{i_3}\left(\K\right)\bigg\|_{L^2_{\ub}L^2(\S)}\\
\notag &\quad +\f{\de^{\f12}}{|u|^3}\sum_{i_1+i_2\leq 4}\|u^{i_1+i_2+1}\nab^{i_1}\q^{i_2+1}\|_{L^\infty_{\ub}L^2(\S)}\\
\notag &\qquad\times\Bigg(\sum_{i_3\leq 2}\bigg\|u^{i_3+1}\nab^{i_3}\left(\K\right)\bigg\|_{L^2_{\ub}L^{\infty}(\S)}+\de^{\f12}|u|\Bigg)\\
\notag &\ls \frac{\delta^2 a^{\frac 34} b^{\f14}}{|u|^2}+\f{\delta}{|u|}+\f{\de^3 a^{\f54} b^{\f12}}{|u|^3}+\f{\de^2\at b^{\f14}}{|u|^2}\\
\notag &\ls \frac{\delta^2 a^{\frac 34} b^{\f14}}{|u|^2}+\f{\delta}{|u|},
\end{align}
where in the second inequality we have used the bootstrap assumptions \eqref{BA.1} and \eqref{BA.3}, and in the final inequality we have used $|u|\geq \de\at b$ and $b\leq a$. Moreover, notice that the contribution for $\f{\delta}{|u|}$ comes from only from the term $i=0$. For $i\geq 1$, since $\nab^i K=\nab^i(K-\frac 1{|u|^2})$, we have the improved bound
\begin{equation}\label{q.5}
\begin{split}
\sum_{1\leq i\leq 4}\left\|\sum_{i_1+i_2+i_3=i} u^i\nab^{i_1}\q^{i_2}\nab^{i_3}K \right\|_{L_{\ub}^{1}L^{2}(\S)}
\ls \frac{\delta^2 a^{\frac 34} b^{\f14}}{|u|^2}.
\end{split}
\end{equation}

The fourth term can be bounded as follows:
\begin{align}\label{q.6}
&\quad\ \sum_{i\leq 4}\left\|\sum_{i_1+i_2+i_3=i}u^i\nabla^{i_1}\q^{i_2+1}\nab^{i_3}(\p,\q)\right\|_{L_{\ub}^{1}L^{2}(\S)}\\
\notag &\ls  \f{\de}{|u|^2}\sum_{i_1+i_2\leq 4}\|u^{i_1+i_2+1}\nabla^{i_1}\q^{i_2+1}\|_{L_{\ub}^{\infty}L^{2}(\S)}\\
\notag &\qquad \times\sum_{i_3\leq 2}\|u^{i_3+1}\nab^{i_3}(\p,\q)\|_{L_{\ub}^{\infty}L^{\infty}(\S)}\\
\notag &\quad + \f{\de}{|u|^2}\sum_{i_1+ i_2\leq 2}\|u^{i_1+i_2+2}\nabla^{i_1}\q^{i_2+1}\|_{L_{\ub}^{\infty}L^{\infty}(\S)}\\
\notag &\qquad\times\sum_{i_3\leq 4}\|u^{i_3}\nab^{i_3}(\p,\q)\|_{L_{\ub}^{\infty}L^{2}(\S)}\\
\notag &\ls \frac{\de^2 a b^{\frac 14}}{|u|^2}+\f{\de^3 a b^{\f12}}{|u|^3}\ls \frac{\de^2 a b^{\frac 14}}{|u|^2},
\end{align}
where in the last line we have used Proposition \ref{p.bd} to control $\p$ and used Proposition \ref{product} to control the product of $\q$, as well as using Sobolev embedding in Proposition \ref{Sobolev}.
We then control the final term as follows:
\begin{align}\label{q.7}
&\quad\ \sum_{i\leq 4}\left\|\sum_{i_1+i_2+i_3=i} u^{i-1}\nabla^{i_1}\q^{i_2}\nab^{i_3}\p\right\|_{L_{\ub}^{1}L^{2}(\S)}\\
\notag &\ls  \f{\delta}{|u|}\sum_{i\leq 4}\|u^i\nab^i\p\|_{L^{\infty}_{\ub}L^2(\S)}\\
\notag &\quad +\frac{\de}{|u|^2}\sum_{i_1+i_2\leq 3}\|u^{i_1+i_2+1}\nabla^{i_1}\q^{i_2+1}\|_{L_{\ub}^{\infty}L^{2}(\S)}\sum_{i_3\leq 2}\|u^{i_3+1}\nab^{i_3}\p\|_{L_{\ub}^{\infty}L^{\infty}(\S)}\\
\notag &\quad + \frac{\de}{|u|^2}\sum_{i_1+ i_2\leq 2}\|u^{i_1+i_2+2}\nabla^{i_1}\q^{i_2+1}\|_{L_{\ub}^{\infty}L^{\infty}(\S)}\sum_{i_3\leq 4}\|u^{i_3}\nab^{i_3}\p\|_{L_{\ub}^{\infty}L^{2}(\S)}\\
\notag &\ls \f{\de \at}{|u|}+\frac{\de^2 a b^{\frac 14}}{|u|^2},
\end{align}
where as before we have used Proposition \ref{p.bd} to control $\p$ and used Proposition \ref{product} to estimate the product of $\q$, as well as using Sobolev embedding in Proposition \ref{Sobolev}.

Therefore, applying Proposition \ref{transport}, using the estimates above and using $|u|\geq \de\at b$, we get 
$$\sum_{i\leq 4}\|u^i\nab^i\q\|_{L^{\infty}_{\ub}L^2(\S)}\ls \f{\delta\at}{|u|}(1+\tilde{\M {O}}_{5,2}+\mathcal R).$$
Moreover, notice that the term $\tilde{\M O}_{5,2}$ only comes from controlling $\nab^5\om$ in~\eqref{q.2}. Therefore, we have the improved estimate
\[
 \sum_{i\leq 4}\|u^{i}\nab^i\q\|_{L^2(\S)} \ls \f{\de\at}{|u|} \left(1+\f{1}{\de^{\f12}a^{\f12}}\|u^5\nab^5\om\|_{L^\infty_uL^2_{\ub}L^2(\S)}+\mathcal R\right).\vspace{-1.5em}
\]
\end{proof}

From the proof of Proposition \ref{q.bd}, we also get the following additional bound for $\nab^i(\trchb+\f{2}{|u|})$ for $i\geq 1$:

\begin{proposition} \label{trchb.bd}
Under the assumptions of Theorem \ref{main.thm} and the bootstrap assumptions \eqref{BA.1}, \eqref{BA.2}, \eqref{BA.3} and \eqref{BA.4}, we have
\[
 \sum_{1\leq i\leq 4}\left\|u^i\nab^i \left(\trchb+\f2{|u|}\right)\right\|_{L^2(\S)} \ls \frac{\de^{\f32} a^{\f34}}{|u|^{\f32}}.
\]

\end{proposition}

\begin{proof}
Notice that $\trchb+\f2{|u|}$ satisfies the equation
$$\nab_4\left(\trchb+\f2{|u|}\right)=K+\nab\etab+\p\q+\q\q+\f1u(\trch,\om).$$
This equation can be derived from the usual equation for $\nab_4\trchb$ and noting that $\nab_4 u=0$.
In particular, compared to the equation for a general component $\q$, we do not have the terms $\beta$, $\nab\om$ and $\f1u \chih$.

Commuting the above equation with angular derivatives, we get
\begin{equation*}
\begin{split}
 \nab_4 \nab^i \left(\trchb+\f2{|u|}\right) &= \nab^{i+1}\etab+\sum_{i_1+i_2+i_3=i} \nabla^{i_1}\q^{i_2}\nab^{i_3}K\\
&\quad +\sum_{i_1+i_2+i_3=i}\nabla^{i_1}\q^{i_2+1}\nab^{i_3}(\p,\q)\\
&\quad +\sum_{i_1+i_2+i_3=i}\frac 1u\nabla^{i_1}\q^{i_2}\nab^{i_3}(\trch,\om).
\end{split}
\end{equation*}
We now revisit the proof of Proposition \ref{q.bd}. Notice that except for the terms in \eqref{q.1}, \eqref{q.2}, \eqref{q.4} and \eqref{q.7}, the error terms in the proof of Proposition \ref{q.bd} satisfy the better bound $\f{\de^2 a b^{\f14}}{|u|^2}$. Since there are no $\nab^i\beta$ and $\nab^{i+1}\om$ terms in the equation for $\nab^i(\trchb+\f2{|u|})$, the contributions of $\f{\de\at}{|u|}(1+\tilde{\M O}_{3,2}+\M R)$ from \eqref{q.1} and \eqref{q.2} are absent. Moreover, since we have $i\geq 1$, the term $\f{\de}{|u|}$ from \eqref{q.4} is also absent (and we have the bound \eqref{q.5} instead). We now revisit the estimate in \eqref{q.7}, using the fact that we only have $\trch$ and $\om$ and that $\chih$ is absent:
\begin{equation*}
\begin{split}
&\quad\ \sum_{1\leq i\leq 4}\left\|\sum_{i_1+i_2+i_3=i} u^{i-1}\nabla^{i_1}\q^{i_2}\nab^{i_3}(\trch,\om)\right\|_{L_{\ub}^{1}L^{2}(\S)}\\
&\ls  \f{\delta}{|u|}\sum_{1\leq i\leq 4}\|u^i\nab^i(\trch,\om)\|_{L^{\infty}_{\ub}L^2(\S)}\\
&\quad +\frac{\de}{|u|^2}\sum_{i_1+i_2\leq 4}\|u^{i_1+i_2+1}\nabla^{i_1}\q^{i_2+1}\|_{L_{\ub}^{\infty}L^{2}(\S)}\\
&\qquad\times\sum_{i_3\leq 2}\|u^{i_3+1}\nab^{i_3}(\trch,\om)\|_{L_{\ub}^{\infty}L^{\infty}(\S)}\\
&\quad + \frac{\de}{|u|^2}\sum_{i_1+ i_2\leq 2}\|u^{i_1+i_2}\nabla^{i_1+2}\q^{i_2+1}\|_{L_{\ub}^{\infty}L^{\infty}(\S)}\\
&\qquad\times\sum_{i_3\leq 4}\|u^{i_3}\nab^{i_3}(\trch,\om)\|_{L_{\ub}^{\infty}L^{2}(\S)}\\
&\ls \f{\de^2 a}{|u|^2}+\f{\de^{\f32} a^{\f34}}{|u|^{\f32}}+\frac{\de^2 a b^{\frac 14}}{|u|^2}\ls \f{\de^{\f32} a^{\f34}}{|u|^{\f32}},
\end{split}
\end{equation*}
where in addition to using Proposition \ref{product}, we have used the improved bounds for $\nab^i\trch$ and $\nab^i\om$ derived in Propositions \ref{trch.bd} and \ref{om.bd} respectively. Combining this with \eqref{q.3}, \eqref{q.5} and \eqref{q.6}, we obtain the desired conclusion.
\end{proof}
Besides $\nab^i(\trchb+\f{2}{|u|})$, we also need an improved bound for $\nab^i\log\Om$ for $1\leq i\leq 4$:
\begin{proposition} \label{Om.bd}
Under the assumptions of Theorem \ref{main.thm} and the bootstrap assumptions \eqref{BA.1}, \eqref{BA.2}, \eqref{BA.3} and \eqref{BA.4}, we have
\[
 \sum_{1\leq i\leq 4}\|u^{i-1}\nab^{i}\log\Om\|_{L^\infty_u L^\infty_{\ub}L^2(\S)} \ls \f{\de^{\f32}a^{\f34}}{|u|^{\f32}}.
\]
\end{proposition}
\begin{proof}
Recall that
$$\nab_4\nab^i\log\Om=\nab^i\om+ \sum_{i_1+i_2+i_3=i-1}\nab^{i_1}\q^{i_2+1}\nab^{i_3}\p.$$
We would like to control $\nab^i\log\Om$ in the norm $\|u^{i-1}\nab^{i}\log\Om\|_{L^\infty_u L^\infty_{\ub}L^2(\S)}$. By Proposition \ref{transport}, we thus have to estimate the right hand side in the norm $\|u^{i-1}\cdot\|_{L^\infty_u L^1_{\ub}L^2(\S)}$. The first term can be bounded with the estimate for $\nab^i\om$ for $1\leq i\leq 4$ in Proposition \ref{om.bd}. We also need to take advantage of the improved bound that we achieve for $i\neq 0$. More precisely, we have
\begin{equation*}
\begin{split}
\sum_{1\leq i\leq 4}\|u^{i-1}\nab^i\om\|_{L^\infty_u L^1_{\ub}L^2(\S)}\ls \f{\de^{\f32} a^{\f34}}{|u|^{\f32}}.
\end{split}
\end{equation*}
The second term can be controlled with the aid of \eqref{q.6}:
\begin{equation*}
\begin{split}
&\quad\ \sum_{1\leq i\leq 4}\left\|\sum_{i_1+i_2+i_3=i-1}u^{i-1}\nab^{i_1}\q^{i_2+1}\nab^{i_3}\p\right\|_{L^\infty_u L^1_{\ub}L^2(\S)}\\
&= \sum_{i\leq 3}\left\|\sum_{i_1+i_2+i_3=i}u^{i}\nab^{i_1}\q^{i_2+1}\nab^{i_3}\p\right\|_{L^\infty_u L^1_{\ub}L^2(\S)}\\
&\ls \f{\de^2 a^{\f34}b^{\f14}}{|u|^2}\ls \f{\de^{\f32} a^{\f34}}{|u|^{\f32}}
\end{split}
\end{equation*}
since $|u|\geq \delta\at b$.
\end{proof}

We also note that the $\nab_4$ equation for $K-\f{1}{|u|^2}$ contains exactly the same type of terms as the $\nab_4$ equation for $\nab\q$. Therefore, $\nab^i(\K)$ obeys the same estimates as $\nab^{i+1}\q$ for $i\leq 3$. More precisely, we have
\begin{proposition} \label{K.bd}
Under the assumptions of Theorem \ref{main.thm} and the bootstrap assumptions \eqref{BA.1}, \eqref{BA.2}, \eqref{BA.3} and \eqref{BA.4}, we have
\[
 \sum_{i\leq 3}\left\|u^{i+2}\nab^{i}\left(\K\right)\right\|_{L^\infty_uL^\infty_{\ub}L^2(\S)} \ls \de a^{\f12} (1+\tilde{\M O}_{5,2}+\M R).
\]
In particular,
\[
 \sum_{i\leq 3}\left\|u^{i+1}\nab^{i}\left(\K\right)\right\|_{L^\infty_uL^\infty_{\ub}L^2(\S)} \ls \f{1}{b^{\f34}}.
\]
\end{proposition}
\begin{proof}
As mentioned above, by repeating the proof of Proposition \ref{q.bd}, we obtain the bound 
\[
 \sum_{i\leq 3}\frac{1}{\de\at}\left\|u^{i+2}\nab^{i}\left(\K\right)\right\|_{L^\infty_uL^\infty_{\ub}L^2(\S)} \ls 1+\tilde{\M O}_{5,2}+\M R.
\]
Multiplying by $\f{\delta\at}{|u|}$, we obtain
$$\sum_{i\leq 3}\left\|u^{i+1}\nab^{i}\left(\K\right)\right\|_{L^\infty_uL^\infty_{\ub}L^2(\S)} \ls \f{\delta\at b^{\f14}}{|u|}\ls \f{1}{b^{\f34}}.\vspace{-2em}$$
\end{proof}

\section{Elliptic estimates for the fifth derivatives of the Ricci coefficients}\label{secelliptic}
We now estimate the fifth angular derivatives of the Ricci coefficients. As in \cite{Chr-Kl}, \cite{KNI:book}, \cite{Chr:book} and \cite{KR:Trapped}, this is achieved by a combination of transport estimates and elliptic estimates. This is because curvature enters as source terms in the transport equations for the Ricci coefficients. As a result, with the bounds for four derivatives of the curvature components, we can only obtain estimates for the four derivatives of the Ricci coefficients. The standard approach is therefore to control some chosen linear combination of some derivatives of the Ricci coefficients and curvature components via transport estimates. We then recover the bounds for the remaining highest derivatives of the Ricci coefficients by $L^2$ elliptic estimates on $\S$.

We now begin with the estimates for $\trch$ and $\chih$:

\begin{proposition}\label{chi.5.bd}
Under the assumptions of Theorem \ref{main.thm} and the bootstrap assumptions \eqref{BA.1}, \eqref{BA.2}, \eqref{BA.3} and \eqref{BA.4}, we have
$$\|u^6\nab^5\trch\|_{L^2_{\ub}L^2(\S)}\ls  \de^{\f32}a b^{\f14}$$
and
$$\|u^5\nab^5\chih\|_{L^2_{\ub}L^2(\S)}\ls  \de^{\f12}\at(1+\mathcal R).$$

\end{proposition}

\begin{proof}
Consider the following equation:
$$\nab_4 \tr \chi+\f12 (\tr\chi)^2=-|\chih|^2-2\omega\tr\chi$$ 
Commuting with angular derivatives for $i$ times, we get
\begin{equation*}
\begin{split}
\nab_4 \nab^{5} \tr\chi
=&\sum_{i_1+i_2+i_3+i_4=5} \nab^{i_1}\q^{i_2}\nab^{i_3}\p\nab^{i_4}\p.
\end{split}
\end{equation*}

We will apply Proposition \ref{transport} to estimate $\|u^6\nab^5\tr\chi\|_{L_{\ub}^{\infty}L_{u}^{\infty}L^{2}(\S)}$ by the $\|u^6\cdot\|_{L_{u}^{\infty}L_{\ub}^{1}L^{2}(\S)}$ norm of the right hand side. We separately estimate the terms with $5$ derivatives on $\p$ and the remaining terms.
\begin{equation*}
\begin{split}
&\quad\ \left\|u^6\sum_{i_1+i_2+i_3+i_4=5} \nab^{i_1}\q^{i_2}\nab^{i_3}\p\nab^{i_4}\p\right\|_{L^{\infty}_u L^1_{\ub}L^2(\S)}\\
&\ls  \de^{\f12} \|u\p\|_{L^\infty_{\ub}L^\infty(\S)}\|u^5\nab^5\p\|_{L^{\infty}_uL^2_{\ub}L^2(\S)}\\
&\quad +\de\sum_{i_1\leq 2}\|u^{i+1}\nab^i\p\|_{L^{\infty}_uL^{\infty}_{\ub}L^\infty(\S)}\sum_{i_2\leq 4}\|u^i\nab^i\p\|_{L^{\infty}_uL^{\infty}_{\ub}L^2(\S)}\\
&\quad +\de\sum_{i_1+i_2\leq 2}\|u^{i_1+i_2+2} \nab^{i_1}\q^{i_2+1}\|_{L^{\infty}_uL^\infty_{\ub}L^\infty(\S)}\\
&\qquad\times\sum_{i_3\leq 2}\|u^{i_3+1}\nab^{i_3}\p \|_{L^\infty_{\ub}L^\infty(\S)}\sum_{i_4\leq 4}\|u^{i_4}\nab^{i_4}\p \|_{L^{\infty}_uL^\infty_{\ub}L^2(\S)}\\
&\quad +\de\sum_{i_1+i_2\leq 4}\|u^{i_1+i_2+1} \nab^{i_1}\q^{i_2+1}\|_{L^{\infty}_uL^\infty_{\ub}L^2(\S)}\\
&\qquad\times\sum_{i_3\leq 2}\|u^{i_3+1}\nab^{i_3}\p \|_{L^{\infty}_uL^\infty_{\ub}L^\infty(\S)}\sum_{i_4\leq 2}\|\nab^{i_4+1}\p \|_{L^{\infty}_uL^\infty_{\ub}L^\infty(\S)}\\
&\ls  \de a b^{\f14},
\end{split}
\end{equation*}
where we have used Proposition \ref{product} to control the product of $\q$ and Proposition \ref{p.bd} to bound $\p$.

Recall that $\nab^5\trch=0$ initially on $\Hb_0$. Therefore, we have
$$\||u|^6\nab^5\trch\|_{L^\infty_{\ub}L^2(\S)}\ls \de a b^{\f14}.$$
Taking $L^2$ in $\ub$, we get
\begin{equation}\label{trch.5.bd}
\||u|^6\nab^5\trch\|_{L^2_{\ub}L^2(\S)}\ls \de^{\f32} a b^{\f14}.
\end{equation}
Recalling the schematic form of the Codazzi equation
$$\div\chih-\frac 12\nab\trch+\beta=\p\q,$$
we can apply the elliptic estimates in Proposition \ref{elliptictraceless} to obtain
\begin{equation*}
\begin{split}
\|u^5\nab^5\chih\|_{L^2(\S)} &\ls  \sum_{i\leq 5}\|u^i\nab^i\tr\chi\|_{L^2(\S)}+\sum_{i\leq 4}\|u^{i+1}\nab^i\beta\|_{L^2(\S)}\\
&\quad +\sum_{i\leq 4}\sum_{i_1+i_2=i}\|u^{i+1}\nab^{i_1}\q\nab^{i_2}\p\|_{L^2(\S)}\\
&\quad +\sum_{i\leq 4}\|u^i\nab^i\chih\|_{L^2(\S)}.
\end{split}
\end{equation*}
Taking $L^2$ in $\ub$, using the control of $\p$ from Proposition \ref{p.bd}, the bound for $\q$ in \eqref{BA.1}, as well as the bound \eqref{trch.5.bd} for $\nab^5\trch$ that we have just achieved above, we obtain
\begin{equation*}
\begin{split}
&\quad\ \|u^5\nab^5\chih\|_{L^2_{\ub}L^2(\S)}\\
&\ls  \sum_{i\leq 5}\|u^i\nab^i\tr\chi\|_{L^2_{\ub}L^2(\S)}+\sum_{i\leq 4}\|u^{i+1}\nab^i\beta\|_{L^2_{\ub}L^2(\S)}\\
&\quad +\sum_{i\leq 4}\sum_{i_1+i_2=i}\|u^{i+1}\nab^{i_1}\q\nab^{i_2}\p\|_{L^2_{\ub}L^2(\S)}+\sum_{i\leq 4}\|u^i\nab^i\chih\|_{L^2_{\ub}L^2(\S)}\\
&\ls  \f{\de^{\f32}ab^{\f14}}{|u|}+\de^{\frac 12}\at+\de^{\f12}\at\mathcal R\\
&\ls  \de^{\f12}\at(1+\mathcal R),
\end{split}
\end{equation*}
since $|u|\geq \delta \at b$.
\end{proof}
We now turn to the estimates for the highest derivative of $\om$. $\nab^5\om$ obeys the same bounds as $\nab^5\chih$. However, the proof will proceed in a slightly different manner as we need to control a transport equation in the $\nab_3$ direction.

\begin{proposition} \label{om.5.bd}
Under the assumptions of Theorem \ref{main.thm} and the bootstrap assumptions \eqref{BA.1}, \eqref{BA.2}, \eqref{BA.3} and \eqref{BA.4}, we have
$$\|u^5\nab^5\omega\|_{L^{\infty}_uL^2_{\ub}L^2(\S)}\ls \de^{\f12}\at(1+\M R).$$

\end{proposition}

\begin{proof}
Define $\od$ to be the solution to 
$$\nab_3 \od=\f 12\sigmac,$$
with zero initial data on $H_1$ and let
$$\kappa:=\nab\omega+^*\nab\od-\f12\beta.$$
It is easy to see using Proposition \ref{evolution lemma} and the bootstrap assumptions \eqref{BA.1} and \eqref{BA.3} that we have
$$\sum_{i\leq 4}\|u^{i}\nab^i\od\|_{L^\infty_uL^\infty_{\ub}L^2(\S)}\ls 1+\f{\de^{\f12} a^{\f34}}{|u|^{\f12}}\ls \at.$$
The proof of this estimate is similar to that for $\om$ in Proposition \ref{om.bd}. Hence for $\od$ and its first four derivatives, $\od$ obeys the same estimates as $\p$ as given in Proposition \ref{p.bd}. Therefore, in the remainder of the proof of this proposition, we will include $\od$ as one of the $\p$ terms and use the notation $\p\in\{\trch,\chih,\om,\od\}$.

With this modified notation, $\kappa$ obeys the following transport equation
\begin{equation*}
\begin{split}
\nab_3\kappa+\f12 \tr\chib\kappa &=\sum_{i_1+i_2+i_3=1}\q^{i_1}\nab^{i_2}\left(\trchb+\f2{|u|},\chibh,\omb\right)\nab^{i_3}\psi\\
&\quad +\sum_{i_1+i_2=1}\q^{i_1+1}\nab^{i_2}(\eta,\etab)+\frac{1}{|u|}\q\p+\frac{1}{|u|}\beta+\q K.
\end{split}
\end{equation*}
Commuting with angular derivatives for $4$ times, and applying the schematic Codazzi equation for $\beta$:
$$\beta=\sum_{i_1+i_2=1}\q^{i_1}\nab^{i_2}\p,$$
we have
\begin{equation*}
\begin{split}
\nab_3 \nab^4 \kappa+\f{5}{2}\tr\chib \nab^4 \kappa &=\sum_{i_1+i_2+i_3+i_4=5} \nab^{i_1}\q^{i_2}\nab^{i_3} \left(\trchb+\f2{|u|},\chibh,\omb\right)\nab^{i_4}\p\\
&\quad +\sum_{i_1+i_2+i_3=5}\nab^{i_1}\q^{i_2+1}\nab^{i_3}(\eta,\etab)\\
&\quad +\frac{1}{|u|}\nab^4\beta+\frac{1}{|u|}\sum_{i_1+i_2+i_3=4} \nab^{i_1}\q^{i_2+1}\nab^{i_3}\p\\
&\quad +\sum_{i_1+i_2+i_3+i_4=4} \nab^{i_1}\q^{i_2}\nab^{i_3}\q\nab^{i_4}K.
\end{split}
\end{equation*}
We estimate $\nab^4\kappa$ using Proposition \ref{evolution lemma}. Since we will only need to estimate $\nab^4\kappa$ after integrated in $\ub$, we will directly control it in the 
$$\|u^4\nab^4\kappa\|_{L^2_{\ub}L_u^{\infty}L^{2}(\S)}$$
norm. Applying Proposition \ref{evolution lemma} with $\lambda_0=\f52$, it suffices to estimate
the initial data and the $\|u^4\cdot\|_{L_{\ub}^{2}L_{u}^{1}L^{2}(\S)}$ norm of the right hand side. We now estimate each of the terms
in the equation. 

For the term $\sum_{i_1+i_2+i_3+i_4=5} \nab^{i_1}\q^{i_2}\nab^{i_3}(\trchb+\f2{|u|},\chibh,\omb)\nab^{i_4}\p$, we estimate separately the contributions with the highest order derivatives. These include the case where there are $5$ derivatives on $(\trchb+\f2{|u|},\chibh,\omb)$ and the case where there are $5$ derivatives on $\p$. In the first case, we have
\begin{equation*}
\begin{split}
&\quad\ \|u^4\nab^5\left(\trchb+\f2{|u|},\chibh,\omb\right)\p\|_{L^2_{\ub}L_{u}^{1}L^{2}(\S)}\\
&\ls \de^{\f12}\|u\p\|_{L^\infty_{\ub}L^\infty_u L^\infty(\S)}\|u^{-2}\|_{L^2_u}\|u^5\nab^5\\
&\quad\times \left(\trchb+\f2{|u|},\chibh,\omb\right)\|_{L^\infty_{\ub}L_{u}^{2}L^{2}(\S)}\\
&\ls  \de^{\f12} \cdot a^{\f12}\cdot\frac{1}{|u|^{\f32}}\cdot\frac{\de a^{\f14}b^{\f14}}{|u|^{\f12}}\ls \f{\de^{\f32} ab^{\f14}}{|u|^2},
\end{split}
\end{equation*}
where we have used the estimate in Proposition \ref{p.bd} and the bootstrap assumption \eqref{BA.3}.

In the second case, we have
\begin{equation*}
\begin{split}
&\quad\ \left\|u^4\left(\trchb+\f2{|u|},\chibh,\omb\right)\nab^5\p\right\|_{L^2_{\ub}L_{u}^{1}L^{2}(\S)}\\
&\ls \|u^2\q\|_{L^\infty_{\ub}L^\infty_u L^\infty(\S)}\|u^{-3}\|_{L^1_u}\\
&\quad\times(\|u^5\nab^5(\trch,\chih,\om)\|_{L^\infty_uL_{\ub}^{2}L^{2}(\S)}+\|u^5\nab^5\od\|_{L^\infty_uL_{\ub}^{2}L^{2}(\S)})\\
&\ls  \frac{\de a^{\f12}b^{\f14}\cdot\de^{\f12} a^{\f12} b^{\f14}}{|u|^2} \left(1+\f{1}{\delta^{\f12}\at}\|u^5\nab^5\od\|_{L^\infty_uL_{\ub}^{2}L^{2}(\S)}\right)\\
&\ls  \frac{\de^{\f32} a b^{\f12}}{|u|^2}\left(1+\f{1}{\delta^{\f12}\at}\|u^5\nab^5\od\|_{L^\infty_uL_{\ub}^{2}L^{2}(\S)}\right),
\end{split}
\end{equation*}
where we have used the bootstrap assumptions \eqref{BA.1} and \eqref{BA.3}.

For the lower order term where each factor has at most $4$ derivatives, we have

\begin{align}\label{kappa.lower}
&\quad\ \sum_{\substack{i_1+i_2+i_3=5\\i_1,i_3\leq 4}}\|u^4\nab^{i_1}\q^{i_2+1}\nab^{i_3}\p\|_{L^2_{\ub}L_{u}^{1}L^{2}(\S)}\\
\notag &\ls \de^{\f12} \sum_{\substack{i_1+i_2\leq 5\\i_1\leq 2}}\|u^{i_1+i_2+2}\nab^{i_1}\q^{i_2+1}\|_{L^\infty_{\ub}L^\infty_u L^\infty(\S)}\\
\notag &\quad\times\sum_{i_3\leq 4}\|u^{i_3}\nab^{i_3}\p\|_{L^\infty_uL_{\ub}^{\infty}L^{2}(\S)}\|u^{-3}\|_{L^1_u}\\
\notag &\quad + \de^{\f12} \sum_{\substack{i_1+i_2\leq 5\\i_1\leq 4}}\|u^{i_1+i_2+1}\nab^{i_1}\q^{i_2+1}\|_{L^\infty_{\ub}L^\infty_u L^2(\S)}\\
\notag &\qquad\times\sum_{i_3\leq 2}\|u^{i_3+1}\nab^{i_3}\p\|_{L^\infty_uL_{\ub}^{\infty}L^{\infty}(\S)}\|u^{-3}\|_{L^1_u}\\
\notag &\ls \f{\de^{\f32} a b^{\f14}}{|u|^2},
\end{align}
where we have used Propositions \ref{product} and \ref{p.bd}.

We now move to the term $\sum_{i_1+i_2+i_3=5}\nab^{i_1}\q^{i_2+1}\nab^{i_3}(\eta,\etab)$. As for the previous term, we will first control the contribution where one of the Ricci coefficients has $5$ derivatives. More precisely, we have
\begin{equation*}
\begin{split}
&\quad\ \|u^4\q\nab^5(\eta,\etab)\|_{L^2_{\ub}L_{u}^{1}L^{2}(\S)}\\
&\ls \|u^2\q\|_{L^\infty_{\ub}L^\infty_u L^\infty(\S)}\|u^{-4}\|_{L^1_u}\|u^6\nab^5(\eta,\etab)\|_{L^\infty_u L_{\ub}^{2}L^{2}(\S)}\\
&\ls  \frac{\de\at b^{\f14}\cdot\de^{\f32} a^{\f34} b^{\f14}}{|u|^3}\ls\frac{\de^{\f52}a^{\f54} b^{\f12}}{|u|^3} ,
\end{split}
\end{equation*}
using the bootstrap assumptions \eqref{BA.1} and \eqref{BA.3}.

For the lower order term, since $\nab^i\q$ satisfies all the estimates that $\nab^i\p$ for $i\leq 4$, we can follow the proof of \eqref{kappa.lower} to get
\begin{equation*}
\begin{split}
\sum_{\substack{i_1+i_2+i_3=5\\i_1,i_3\leq 4}}\|u^4\nab^{i_1}\q^{i_2+1}\nab^{i_3}\q\|_{L^2_{\ub}L_{u}^{1}L^{2}(\S)}
\ls \f{\de^{\f32} a b^{\f14}}{|u|^2}.
\end{split}
\end{equation*}
We now move to the third term, which is the term containing $\beta$:
\begin{equation*}
\begin{split}
\|u^3\nab^4\beta\|_{L^2_{\ub}L_{u}^{1}L^{2}(\S)}
\ls \|u^{-2}\|_{L^1_u}\|u^5\nab^4\beta\|_{L^\infty_uL^2_{\ub}L^2(\S)}
\ls \f{\de^{\f12} \at}{|u|}\M R.
\end{split}
\end{equation*}
For the fourth term, we have the estimate:
\begin{equation*}
\begin{split}
&\quad\ \sum_{i_1+i_2+i_3=4}\|u^3\nab^{i_1}\q^{i_2+1}\nab^{i_3}\p\|_{L^2_{\ub}L_{u}^{1}L^{2}(\S)}\\
&\ls \de^{\f12} \sum_{i_1+i_2\leq 2}\|u^{i_1+i_2+2}\nab^{i_1}\q^{i_2+1}\|_{L^\infty_{\ub}L^\infty_u L^\infty(\S)}\\
&\quad\times \sum_{i_3\leq 4}\|u^{i_3}\nab^{i_3}\p\|_{L^\infty_uL_{\ub}^{\infty}L^{2}(\S)}\|u^{-3}\|_{L^1_u}\\
&\quad + \de^{\f12} \sum_{i_1+i_2\leq 4}\|u^{i_1+i_2+1}\nab^{i_1}\q^{i_2+1}\|_{L^\infty_{\ub}L^\infty_u L^2(\S)}\\
&\qquad\times\sum_{i_3\leq 2}\|u^{i_3+1}\nab^{i_3}\p\|_{L^\infty_uL_{\ub}^{\infty}L^{\infty}(\S)}\|u^{-3}\|_{L^1_u}\\
&\ls \f{\de^{\f32} a b^{\f14}}{|u|^2},
\end{split}
\end{equation*}
using the bounds in Propositions \ref{product} and \ref{p.bd}.

Finally, for the fifth term, i.e., the term containing $K$, we have
\begin{equation*}
\begin{split}
&\quad\ \sum_{i_1+i_2+i_3=4}\|u^4\nab^{i_1}\q^{i_2+1}\nab^{i_3}K\|_{L^2_{\ub}L_{u}^{1}L^{2}(\S)}\\
&\ls \sum_{i_1+i_2\leq 2}\|u^{i_1+i_2+2}\nab^{i_1}\q^{i_2+1}\|_{L^\infty_{\ub}L^\infty_u L^\infty(\S)}\\
&\quad\times\sum_{i_3\leq 4}\|u^{i_3+2}\nab^{i_3}\left(K-\frac 1{|u|^2}\right)\|_{L^\infty_uL_{\ub}^2L^{2}(\S)}\|u^{-4}\|_{L^1_u}\\
&\quad + \sum_{i_1+i_2\leq 4}\|u^{i_1+i_2+1}\nab^{i_1}\q^{i_2+1}\|_{L^\infty_{\ub}L^\infty_u L^2(\S)}\\
&\qquad\times\left(\sum_{i_3\leq 2}\|u^{i_3+3}\nab^{i_3}\left(K-\frac 1{|u|^2}\right)\|_{L^\infty_uL_{\ub}^2L^{\infty}(\S)}\|u^{-4}\|_{L^1_u}+\de^{\f12}\|u^{-3}\|_{L^1_u}\right)\\
&\ls \de a^{\f12}b^{\f14}\left(\frac{\de^{\f32}a^{\f34}b^{\f14}}{|u|^3}+\frac{\de^{\f12}}{|u|^2}\right)\ls \f{\de^{\f32}a^{\f34}}{|u|^2},
\end{split}
\end{equation*}
using Proposition \ref{product} and the bootstrap assumption \eqref{BA.3}.

Combining the above estimates and using the bootstrap assumption \eqref{BA.3}, we have
$$\|u^4\nab^4\kappa\|_{L^2_{\ub}L^{\infty}_uL^2(\S)}\ls \frac{\de^{\f12}\at}{|u|}(1+\M R)+\frac{\de^{\f12} a^{\f12}}{b^{\f12}|u|}\f{1}{\delta^{\f12}\at}\|u^5\nab^5\od\|_{L^\infty_uL_{\ub}^{2}L^{2}(\S)},$$
which implies
$$\|u^5\nab^4\kappa\|_{L^2_{\ub}L^{\infty}_uL^2(\S)}\ls \de^{\f12}\at\left(1+\M R+\f{1}{\de^{\f12}a^{\f12}b^{\f12}}\|u^5\nab^5\od\|_{L^\infty_uL_{\ub}^{2}L^{2}(\S)}\right).$$

By the following $\div$-$\curl$ system:
\begin{gather*}
\div\nab\omega=\div\kappa+\f12\div\beta,\\
\curl\nab\omega=0,\\
\curl\nab\od=\curl\kappa+\f12\curl\beta,\\
\div\nab\od=0.
\end{gather*}
and the elliptic estimates from Proposition \ref{ellipticthm}, we have
\begin{equation*}
\begin{split}
&\quad\ \|u^5\nab^5(\omega,\od)\|_{L^2(\S)}\\
&\ls \sum_{j\leq 4}(\|u^{j+1}\nab^j\kappa\|_{L^2(\S)}+\|u^{j+1}\nab^{j}\beta\|_{L^2(\S)}+\|u^j\nab^j(\omega,\od)\|_{L^2(\S)})\\
&\ls \|u^5\nab^4\kappa\|_{L^2(\S)}+\sum_{j\leq 4}(\|u^{j+1}\nab^{j}\beta\|_{L^2(\S)}+\|u^j\nab^j(\omega,\od)\|_{L^2(\S)}).
\end{split}
\end{equation*}
This implies, after taking $L^2$ in $\ub$,
\begin{equation*}
\begin{split}
\|u^5\nab^5(\omega,\od)\|_{L^2_{\ub}L^2(\S)}
\ls \de^{\f12}\at\left(1+\M R+\f{1}{\delta^{\f12}\at b^{\f12}}\|u^5\nab^5\od\|_{L^\infty_uL_{\ub}^{2}L^{2}(\S)}\right).
\end{split}
\end{equation*}
For $b$ sufficiently large, we can absorb the last term to the left hand side to get
\begin{equation*}
\|u^5\nab^5(\omega,\od)\|_{L^2_{\ub}L^2(\S)}\ls \de^{\f12}\at(1+\M R).\vspace{-1em}
\end{equation*}
\end{proof}
In the following proposition, we prove the highest order derivative estimates for $\eta$. We need in particular an improved bound for $\nab^i\mu$ compared to $\nab^{i+1}\eta$ for $1\leq i\leq 4$.

\begin{proposition}\label{eta.5.bd}
Under the assumptions of Theorem \ref{main.thm} and the bootstrap assumptions \eqref{BA.1}, \eqref{BA.2}, \eqref{BA.3} and \eqref{BA.4}, we have
\begin{equation*}
\|u^6\nab^5\eta\|_{L^{\infty}_uL^2_{\ub}L^2(\S)}
\ls \de^{\f32} a^{\f34}(1+\M R)
\end{equation*}
and 

$$\|u^5\nab^5\eta\|_{L^{\infty}_{\ub}L^2_uL^2(\S)}\ls \de^{\f12}\at (1+\M R). $$
For $1\leq i\leq 4$, the derivatives of the mass aspect function $\nab^i\mu$ obey the following improved estimates:
$$\sum_{1\leq i\leq 4}\|u^{i+3}\nab^i\mu\|_{L^\infty_u L^\infty_{\ub}L^2(\S)}\ls\de^2 a^{\f54}b^{\f14}.$$
\end{proposition}

\begin{proof}

Define $\mu$ by 
$$\mu=-\div\eta+\K.$$
Thus $\eta$ obeys the following elliptic system
$$\div\eta=-\mu+\K, \quad \quad \curl\eta=\sigmac$$
and $\mu$ satisfies the equation
$$\nab_4 \mu=\p\nab(\eta,\etb)+\p\q\hspace{1pt}\q+\q\nab(\trch,\chih)+\trch K.$$ 
Commuting with angular derivatives for $i$ times with $1\leq i\leq 4$, we get
\begin{equation*}
\begin{split}
\nab_4 \nab^i \mu &=\p\nab^{i+1}(\eta,\etb)+\q\nab^{i+1}(\trch,\chih)+\sum_{\substack{i_1+i_2+i_3=i+1\\i_1,i_3\leq i}} \nab^{i_1}\q^{i_2+1}\nab^{i_3}\p\\
&\quad +\sum_{i_1+i_2+i_3+i_4=i}\nab^{i_1}\q^{i_2}\nab^{i_3}\trch \nab^{i_4} K.
\end{split}
\end{equation*}
Notice that $\nab^i\mu$ vanishes initially on $\Hb_0$. By Proposition \ref{transport}, in order to estimate $\|u^{i+2}\nab^i\mu\|_{L_{\ub}^{\infty}L_{u}^{\infty}L^{2}(\S)}$, it suffices to estimate the $\|u^{i+2}\cdot\|_{L_{u}^{\infty}L_{\ub}^{1}L^{2}(\S)}$ norm of the right hand side. We now estimate each of the terms in the equation.

For the term containing highest derivative of $\etb$, we have
\begin{equation*}
\begin{split}
&\quad\ \sum_{i\leq 4}\|u^{i+2}\p\nab^{i+1}(\eta,\etb)\|_{L^1_{\ub}L^2(\S)}\\
&\ls  \frac{\de^{\f12}}{|u|}\|u\p\|_{L^{\infty}_{\ub}L^{\infty}(\S)}\\
&\quad\times\left(\delta^{\f12}\sum_{i\leq 3}\|u^{i+2}\nab^{i+1}(\eta,\etb)\|_{L^\infty_{\ub}L^2(\S)}+\|u^6\nab^5(\eta,\etb)\|_{L^2_{\ub}L^2(\S)}\right)\\
&\ls  \frac{\de^{\frac 12}}{|u|}\at \de^{\f32}a^{\f34}b^{\f14}\\
&\ls \f{\delta^2 a^{\f54} b^{\f14}}{|u|},
\end{split}
\end{equation*}
where we have used the bounds in Proposition \ref{p.bd} and the bootstrap assumption \eqref{BA.3}.

For the other term with the highest order derivative, i.e., the term containing the highest derivative of $(\trch,\chih)$, we get
\begin{equation*}
\begin{split}
&\quad\ \sum_{i\leq 4}\|u^{i+2}\q\nab^{i+1}(\trch,\chih)\|_{L^1_{\ub}L^2(\S)}\\
&\ls  \frac{\de^{\f12}}{|u|}\|u^2\q\|_{L^{\infty}_{\ub}L^{\infty}(\S)}\\
&\quad\times\left(\sum_{i\leq 3}\|u^{i+1}\nab^{i+1}(\trch,\chih)\|_{L^2_{\ub}L^2(\S)}+\|u^5\nab^5(\trch,\chih)\|_{L^2_{\ub}L^2(\S)}\right)\\
&\ls  \frac{\de^{\frac 12}}{|u|}\de\at b^{\f14}\cdot\de^{\f12}\at b^{\f14}\\
&\ls \frac{\de^2 a b^{\f12}}{|u|}
\end{split}
\end{equation*}
where we have used the bootstrap assumptions \eqref{BA.1}, \eqref{BA.2} and \eqref{BA.3}.

We now estimate the lower order terms in the Ricci coefficients. Using Proposition \ref{product} together with the bootstrap assumption \eqref{BA.2}, we obtain

\begin{equation*}
\begin{split}
&\quad\ \sum_{i\leq 4}\Bigg\|u^{i+2}\sum_{\substack{i_1+i_2+i_3=i+1\\i_1,i_3\leq i}} \nab^{i_1}\q^{i_2+1}\nab^{i_3}\p\Bigg\|_{L^1_{\ub}L^2(\S)}\\
&\ls  \frac{\de}{|u|}\sum_{i_1+i_2\leq 2}\|u^{i_1+i_2+2}\nab^{i_1}\q^{i_2+1}\|_{L^{\infty}_{\ub}L^{\infty}(\S)}\sum_{i_3\leq 4}\|u^{i_3}\nab^{i_3}\p\|_{L^{\infty}_{\ub}L^2(\S)}\\
&\quad + \frac{\de}{|u|}\sum_{\substack{i_1+i_2\leq 5\\i_1\leq 4}}\|u^{i_1+i_2+1}\nab^{i_1}\q^{i_2+1}\|_{L^{\infty}_{\ub}L^{2}(\S)}\sum_{i_3\leq 2}\|u^{i_3+1}\nab^{i_3}\p\|_{L^{\infty}_{\ub}L^{\infty}(\S)}\\
&\ls  \frac{\de}{|u|}\de\at b^{\f14}\cdot\at\ls\frac{\de^2 a b^{\f14}}{|u|}.
\end{split}
\end{equation*}\enlargethispage{2em}
Finally, we control the term containing the Gauss curvature $K$. For this term, we need to make use of the fact $i\geq 1$ and apply the improved estimate for $\nab^i\trch$ from Proposition \ref{trch.bd}. Using H\"older's inequality and Sobolev embedding in Proposition \ref{Sobolev}, we have
\begin{equation*}
\begin{split}
&\quad\ \sum_{1\leq i\leq 4}\left\|u^{i+2}\sum_{i_1+i_2+i_3+i_4=i}\nab^{i_1}\q^{i_2}\nab^{i_3}\trch \nab^{i_4} K\right\|_{L^1_{\ub}L^2(\S)}\\
&\ls  \frac{\de^{\f12}}{|u|^2}\sum_{i_1+i_2\leq 2}\|u^{i_1+i_2+1}\nab^{i_1}\q^{i_2}\|_{L^{\infty}_{\ub}L^{\infty}(\S)}\sum_{i_3\leq 2}\|u^{i_3+1}\nab^{i_3}\trch\|_{L^{\infty}_{\ub}L^{\infty}(\S)}\\
&\quad\times\sum_{i_4\leq 4}\left\|u^{i_4+2}\nab^{i_4}\left(K-\f1{|u|^2}\right)\right\|_{L^2_{\ub}L^2(\S)}\\
&\quad + \frac{\de^{\f12}}{|u|^2}\sum_{i_1+i_2\leq 2}\|u^{i_1+i_2+1}\nab^{i_1}\q^{i_2}\|_{L^{\infty}_{\ub}L^{\infty}(\S)}\sum_{i_3\leq 4}\|u^{i_3}\nab^{i_3}\trch\|_{L^{\infty}_{\ub}L^{2}(\S)}\\
&\qquad\times\Bigg(\sum_{i_4\leq 2}\left\|u^{i_4+3}\nab^{i_4}\left(K-\f1{|u|^2}\right)\right\|_{L^2_{\ub}L^\infty(\S)}+\|u\|_{L^2_{\ub}L^\infty(\S)}\Bigg)\\
&\quad + \frac{\de^{\f12}}{|u|^2}\sum_{i_1+i_2\leq 4}\|u^{i_1+i_2}\nab^{i_1}\q^{i_2}\|_{L^{\infty}_{\ub}L^2(\S)}\sum_{i_3\leq 2}\|u^{i_3+1}\nab^{i_3}\trch\|_{L^{\infty}_{\ub}L^{\infty}(\S)}\\
&\qquad\times\Bigg(\sum_{i_4\leq 2}\left\|u^{i_4+3}\nab^{i_4}\left(K-\f1{|u|^2}\right)\right\|_{L^2_{\ub}L^\infty(\S)}+\|u\|_{L^2_{\ub}L^\infty(\S)}\Bigg)\\
&\ls  \frac{\de^{\f12}}{|u|^2} (|u|+\de \at b^{\f14})\f{\delta a}{|u|}(\de^{\f32}a^{\f34}b^{\f14}+\de^{\f12}|u|)\\
&\ls  \f{\delta^3 a^{\f74} b^{\f14}}{|u|^2}+\f{\delta^2 a}{|u|}+\f{\delta^4 a^{\f94}b^{\f12}}{|u|^3}+ \frac{\de^3 a^{\f32} b^{\f14}}{|u|^2} \ls \f{\delta^2 a^{\f54}}{|u|}
\end{split}
\end{equation*}
using Propositions \ref{product} and \ref{trch.bd} and the bootstrap assumption \eqref{BA.3}.

Combining all the estimates above and using $|u|\geq\de\at b$, we have
\begin{equation}\label{mu.est}
\begin{split}
\sum_{1\leq i\leq 4}\|u^{i+2}\nab^i\mu\|_{L^{\infty}_uL^2(\S)}
\ls& \f{\de^2 a^{\f54} b^{\f14}}{|u|}.
\end{split}
\end{equation}

By the div-curl system
$$\div\eta=-\mu+\K,\quad\curl \eta=\sigmac$$
and elliptic estimates from Proposition \ref{ellipticthm}, we have
\begin{equation*}
\begin{split}
&\quad\ \|u^6\nab^5\eta\|_{L^2(\S)}\\
&\ls \sum_{i\leq 4}\Bigg(\|u^{i+2}\nab^i\mu\|_{L^2(\S)}\\
&\qquad\quad  + \left\|u^{i+2}\nab^i \left(K-\frac 1{|u|^2},\sigmac\right)\right\|_{L^2(\S)}+\|u^{i+1}\nab^i\eta\|_{L^2(\S)}\Bigg)\\
&\ls \|u^6\nab^4\mu\|_{L^2(\S)}+\sum_{i\leq 4}\left\|u^{i+2}\nab^i \left(K-\frac 1{|u|^2},\sigmac\right)\right\|_{L^2(\S)}+\de\at b^{\f14},
\end{split}
\end{equation*}
where we have used the bootstrap assumption \eqref{BA.1}.
Hence, using \eqref{mu.est} and taking $L^2$ in $\ub$ and $u$ respectively, we obtain
\begin{equation*}
\|u^6\nab^5\eta\|_{L^{\infty}_uL^2_{\ub}L^2(\S)}
\ls \frac{\de^{\f32} a^{\f34}}{b^{\f12}}+\de^{\f32}a^{\f12}b^{\f14}+\de^{\f32} a^{\f34}\M R\ls \de^{\f32} a^{\f34}(1+\M R)
\end{equation*}
and 
$$\|u^5\nab^5\eta\|_{L^{\infty}_{\ub}L^2_uL^2(\S)}\ls \|u^{-1}\|_{L^2_u}\left(\frac{\de a^{\f34}}{b^{\f12}}+\de\at b^{\f14}\right)+\de^{\f12}\at\M R\ls \de^{\f12}\at (1+\M R), $$
since $|u|\geq \de\at b$. This concludes the proof of the proposition.
\end{proof}
We then turn to the estimates for $\nab^5\etab$. Unlike that for $\nab^5\eta$, we only achieve an estimate that is integrated in the $\ub$ direction but we lose the bound that is integrated in the $u$ direction.

\begin{proposition} \label{etab.5.bd}
Under the assumptions of Theorem \ref{main.thm} and the bootstrap assumptions \eqref{BA.1}, \eqref{BA.2}, \eqref{BA.3} and \eqref{BA.4}, we have
\[
 \|u^6\nab^5\etb\|_{L^\infty_uL^2_{\ub}L^2(\S)} \ls \delta^{\f32}a^{\f34}(1+\mathcal R).
\]
\end{proposition}

\begin{proof}
Let $\mub$ be defined by 
$$\mub=-\div\etb+\K.$$ Thus, we have the Hodge system for $\etb$:
$$\div\etb=-\mub+\K,\quad\curl\etb=-\sigmac.$$
Moreover, $\mub$ obeys the following equation:
$$\nab_3\mub+\tr\chib\,\mub=\q\nab(\eta,\etb)+\q\,\q\,\q+\q\nab(\chibh,\trchb)+\tr\chib\div\eta+\tr\chib K+\f1{|u|^3}.$$
Commuting with angular derivatives for $4$ times, we get
\begin{equation*}
\begin{split}
&\quad\ \nab_3 \nab^4 \mub+3\tr\chib\nab^4\mub\\
&=\q\nab^5(\eta,\etab,\trchb,\chibh)+\frac{1}{|u|}\nab^5 \eta\\
&\quad +\frac{1}{|u|}\sum_{i_1+i_2+i_3=4}\nab^{i_1}\q^{i_2+1}\nab^{i_3} \q+\frac{1}{|u|}\sum_{i_1+i_2+i_3=4}\nab^{i_1}\q^{i_2}\nab^{i_3} K\\
&\quad +\sum_{i_1+i_2+i_3=4} \nab^{i_1}\q^{i_2+2}\nab^{i_3}\q+\sum_{i_1+i_2+i_3=4} \nab^{i_1}\q^{i_2+1}\nab^{i_3}K.
\end{split}
\end{equation*}

We apply Proposition \ref{evolution lemma} with $\lambda_0=3$ which shows that for every fixed $\ub$, $\|u^{5}\nab^4\mub\|_{L_{u}^{\infty}L^{2}(\S)}$ can be estimated by the $\|u^5\cdot\|_{L_{u}^{1}L^{2}(\S)}$ norm of the right hand side. Since we will only need to control $\nab^i\mub$ after taking $L^2$ norm in $\ub$, we will directly control $\|u^5 \cdot\|_{L^2_{\ub}L^1_uL^2(\S)}$ of the right hand side. We now estimate each of the terms in the equation. First, we start with terms with $5$ angular derivatives on the Ricci coefficients. We have two contributions from $\nab^5\eta$, one multiplied by $\f{1}{|u|}$ and one multiplied by $\q$. In the former case, we have, by Proposition \ref{eta.5.bd},
\begin{equation*}
\begin{split}
&\quad\ \|u^4\nab^5\eta\|_{L^2_{\ub}L^1_uL^2(\S)}\\
&\ls \|u^{-2}\|_{L^1_u}\|u^6\nab^5\eta\|_{L^\infty_uL^2_{\ub}L^2(\S)}\\
&\ls \frac{\delta^{\f32}a^{\f34}}{|u|}(1+\mathcal R).
\end{split}
\end{equation*}
Notice that in this bound, we do not gain an extra smaller factor compared to the desired estimate. It is therefore important that we have already obtained the sharp estimates in Proposition \ref{eta.5.bd} and do not have to resort to the bootstrap assumption \eqref{BA.3}.

The other contribution from $\nab^5\eta$ can be estimated together with that from $\nab^5\etab$. For these terms, we have
\begin{equation*}
\begin{split}
&\quad\ \|u^5\q\nab^5(\eta,\etab)\|_{L^2_{\ub}L^1_uL^2(\S)}\\
&\ls \|u^{-3}\|_{L^1_u}\|u^2\q\|_{L^\infty_{\ub}L^\infty_uL^\infty(\S)}\|u^6\nab^5(\eta,\etab)\|_{L^\infty_uL^2_{\ub}L^2(\S)}\\
&\ls \f{\de^{\f52}a^{\f54}b^{\f12}}{|u|^2},
\end{split}
\end{equation*}
where we have used the bootstrap assumptions \eqref{BA.1} and \eqref{BA.3} in the last inequality.

For the contributions from $\nab^5(\trchb,\chibh)$, we have
\begin{equation*}
\begin{split}
&\quad\ \|u^5\q\nab^5(\trchb,\chibh)\|_{L^2_{\ub}L^1_uL^2(\S)}\\
&\ls \delta^{\f12}\|u^{-2}\|_{L^2_u}\|u^2\q\|_{L^\infty_{\ub}L^\infty_uL^\infty(\S)}\|u^5\nab^5(\trchb,\chibh)\|_{L^\infty_{\ub}L^2_uL^2(\S)}\\
&\ls \de^{\f12}\cdot\frac{\de\at b^{\f14}}{|u|^{\f32}}\cdot\frac{\de\at b^{\f14}}{|u|^{\f12}}\ls \f{\de^{\f52}a b^{\f12}}{|u|^2},
\end{split}
\end{equation*}
where we have used the bootstrap assumptions \eqref{BA.1} and \eqref{BA.3}.

We then move to the lower order terms. First, we have
\begin{equation*}
\begin{split}
&\quad\ \left\|u^5\frac{1}{|u|}\sum_{i_1+i_2+i_3=4}\nab^{i_1}\q^{i_2+1}\nab^{i_3} \q\right\|_{L^2_{\ub}L_{u}^{1}L^{2}(\S)}\\
&\ls  \delta^{\f12}\|u^{-3}\|_{L^1_u}\sum_{i_1+ i_2\leq 2}\|u^{i_1+i_2+2}\nabla^{i_1}\q^{i_2+1}\|_{L^\infty_{\ub}L_{u}^{\infty}L^\infty(\S)}\\
&\quad\times\sum_{i_3\leq 4}\|u^{i_3+1}\nab^{i_3}\q\|_{L^\infty_{\ub}L_{u}^{\infty}L^2(\S)}\\
&\quad +\delta^{\f12} \|u^{-3}\|_{L^1_u}\sum_{i_1+ i_2\leq 4}\|u^{i_1+i_2+1}\nabla^{i_1}\q^{i_2+1}\|_{L^\infty_{\ub}L_{u}^{\infty}L^2(\S)}\\
&\qquad\times\sum_{i_3\leq 2}\|u^{i_3+2}\nab^{i_3}\q\|_{L^\infty_{\ub}L_{u}^{\infty}L^{\infty}(\S)}\\
&\ls  \frac{\de^{\f52} a b^{\f12}}{|u|^2},
\end{split}
\end{equation*}
where we have used Proposition \ref{product} and the bootstrap assumption \eqref{BA.1}.

Then, for the term with $K$, using Sobolev embedding in Proposition \ref{Sobolev}, we have
\begin{equation*}
\begin{split}
&\quad\ \left\|u^5\frac{1}{|u|}\sum_{i_1+i_2+i_3=4}\nab^{i_1}\q^{i_2}\nab^{i_3} K\right\|_{L^2_{\ub}L_{u}^{1}L^{2}(\S)}\\
&\ls \|u^{-2}\|_{L^1_u} \left\|u^6\nab^4\left(\K\right)\right\|_{L^\infty_uL^2_{\ub}L^2(\S)} \\
&\quad +\delta^{\f12}\|u^{-2}\|_{L^2_u}\sum_{i_1+ i_2\leq 2}\|u^{i_1+i_2+2}\nabla^{i_1}\q^{i_2+1}\|_{L^\infty_{\ub}L_{u}^{\infty}L^\infty(\S)}\\
&\qquad\times\sum_{i_3\leq 4}\left\|u^{i_3+1}\nab^{i_3}\left(\K\right)\right\|_{L^\infty_{\ub}L_{u}^2L^2(\S)}\\
&\quad + \delta^{\f12}\sum_{i_1+ i_2\leq 4}\|u^{i_1+i_2+1}\nabla^{i_1}\q^{i_2+1}\|_{L^\infty_{\ub}L_{u}^{\infty}L^2(\S)}\\
&\qquad\times\left(\|u^{-2}\|_{L^2_u}\sum_{i_3\leq 2}\|u^{i_3+2}\nab^{i_3}K\|_{L^\infty_{\ub}L_{u}^2L^{\infty}(\S)}+\|u^{-2}\|_{L^1_u}\right)\\
&\ls  \frac{\de^{\f32} a^{\f34}}{|u|}\M R+\f{\delta^2 a b^{\f12}}{|u|^{\f32}}+\f{\de^{\f32} a^{\f12} b^{\f14}}{|u|}\ls \frac{\de^{\f32} a^{\f34}}{|u|}\M R+\f{\de^{\f32} a^{\f34} }{|u|},
\end{split}
\end{equation*}
where we have used Proposition \ref{product}.

The remaining two terms are actually better behaved than the two terms that we have just estimates since $\q$ obeys better bounds that $\f{1}{|u|}$. More precisely, we have
\begin{equation*}
\begin{split}
&\quad\ \left\|\sum_{i_1+i_2+i_3=4}u^5\nab^{i_1}\q^{i_2+2}\nab^{i_3}\q\right\|_{L^2_{\ub}L_{u}^{1}L^{2}(\S)}\\
&\ls \de^{\f12}\|u^{-4}\|_{L^1_u}\sum_{i_1+ i_2\leq 2}\|u^{i_1+i_2+4}\nabla^{i_1}\q^{i_2+2}\|_{L^\infty_{\ub}L_{u}^{\infty}L^\infty(\S)}\\
&\quad\times\sum_{i_3\leq 4}\|u^{i_3+1}\nab^{i_3}\q\|_{L^\infty_{\ub}L_{u}^{\infty}L^2(\S)}\\
&\quad +\delta^{\f12}\|u^{-4}\|_{L^1_u}\sum_{i_1+ i_2\leq 4}\|u^{i_1+i_2+3}\nabla^{i_1}\q^{i_2+2}\|_{L^\infty_{\ub}L_{u}^{\infty}L^2(\S)}\\
&\qquad\times\sum_{i_3\leq 4}\|u^{i_3+2}\nab^{i_3}\q\|_{L^\infty_{\ub}L_{u}^{\infty}L^\infty(\S)}\\
&\ls \f{\de^{\f12}\cdot\de^2 a b^{\f12} \cdot \delta\at b^{\f14}}{|u|^3}\ls\f{\de^{\f72}a^{\f32}b^{\f34}}{|u|^3}
\end{split}
\end{equation*}
using Proposition \ref{product} and the bootstrap assumption \eqref{BA.1}.

Finally, the last term can be bounded by
\begin{equation*}
\begin{split}
&\quad\ \left\|u^5\sum_{i_1+i_2+i_3=4}\nab^{i_1}\q^{i_2+1}\nab^{i_3} K\right\|_{L^2_{\ub}L_{u}^{1}L^{2}(\S)}\\
&\ls \delta^{\f12}\|u^{-2}\|_{L^2_u}\sum_{i_1+ i_2\leq 2}\|u^{i_1+i_2+2}\nabla^{i_1}\q^{i_2+1}\|_{L^\infty_{\ub}L_{u}^{\infty}L^\infty(\S)}\\
&\quad\times\sum_{i_3\leq 4}\|u^{i_3+1}\nab^{i_3}\left(\K\right)\|_{L^\infty_{\ub}L_{u}^2L^2(\S)}\\
&\quad + \delta^{\f12}\sum_{i_1+ i_2\leq 4}\|u^{i_1+i_2+1}\nabla^{i_1}\q^{i_2+1}\|_{L^\infty_{\ub}L_{u}^{\infty}L^2(\S)}\\
&\qquad\times\left(\|u^{-2}\|_{L^2_u}\sum_{i_3\leq 2}\left\|u^{i_3+2}\nab^{i_3}\left(\K\right)\right\|_{L^\infty_{\ub}L_{u}^2L^{\infty}(\S)}+\|u^{-2}\|_{L^1_u}\right)\\
&\ls  \frac{\de^2 a b^{\f12}}{|u|^{\f32}}+\f{\de^{\f32} a^{\f12} b^{\f14}}{|u|}\ls \f{\de^{\f32} a^{\f34}}{|u|},
\end{split}
\end{equation*}
where we have used Proposition \ref{product} and the bootstrap assumption \eqref{BA.3}.

Therefore, combining the above estimates, we get
$$\|u^5\nab^4\mub\|_{L^2_{\ub}L^\infty_u L^2(\S)}\ls \f{\delta^{\f32} a^{\f34}}{|u|}(1+\M R)$$
using $|u|\geq \de\at b$ and $b\leq a$. This implies, after multiplying by $u$, that
$$\|u^6\nab^4\mub\|_{L^\infty_u L^2_{\ub}L^2(\S)}\ls \delta^{\f32} a^{\f34}(1+\M R).$$
Therefore, using the div-curl system
$$\div\etb=-\mub+\K,\quad\curl \etb=-\sigmac$$
and applying elliptic estimates from Proposition \ref{ellipticthm}, we have
\begin{equation*}
\begin{split}
&\quad\ \|u^6\nab^5\etb\|_{L^2(\S)}\\
&\ls  \sum_{i\leq 4}\Bigg(\|u^{i+2}\nab^{i}\mub\|_{L^2(\S)}+\left\|u^{i+2}\nab^{i}\left(\K,\sigmac\right)\right\|_{L^2(\S)}\\
&\qquad\quad +\|u^{i+1}\nab^i\etab\|_{L^2(\S)}\Bigg)\\
&\ls  \|u^6\nab^4\mub\|_{L^2(\S)}\\
&\quad +\sum_{i\leq 4}\left(\|u^{i+1}\nab^i\etab\|_{L^2(\S)}+\left\|u^{i+2}\nab^{i}\left(\K,\sigmac\right)\right\|_{L^2(\S)}\right)\\
&\ls  \|u^6\nab^4\mub\|_{L^2(\S)}+\sum_{i\leq 4} \left\|u^{i+2}\nab^{i}\left(\K,\sigmac\right)\right\|_{L^2(\S)}+\de\at b^{\f14},
\end{split}
\end{equation*}
where we have used the bootstrap assumption \eqref{BA.1} in the last step. This implies, after taking $L^2_{\ub}$ norm, that 
$$\|u^6\nab^5\etb\|_{L^\infty_uL^2_{\ub}L^2(\S)}\ls \delta^{\f32}a^{\f34}(1+\M R).\vspace{-1em}$$
\end{proof}
We now prove the highest order bounds for $\omb$:

\begin{proposition}\label{omb.5.bd}
Under the assumptions of Theorem \ref{main.thm} and the bootstrap assumptions \eqref{BA.1}, \eqref{BA.2}, \eqref{BA.3} and \eqref{BA.4}, we have
\[
 \|u^5\nab^5\omb\|_{L^\infty_{\ub}L^2_uL^2(\S)}\ls \f{\delta\at}{|u|^{\f12}}(1+\M R).
\]

\end{proposition}

\begin{proof}
Recall that in the proof of Proposition \ref{om.5.bd} we have defined an auxiliary function $\om^{\dagger}$ in order to apply elliptic estimates to obtain the highest order estimate for $\om$. Here, we similarly define an auxiliary function $\ombd$ by
$$\nab_4\ombd=\f12\sigmac$$ 
with zero initial data on $\Hb_0$.
We then define $\kappab$ by
$$\kappab:=-\nab\omb+^*\nab\ombd-\f12\beb.$$
Before we proceed, observe that the proof of Proposition \ref{q.bd} implies that 
\begin{align}\label{ombd.bd}
&\quad\ \sum_{i\leq 4}\|u^{i+1}\nab^i \ombd\|_{L^\infty_uL^\infty_{\ub}L^2(\S)}\\
\notag &\ls \de\at\left(1+\frac{1}{\de^{\frac 12}\at}\|u^5\nab^5\om\|_{L^\infty_uL^2_{\ub}L^2(\S)}+\mathcal R\right).
\end{align}
(In fact, the stronger bound with $\de\at$ on the right hand side holds. We will not need this refinement.)
In view of this bound, we will allow $\q$ to also denote $\ombd$ in the remainder of the proof of this proposition. With this new convention, it is easy to check that $\kappab$ obeys the equation
\begin{equation*}
\begin{split}
\nab_4\kappab &= \sum_{i_1+i_2+i_3=1}\nab^{i_1}\q^{i_2+1}\nab^{i_3}\q+\sum_{i_1+i_2+i_3=1}\q^{i_1}\nab^{i_2}\left(\trchb+\f2{|u|},\chibh,\omb\right)\nab^{i_2}\p\\
&\quad +\q\left(\K,\sigmac\right)+\f{1}{|u|^2}\q+\f{1}{|u|}\nab\trch+\f1{|u|}\q\p.
\end{split}
\end{equation*}
Commuting with angular derivatives for $i$ times, we get
\begin{equation*}
\begin{split}
\nab_4 \nab^4 \kappab &=(\p,\q)\nab^5(\trchb,\chibh,\omb)+\q\nab^5(\eta,\etab)+\q\nab^5\p+\f1{|u|}\nab^5\trch\\
&\quad +\sum_{\substack{i_1+i_2+i_3=5\\i_1,i_3\leq 4}} \nab^{i_1}\q^{i_2+1}\nab^{i_3}(\p,\q)\\
&\quad +\sum_{i_1+i_2+i_3=4}\nab^{i_1}\q^{i_2+1}\nab^{i_3}\left(\K,\sigmac\right)\\
&\quad +\sum_{i_1+i_2=4}\f{1}{|u|^2}\nab^{i_1}\q^{i_2+1}+\sum_{i_1+i_2+i_3=4}\f1{|u|}\nab^{i_1}\q^{i_2+1}\nab^{i_3}\p.
\end{split}
\end{equation*}

Using Proposition \ref{transport} and the fact that $\nab^5\kappab=0$ on $\Hb_0$, we can estimate $\|u^5\nab^4\kappab\|_{L_{\ub}^{\infty}L_{u}^2L^{2}(S)}$ by controlling the $\|u^5\cdot\|_{L_u^2L_{\ub}^{1}L^{2}(S)}$ norm of the right hand side. We now estimate each of the terms in the equation.

We first estimate the term with highest order derivatives of the Ricci coefficients, i.e., the terms
$$(\q,\p)\nab^5(\tr \chib,\chibh,\omb),\quad\q\nab^5(\eta,\etab),\quad\q\nab^5\p,\quad\frac{1}{|u|}\nab^5\trch.$$ 
For the contributions from $(\q,\p)\nab^5(\tr \chib,\chibh,\omb)$, we have
\begin{equation*}
\begin{split}
&\quad\ \|u^5(\p,\q)\nab^5(\tr \chib,\chibh,\omb)\|_{L^2_{u}L^1_{\ub}L^2(S_{u,\ub})}\\
&\ls \|u^5\nab^5(\tr\chib,\chibh,\omb)\|_{L^{\infty}_{\ub}L^2_{u}L^2(S_{u,\ub})}\|(\p,\q)\|_{L^1_{\ub}L^{\infty}_{u}L^{\infty}(S_{u,\ub})}\\
&\ls \frac{\delta a^{\f12}b^{\f14}}{|u|^{\f12}}\f{\delta a^{\f12}b^{\f14}}{|u|}\ls \f{\delta^2a b^{\f12}}{|u|^{\f32}},
\end{split}
\end{equation*}
where we have used bootstrap assumptions \eqref{BA.1}, \eqref{BA.2} and \eqref{BA.3}. 

For the contributions from $\q\nab^5(\eta,\etb)$, we have
\begin{equation*}
\begin{split}
&\quad\ \|u^5\q\nab^5(\eta,\etb)\|_{L^2_{u}L^1_{\ub}L^2(S_{u,\ub})}\\
&\ls \|u^6\nab^5(\eta,\etb)\|_{L^{\infty}_{u}L^2_{\ub}L^2(S_{u,\ub})}\|u^2\q\|_{L^2_{\ub}L^{\infty}_{u}L^{\infty}(S_{u,\ub})}\|u^{-3}\|_{L^2_u}\\
&\ls \delta^{\f32}a^{\f34}b^{\f14}\f{\delta^{\f32}a^{\f12}b^{\f14}}{|u|^{\f52}}\ls \f{\delta^3 a^{\f54} b^{\f12}}{|u|^{\f52}},
\end{split}
\end{equation*}
where we have used the bootstrap assumptions \eqref{BA.1} and \eqref{BA.3}. 

We then control the term $u^5\q\nab^5\p$, for which we have
\begin{equation*}
\begin{split}
&\quad\ \|u^5\q\nab^5\p\|_{L^2_{u}L^1_{\ub}L^2(S_{u,\ub})}\\
&\ls \|u^5\nab^5\p\|_{L^{\infty}_{u}L^2_{\ub}L^2(S_{u,\ub})}\|\q\|_{L^2_{\ub}L^{2}_{u}L^{\infty}(S_{u,\ub})}\\
&\ls \delta^{\f12}a^{\f12}b^{\f14}\f{\delta^{\f32}a^{\f12}b^{\f14}}{|u|^{\f32}}\ls  \f{\delta^2 a b^{\f12}}{|u|^{\f32}}.
\end{split}
\end{equation*}
Here, we have used the bootstrap assumptions \eqref{BA.1} and \eqref{BA.3}. 

To estimate the remaining highest order term $u^4\q\nab^5\tr\chi$, we use the bound from Proposition \ref{chi.5.bd} which is stronger than that in the bootstrap assumptions. More precisely, we have
\begin{equation*}
\begin{split}
&\quad\ \|u^4\nab^5\tr\chi\|_{L^2_{u}L^1_{\ub}L^2(S_{u,\ub})}\\
&\ls \|u^5\nab^5\tr\chi\|_{L^{\infty}_{u}L^2_{\ub}L^2(S_{u,\ub})}\left\|\f{1}{|u|}\right\|_{L^2_{\ub}L^{2}_{u}L^{\infty}(S_{u,\ub})}\\
&\ls \f{\delta^{\f32}ab^{\f14}}{|u|}\f{\delta^{\f12}}{|u|^{\f12}}\ls \f{\delta^2ab^{\f14}}{|u|^{\f32}}
\end{split}
\end{equation*}
where we have used Proposition \ref{chi.5.bd}.

After estimating the highest order Ricci coefficient term, we now control the curvature terms, i.e., the terms with $(\K,\sigmac)$. For these terms, we have the estimate
\begin{equation*}
\begin{split}
&\quad\ \left\|u^5\sum_{i_1+i_2+i_3=4}\nab^{i_1}\q^{i_2+1}\nab^{i_3}\left(\K,\sigmac\right)\right\|_{L^2_{u}L^1_{\ub}L^2(S_{u,\ub})}\\
&\ls \sum_{i_3\leq 4} \left\|u^{i_3+1}\nab^{i_3}\left(\K,\sigmac\right)\right\|_{L^{\infty}_{\ub}L^2_{u}L^2(S_{u,\ub})}\\
&\quad\times\sum_{i_1+i_2\leq 2}\|u^{i_1+i_2}\nab^{i_1}\q^{i_2+1}\|_{L^{\infty}_{u}L^{1}_{\ub}L^{\infty}(S_{u,\ub})}\\
&\quad +\sum_{i_3\leq 1}\left\|u^{i_3+3}\nab^{i_3}\left(\K,\sigmac\right)\right\|_{L^{\infty}_{u}L^2_{\ub}L^{\infty}(S_{u,\ub})}\\
&\qquad \times\sum_{i_1+i_2\leq 4}\|u^{i_1+i_2+1}\nab^{i_1}\q^{i_2+1}\|_{L^{\infty}_{u}L^{2}_{\ub}L^{2}(S_{u,\ub})}\|u^{-3}\|_{L^2_u}\\
&\ls \delta^{\f12}a^{\f12}b^{\f14}\f{\delta^2 a^{\f12}}{|u|^2}b^{\f14}+\delta^{\f32}a^{\f34}b^{\f14}\f{\delta^{\f32} a^{\f12} b^{\f14}}{|u|^{\f52}}\\
&\ls \f{\delta^{\f52}ab^{\f12}}{|u|^2}+\f{\delta^3a^{\f54}b^{\f12}}{|u|^{\f52}}\ls \f{\delta^{\f52}ab^{\f12}}{|u|^2},
\end{split}
\end{equation*}
where we have used the Sobolev embedding in Proposition \ref{Sobolev} to control the curvature terms in $L^{\infty}(\S)$ by the curvature norm $\mathcal R$. In the above, we have also used Proposition \ref{product} to bound the product of derivatives of $\q$ and used bootstrap assumption \eqref{BA.3} to estimate $\mathcal R$.

We then move to the lower order terms, i.e., the terms containing only Ricci coefficients and such that there are at most $4$ angular covariant derivatives on $\p$ and $\q$. These are the terms 
$$\sum_{\substack{i_1+i_2+i_3=5,\\ i_1,i_3\leq 4}}\nab^{i_1}\q^{i_2+1}\nab^{i_3}(\p,\q),\quad\sum_{i_1+i_2=4}\f{1}{|u|^2}\nab^{i_1}\q^{i_2+1}$$
and
$$\sum_{i_1+i_2+i_3=4}\f{1}{|u|}\nab^{i_1}\q^{i_2+1}\nab^{i_3}\p.$$
We estimate these terms according to the order above. First, we have
\begin{equation*}
\begin{split}
&\quad\ \Bigg\|\sum_{\substack{i_1+i_2+i_3=5,\\ i_1,i_3\leq 4}}u^5\nab^{i_1}\q^{i_2+1}\nab^{i_3}(\p,\q)\Bigg\|_{L^2_{u}L_{\ub}^{1}L^{2}(\S)}\\
& \ls \sum_{i_1+ i_2\leq 2}\|u^{i_1+i_2}\nabla^{i_1}\q^{i_2+1}\|_{L^2_{u}L_{\ub}^{\infty}L^\infty(\S)}\\
&\quad\times\sum_{i_3\leq 4}\|u^{i_3}\nab^{i_3}(\p,\q)\|_{L^\infty_{u}L_{\ub}^{1}L^2(\S)}\\
&\quad +\sum_{i_1+ i_2\leq 4}\|u^{i_1+i_2-1}\nabla^{i_1}\q^{i_2+1}\|_{L^2_{u}L_{\ub}^{\infty}L^2(\S)}\\
&\qquad\times\sum_{i_3\leq 2}\|u^{i_3+1}\nab^{i_3}(\p,\q)\|_{L^\infty_{u}L_{\ub}^{1}L^\infty(\S)}\\
&\ls \f{\delta^2 a b^{\f14}}{|u|^{\f32}}
\end{split}
\end{equation*}
using Proposition \ref{product} and \ref{p.bd}, the bootstrap assumptions \eqref{BA.1} and \eqref{BA.3} and the condition $|u|\geq \de\at b$. The second term can be controlled as follows:
\begin{equation*}
\begin{split}
&\quad\ \left\|\sum_{i_1+i_2=4}u^3\nab^{i_1}\q^{i_2+1}\right\|_{L^2_{u}L_{\ub}^{1}L^{2}(\S)}\\
&\ls \sum_{i_1+ i_2\leq 2}\|u^{i_1+i_2-1}\nabla^{i_1}\q^{i_2}\|_{L^2_{u}L_{\ub}^{\infty}L^\infty(\S)}\sum_{i_3\leq 4}\|u^{i_3}\nab^{i_3}\q\|_{L^\infty_{u}L_{\ub}^{1}L^2(\S)}\\
&\ls \f{\delta^2a^{\f12}b^{\f14}}{|u|^{\f32}}
\end{split}
\end{equation*}
using Proposition \ref{product} and the bootstrap assumption \eqref{BA.1}. Finally, we bound the remaining term by
\begin{equation*}
\begin{split}
&\quad\ \left\|\sum_{i_1+i_2+i_3=4}u^4\nab^{i_1}\q^{i_2+1}\nab^{i_3}\p\right\|_{L^2_{u}L_{\ub}^{1}L^{2}(\S)}\\
&\ls \sum_{i_1+ i_2\leq 2}\|u^{i_1+i_2}\nabla^{i_1}\q^{i_2+1}\|_{L^2_{u}L_{\ub}^{\infty}L^\infty(\S)}\sum_{i_3\leq 4}\|u^{i_3}\nab^{i_3}\p\|_{L^\infty_{u}L_{\ub}^{1}L^2(\S)}\\
&\quad +\sum_{i_1+ i_2\leq 4}\|u^{i_1+i_2-1}\nabla^{i_1}\q^{i_2+1}\|_{L^2_{u}L_{\ub}^{\infty}L^2(\S)}\sum_{i_3\leq 2}\|u^{i_3+1}\nab^{i_3}\p\|_{L^\infty_{u}L_{\ub}^{1}L^\infty(\S)}\\
&\ls \f{\delta^2 a b^{\f14}}{|u|^{\f32}}
\end{split}
\end{equation*}
using Proposition \ref{product} and \ref{p.bd}.

Therefore, combining the above estimates, we get
\begin{equation}\label{kappa}
 \|u^5\nab^4\underline{\kappa}\|_{L^{\infty}_{\ub}L^2_{u}L^2(S_{u,\ub})}\ls \f{\delta a^{\f12}}{|u|^{\f12}}.
\end{equation}
By the following div-curl system:
\begin{gather*}
\div\nab\omb=-\div\kappab-\f12\div\beb,\\
\curl\nab\omb=0,\\
\curl\nab\ombd=\curl\kappab+\f12\curl\beb,\\
\div\nab\ombd=0.
\end{gather*}
and Proposition \ref{ellipticthm}, we have
\begin{equation*}
\begin{split}
&\quad\ \|u^5\nab^5(\omb,\ombd)\|_{L^2(S_{u,\ub})}\\
&\ls \sum_{i\leq 4}\left(\|u^{i+1}\nab^i\underline{\kappa}\|_{L^2(S_{u,\ub})}+\|u^{i+1}\nab^{i}\beb\|_{L^2(S_{u,\ub})}+\|u^i\nab^i(\omb,\ombd)\|_{L^2(S_{u,\ub})}\right)\\
&\ls \|u^{5}\nab^{4}\underline{\kappa}\|_{L^2(S_{u,\ub})}+\sum_{i\leq 4}\left(\|u^{i+1}\nab^i\beb\|_{L^2(S_{u,\ub})}+\|u^i\nab^i(\omb,\ombd)\|_{L^2(S_{u,\ub})}\right)\\
\end{split}
\end{equation*}
This implies, after taking $L^2$ norm in $u$ that
\begin{equation*}
\begin{split}
&\quad\ \|u^5\nab^5(\omb,\ombd)\|_{L^2_{u}L^{2}(S_{u,\ub})}\\
&\ls \f{\delta a^{\f12}}{|u|^{\f12}}+\f{\delta^{\f32}a^{\f34}b^{\f14}}{|u|}+\f{\delta a^{\f12}}{|u|^{\f12}}(1+\M R)\\
& \ls\f{\delta a^{\f12}}{|u|^{\f12}}(1+\M R).
\end{split}
\end{equation*}
after substituting in the bound \eqref{kappa} and using bootstrap assumption \eqref{BA.3} together with \eqref{ombd.bd}, and also Propositions \ref{q.bd} and \ref{om.5.bd}.
\end{proof}
Finally, we prove the highest order estimates for the remaining Ricci coefficients, $\trchb$ and $\chibh$:

\begin{proposition}
Under the assumptions of Theorem \ref{main.thm} and the bootstrap assumptions \eqref{BA.1}, \eqref{BA.2}, \eqref{BA.3} and \eqref{BA.4}, we have
\[
 \|u^5\nab^5(\tr\chib,\chibh)\|_{L^\infty_{\ub}L^2_uL^2(\S)} \ls \f{\delta\at}{|u|^{\f12}}(1+\mathcal R).
\]
\end{proposition}

\begin{proof}

Consider the following equation for $\trchb$:
$$\nab_3 \tr\chib+\f12(\tr\chib)^2=-2\omb\tr\chib-|\chibh|^2.$$
Commuting the equation with angular derivatives for $5$ times, we get
\begin{equation*}
\begin{split}
&\quad\ \nab_3 \nab^5 \tr\chib+3\tr\chib\nab^5\tr\chib\\
&=\q\nab^5(\trchb,\chibh,\omb)+\f1{|u|}\nab^5\omb+\sum_{\substack{i_1+i_2+i_3=5\\i_1,i_3\leq 4}} \nab^{i_1}\q^{i_2+1}\nab^{i_3}\q\\
&\quad +\sum_{i_1+i_2+i_3=4}\f{1}{|u|}\nab^{i_1}\q^{i_2+1}\nab^{i_3}\q+\sum_{i_1+i_2=4}\f{1}{|u|^2}\nab^{i_1}\q^{i_2+1}.
\end{split}
\end{equation*}

We apply Proposition \ref{evolution lemma} with $\lambda_0=3$. This allows us to estimate the quantity $\|u^5\nab^5\tr\chib\|_{L_{\ub}^{\infty}L_{u}^{\infty}L^{2}(\S)}$ by bounding the $\|u^5\cdot\|_{L_{\ub}^{\infty}L_{u}^{1}L^{2}(\S)}$ norm of the right hand side of the equation above. We now estimate each of the terms in the equation. The first term can be controlled by
\begin{equation*}
\begin{split}
&\quad\ \|u^5\q\nab^5(\trchb,\chibh,\omb)\|_{L_{\ub}^{\infty}L_{u}^{1}L^{2}(\S)}\\
&\ls \|u^{-2}\|_{L^2_u}\|u^2\q\|_{L^\infty_{\ub}L^\infty_u L^\infty(\S)}\|u^5 \nab^5(\trchb,\chibh,\omb)\|_{L^\infty_{\ub}L^2_uL^2(\S)}\\
&\ls \f{\de^2 a b^{\f12}}{|u|^2},
\end{split}
\end{equation*}
where we have used the bootstrap assumption \eqref{BA.1} and \eqref{BA.3}. For the second term, we use Proposition \ref{omb.5.bd} to get
\begin{equation*}
\begin{split}
\|u^4\nab^5\omb\|_{L^{\infty}_{\ub}L^1_uL^2(\S)}
\ls &\|u^{-1}\|_{L^2_u}\f{\de\at}{|u|^{\f12}}(1+\mathcal R)\ls \f{\de\at}{|u|}(1+\mathcal R).
\end{split}
\end{equation*}\enlargethispage{1em}
We now turn to the lower order terms in the Ricci coefficients. For the third term, we have
\begin{equation*}
\begin{split}
&\quad\ \Bigg\|u^5\sum_{\substack{i_1+i_2+i_3=5\\i_1,i_3\leq 4}} \nab^{i_1}\q^{i_2+1}\nab^{i_3}\q\Bigg\|_{L_{\ub}^{\infty}L_{u}^{1}L^{2}(\S)}\\
&\ls  \|u^{-3}\|_{L^1_u}\sum_{i_1+i_2\leq 2}\|u^{i_1+i_2+2}\nab^{i_1}\q^{i_2+1}\|_{L^\infty_{\ub}L^\infty_u L^\infty(\S)}\\
&\quad\times\sum_{i_3\leq 4}\|u^{i_3+1} \nab^i\q\|_{L^\infty_{\ub}L^\infty_uL^2(\S)}\\
&\ls \f{\de^2 a b^{\f12}}{|u|^2},
\end{split}
\end{equation*}
where we have used Proposition \ref{product} and the bootstrap assumption \eqref{BA.1}.

For the fourth term, we have the estimate
\begin{equation*}
\begin{split}
&\quad\ \left\|u^5\sum_{i_1+i_2+i_3=4} \frac{1}{|u|}\nab^{i_1}\q^{i_2+1}\nab^{i_3}\q\right\|_{L_{\ub}^{\infty}L_{u}^{1}L^{2}(\S)}\\
&\ls  \|u^{-3}\|_{L^1_u}\sum_{i_1+i_2\leq 2}\|u^{i_1+i_2+2}\nab^{i_1}\q^{i_2+1}\|_{L^\infty_{\ub}L^\infty_u L^\infty(\S)}\\
&\quad\times\sum_{i_3\leq 3}\|u^{i_3+1} \nab^i\q\|_{L^\infty_{\ub}L^\infty_uL^2(\S)}\\
&\ls \f{\de^2 a b^{\f12}}{|u|^2},
\end{split}
\end{equation*}
where we have used Proposition \ref{product} the bootstrap assumption \eqref{BA.1}.

Finally, the last term can be controlled by
\begin{equation*}
\begin{split}
&\quad\ \left\|u^5\sum_{i_1+i_2=4}\f{1}{|u|^2}\nab^{i_1}\q^{i_2+1}\right\|_{L_{\ub}^{\infty}L_{u}^{1}L^{2}(\S)}\\
&\ls \|u^{-2}\|_{L^1_u}\sum_{i_1+i_2\leq 4}\|u^{i_1+i_2+1}\nab^{i_1}\q^{i_2+1}\|_{L^\infty_{\ub}L^\infty_u L^2(\S)}\\
&\ls \f{\de a^{\f12}}{|u|}\left(1+\f{1}{\delta^{\f12}\at}\|u^5\nab^5\om\|_{L^\infty_uL^2_{\ub}L^2(\S)}+\M R\right)\\
&\ls \f{\de a^{\f12}}{|u|}(1+\M R),
\end{split}
\end{equation*}
using Propositions \ref{product}, \ref{q.bd} and \ref{om.5.bd}. Notice that we do not have extra smallness in the above estimate and it is important that we applied the sharp estimates in Propositions \ref{q.bd} and \ref{om.5.bd} instead of using the bootstrap assumptions.

Collecting all the above terms, we have
$$\|u^5\nab^5\trchb\|_{L^\infty_uL^\infty_{\ub}L^2(\S)}\ls \f{\de a^{\f12}}{|u|}(1+\mathcal R),$$
which implies
\begin{equation}\label{trchb.5.1}
\|u^5\nab^5\trchb\|_{L^\infty_{\ub}L^2_uL^2(\S)}\ls \de a^{\f12}(1+\M R)\|u^{-1}\|_{L^2_u}\ls \f{\de a^{\f12}}{|u|^{\f12}}(1+\M R).
\end{equation}

In order to obtain the estimates for the fifth angular derivatives of $\chibh$, we use the Codazzi equation
$$\div\chibh=\beb+\f12\nab\tr\chib-\f12(\eta-\etb)\cdot(\chibh-\f12\tr\chib)$$
and Proposition \ref{elliptictraceless} to derive elliptic estimates for $\nab^5\chibh$:
\begin{equation*}
\begin{split}
&\quad\ \|u^5\nab^5\chibh\|_{L^2(\S)}\\
&\ls \sum_{i\leq 4}(\|u^{i+1}\nab^{i+1}\tr\chib\|_{L^2(\S)}+\|u^{i+1}\nab^{i}\beb\|_{L^2(\S)})\\
&\quad +\sum_{i\leq 4}\sum_{i_1+i_2=i}\|u^{i+1}\nab^{i_1}\q\nab^{i_2}\q\|_{L^2(\S)}+\sum_{i\leq 4}\|u^i\nab^i\chibh\|_{L^2(\S)}\\
&\ls  \sum_{i\leq 4}(\|u^{i+1}\nab^{i+1}\tr\chib\|_{L^2(\S)}+\|u^{i+1}\nab^{i}\beb\|_{L^2(\S)}+\|u^i\nab^i\chibh\|_{L^2(\S)})\\
&\quad +\f{1}{|u|^2}\sum_{i_1\leq 2}\|u^{i_1+2}\nab^{i_1}\q\|_{L^\infty(\S)}\sum_{i_2\leq 4}\|u^{i_2+1}\nab^{i_2}\q\|_{L^2(\S)}\\
&\ls \sum_{i\leq 4}(\|u^{i+1}\nab^{i+1}\tr\chib\|_{L^2(\S)}+\|u^{i+1}\nab^{i}\beb\|_{L^2(\S)})\\
&\quad +\f{\de\at}{|u|}(1+\mathcal R)+\f{\de^2 a b^{\f12}}{|u|^2},
\end{split}
\end{equation*}
where we have used the bootstrap assumption \eqref{BA.1} together with Propositions \ref{product}, \ref{q.bd} and \ref{om.5.bd} in the last step.
Taking $L^2$ norm in $u$ and using \eqref{trchb.5.1}, we obtain
\begin{equation*}
\begin{split}
&\quad\ \|u^5\nab^5\chibh\|_{L^\infty_{\ub}L^2_u L^2(\S)}\\
&\ls \sum_{i\leq 4}(\|u^{i+1}\nab^{i+1}\tr\chib\|_{L^\infty_{\ub}L^2_u L^2(\S)}\!+\!\|u^{i+1}\nab^{i}\beb\|_{L^\infty_{\ub}L^2_u L^2(\S)})\!+\!\f{\delta\at}{|u|^{\f12}}\!+\!\f{\de^{2} a b^{\f12}}{|u|^{\f32}}\\
&\ls \f{\de a^{\f12}}{|u|^{\f12}}(1+\M R)
\end{split}
\end{equation*} 
using \eqref{trchb.5.1} and the definition for the $\mathcal R$ norm.
\end{proof}

We conclude this section by the following proposition, summarizing all the estimates for the $\tilde{\mathcal O}_{5,2}$ norm derived in this section:
\begin{proposition}\label{O52.bd}
Under the assumptions of Theorem \ref{main.thm} and the bootstrap assumptions \eqref{BA.1}, \eqref{BA.2}, \eqref{BA.3} and \eqref{BA.4}, we have
\[
 \tilde{\mathcal O}_{5,2}\ls 1+\mathcal R.
\]

\end{proposition}

\section{Estimates for curvature}\label{seccurv}

In this section, we derive and prove the energy estimates for the renormalized curvature components and their first four angular derivatives. We will show that $\M R \ls 1$. Together with the estimates in the previous sections, we will therefore improve all of the bootstrap assumptions \eqref{BA.1}, \eqref{BA.2}, \eqref{BA.3} and \eqref{BA.4} and obtain Theorem \ref{main.thm}.

To derive the energy estimates, we will need the following integration by parts formula, which can be proved by direct computations:
\begin{proposition}\label{intbyparts34}
Suppose $\phi_1$ and $\phi_2$ are $r$ tensorfields, then
\begin{equation*}
\begin{split}
&\quad\ \int_{D_{u,\ub}} \phi_1 \nabla_4\phi_2+\int_{D_{u,\ub}}\phi_2\nabla_4\phi_1\\
&=\int_{\Hb_{\ub}(1,u)} \phi_1\phi_2-\int_{\Hb_0(1,u)} \phi_1\phi_2+\int_{D_{u,\ub}}(2\omega-\trch)\phi_1\phi_2.
\end{split}
\end{equation*}
Here, we have used the convention that $\Hb_{\ub}(1,u)=\{(u',\ub,\theta^1,\theta^2): u\leq u'\leq 1\}$.
\end{proposition}
\begin{proposition}\label{intbypartssph}
Suppose we have an $r$ tensorfield $^{(1)}\phi$ and an $r-1$ tensorfield $^{(2)}\phi$.
\begin{equation*}
\begin{split}
&\quad\ \int_{D_{u,\ub}}{ }^{(1)}\phi^{A_1A_2\cdots A_r}\nabla_{A_r}{ }^{(2)}\phi_{A_1\cdots A_{r-1}}+\int_{D_{u,\ub}}\nabla^{A_r}{ }^{(1)}\phi_{A_1A_2\cdots A_r}{ }^{(2)}\phi^{A_1\cdots A_{r-1}}\\
&= -\int_{D_{u,\ub}}(\eta+\etab){ }^{(1)}\phi{ }^{(2)}\phi.
\end{split}
\end{equation*}
\end{proposition}
To derive the energy estimates from the $\nab_3$ equations, we also need the following analogue of Proposition \ref{intbyparts34} in the $\nab_3$ direction. Moreover, we need to obtain an analogue which incorporates the weights in $u$. More precisely, we have
\begin{proposition}\label{intbyparts3}
Suppose $\phi$ is an $r$ tensorfield and let $\lambda_1=2(\lambda_0-\f12)$. Then
\begin{equation*}
\begin{split}
&\quad\ 2\int_{D_{u,\ub}} |u|^{2\lambda_1}\phi (\nabla_3+\lambda_0\trchb)\phi\\
&= \int_{H_u(0,\ub)} |u|^{2\lambda_1}|\phi|^2-\int_{H_0(0,\ub)} |u|^{2\lambda_1}|\phi|^2+\int_{D_{u,\ub}}|u|^{2\lambda_1}f|\phi|^2,
\end{split}
\end{equation*}
where $f$ obeys the estimate
$$|f|\ls \f{\de\at b^{\f14}}{|u|^2}.$$
Similar to Proposition \ref{intbyparts34}, we have used the convention that $H_{u}(0,\ub)=\{(u,\ub',\theta^1,\theta^2): 0\leq \ub'\leq \ub\}$.
\end{proposition}
\begin{proof}
Slightly modifying \eqref{evolution.id}, we have
\begin{equation*}
\begin{split}
&\quad\ -\frac{d}{du}\left(\int_{\S}|u|^{2\lambda_1}\Omega|\phi|^2\right)\\
&=\int_{\S}\Omega^2 \left( 2|u|^{2\lambda_1}<\phi, \nab_3\phi+\lambda_0\trchb\phi>\right)\\
&\quad +\int_{\S}\Omega^2\bigg(  |u|^{2\lambda_1} \left(\f{2\lambda_1 (e_3u)}{|u|}+(1-2\lambda_0)\trchb-2\omb\right)|\phi|^2\bigg).
\end{split}
\end{equation*}
The proposition follows after integrating with respect to $du\,d\ub$, applying the fundamental theorem of calculus in $u$ and noting that 
$$\f{2\lambda_1 (e_3u)}{|u|}+(1-2\lambda_0)\trchb-2\omb\ls \f{\de\at b^{\f14}}{|u|^2},$$
using Propositions \ref{Omega} and the bootstrap assumption \ref{BA.1}.
\end{proof}

We now derive energy estimates for $\nab^i(K-\frac 1{|u|^2},\sigmac)$ in $L^2_{\ub}L^2(\S)$ and for $\nab^i\betab$ in $L^2_uL^2(\S)$. Notice that by the definition of our norms, we also need to obtain bounds for $\K$ in the case $i=0$. We will leave this case to later (see Proposition \ref{EE.lower}). We now prove estimates for $\nab^i(K-\frac 1{|u|^2},\sigmac)$ in $L^2_{\ub}L^2(\S)$ and for $\nab^i\betab$ in $L^2_uL^2(\S)$ in the case $i\geq 1$.
\begin{proposition} \label{EE.2}
Under the assumptions of Theorem \ref{main.thm} and the bootstrap assumptions \eqref{BA.1}, \eqref{BA.2}, \eqref{BA.3} and \eqref{BA.4}, we have
\begin{equation*}
\begin{split}
\sum_{1\leq i\leq 4}\Bigg(&\left\|u^{i+2}\nab^i\left(\K,\sigmac\right)\right\|_{L^{\infty}_uL^{2}_{\ub}L^2(\S)}\\
&+\|u^{i+2}\nab^i\beb\|_{L^{\infty}_{\ub}L^{2}_{u}L^2(\S)}\Bigg)
\ls  \de^{\f32} a^{\f34}.
\end{split}
\end{equation*}
\end{proposition}

\begin{proof}
We begin with the following schematic Bianchi equations for $\sigmac$, $\K$ and $\betab$:
$$\nab_3\sigmac+\div^*\beb+\f32\tr\chib\sigmac=\sum_{i_1+i_2=1}\q^{i_1+1}\nab^{i_2}\q,$$
and
\begin{align}\label{K.eqn}
&\quad\ \nab_3\left(\K\right)+\div\beb+\f32\tr\chib\left(\K\right)\\
&=\sum_{i_1+i_2=1}\q^{i_1+1}\nab^{i_2}\q+\f{1}{|u|}\mu+\f{1}{|u|^2}\left(\tr\chib+\f{2}{u}\right)+\frac1{|u|^3}(\Om^{-1}-1),\notag 
\end{align}
and
\begin{equation*}
\begin{split}
\nab_4\beb-\nab K-^*\nab\sigmac&=\q(K,\sigmac)+\sum_{i_1+i_2+i_3=1}\q^{i_1}\nab^{i_2}\left(\trchb+\frac{2}{|u|},\chibh,\omb\right)\nab^{i_3}\p\\
&\quad +\sum_{i_1+i_2=1}\f1{|u|}\q^{i_1}\nab^{i_2}\tr\chi.
\end{split}
\end{equation*}
Notice that in the above, we have used a special cancellation to obtain \eqref{K.eqn}. More precisely, we start with
\begin{equation*}
\begin{split}
\nab_3 K+\div\beb+\tr\chib K+\f12 \trchb\div\eta=\sum_{i_1+i_2=1}\q^{i_1+1}\nab^{i_2}\q.
\end{split}
\end{equation*}
Now, note\footnote{In the three displayed equations below, since we need to capture the cancellation, we will not use the schematic notation but track the exact coefficients in each of the terms.} that 
$$\nab_3\left(\K\right)=\nab_3 K-\f{2\Om^{-1}}{|u|^3}.$$
On the other hand, we have
$$\f12 \trchb\div\eta=\f12\trchb\left(\K\right)-\f12\trchb\mu.$$
Therefore, 
\begin{equation*}
\begin{split}
&\quad\ \nab_3 K+\tr\chib K+\f12 \trchb\div\eta\\
&=\nab_3\left(\K\right)+\f32\trchb\left(\K\right)-\f12\trchb\mu\\
&\quad +\trchb\frac{1}{|u|^2}+\f{2\Om^{-1}}{|u|^3}\\
&=\nab_3\left(\K\right)+\f32\trchb\left(\K\right)-\f12\trchb\mu\\
&\quad +\left(\trchb+\f2{|u|}\right)\frac{1}{|u|^2}-\f{2(1-\Om^{-1})}{|u|^3},
\end{split}
\end{equation*}
from which \eqref{K.eqn} follows.
Commuting the equations with $\nab$ for i times, we have 
\begin{equation}\label{eqn.sigma}
\nab_3\nab^i\sigmac+\div^*\nab^i\beb+\f{3+i}{2}\tr\chib\nab^i\sigmac=F_{1,i},
\end{equation}
where $F_{1,i}$ is given by
\begin{equation*}
\begin{split}
F_{1,i}=\sum_{i_1+i_2+i_3=i+1}\nab^{i_1}\q^{i_2+1}\nab^{i_3}\q+\f{1}{|u|}\sum_{i_1+i_2+i_3=i-1}\nab^{i_1}\q^{i_2+1}\nab^{i_3}\sigmac
\end{split}
\end{equation*}
and the equation
\begin{equation}\label{eqn.K}
\nab_3\nab^i\left(\K\right)-\div \nab^i\beb+\f{3+i}{2}\tr\chib\nab^i\left(\K\right)=F_{2,i},
\end{equation}
where $F_{2,i}$ is defined as
\begin{equation*}
\begin{split}
F_{2,i}&=\sum_{i_1+i_2+i_3=i+1}\nab^{i_1}\q^{i_2+1}\nab^{i_3}\q+\frac{1}{|u|}\nab^i\mu\\
&\quad +\frac{1}{|u|}\sum_{i_1+i_2+i_3=i}\nab^{i_1}\q^{i_2+1}\nab^{i_3}\q\\
&\quad +\frac{1}{|u|^2}\sum_{i_1+i_2+i_3=i}\nab^{i_1}\q^{i_2}\nab^{i_3}\left(\trchb+\frac 2{|u|}\right)\\
&\quad +\frac{1}{|u|^3}\sum_{i_1+i_2+i_3=i}\nab^{i_1}\q^{i_2}\nab^{i_3}(1-\Om^{-1})\\
&\quad +\sum_{i_1+i_2+i_3=i}\nab^{i_1}\q^{i_2+1}\nab^{i_3}\left(\K\right)\\
&\quad +\sum_{i_1+i_2+i_3=i-1}\frac 1{|u|}\nab^{i_1}\q^{i_2+1}\nab^{i_3}\left(\K\right)
\end{split}
\end{equation*}
and also the equation
\begin{equation}\label{eqn.betab}
\nab_4\nab^i\beb-\div \nab^{i}\left(\K\right)-^*\nab\nab^i\sigmac=F_{3,i},
\end{equation}
with $F_{3,i}$ given by
\begin{equation*}
\begin{split}
F_{3,i}&=\sum_{i_1+i_2+i_3+i_4=i+1}\nab^{i_1}\q^{i_2}\nab^{i_3}\left(\trchb+\f2{|u|},\chibh,\omb\right)\nab^{i_4}\p\\
&\quad +\sum_{i_1+i_2+i_3=i}\nab^{i_1}\q^{i_2+1}\nab^{i_3}(K,\sigmac)\\
&\quad +\frac{1}{|u|}\sum_{i_1+i_2+i_3=i+1}\nab^{i_1}\q^{i_2}\nab^{i_3}\trch.
\end{split}
\end{equation*}

Using the equations \eqref{eqn.sigma}, \eqref{eqn.K} and \eqref{eqn.betab}, we can derive the energy estimates.
More precisely, using Proposition \ref{intbyparts34} and the equation \eqref{eqn.betab}, we have
\begin{align}\label{est.betab}
&\quad\ \f12\int_{\Hb_{\ub}{(1,u)}}\bigg( u^{i+2}\nab^i\beb\bigg)^2\\
\notag &= \f12\int_{\Hb_{0}{(1,u)}}\bigg( u^{i+2}\nab^i\beb\bigg)^2+\int_{D_{u,\ub}}<u^{i+2}\nab^{i}\beb,u^{i+2}\nab_4\nab^{i}\beb>_{\gamma}\\
\notag &\quad -\int_{D_{u,\ub}}\left(\omega-\f12\tr\chi\right)\left( u^{i+2}\nab^i\beb\right)^2\\
\notag &= \f12\int_{\Hb_{0}{(1,u)}}\bigg( u^{i+2}\nab^i\beb\bigg)^2+\int_{D_{u,\ub}}<u^{i+2}\nab^{i}\beb,u^{i+2}F_{3,i}>_{\gamma}\\
\notag &\quad +\int_{D_{u,\ub}}\bigg<u^{i+2}\nab^{i}\beb,u^{i+2}\left(\nab \nab^i\left(\K\right)+{^{*}\nab}\nab^{i}\sigmac\right)\!\!\bigg>_{\gamma}.
\end{align}
Here, we have abused notation to drop the term 
$$-\int_{D_{u,\ub}}\left(\omega-\f12\tr\chi\right)\left( u^{i+2}\nab^i\beb\right)^2$$
since it has the same schematic form as one of the terms represented by 
$$\int_{D_{u,\ub}}<u^{i+2}\nab^{i}\beb,u^{i+2}F_{3,i}>_{\gamma}.$$
Now, applying Proposition \ref{intbyparts3} and equation \eqref{eqn.sigma}, we obtain
\begin{align}\label{est.sigma}
&\quad\ \f12\int_{H_{u}{(0,\ub)}}\left( u^{i+2}\nab^i\sigmac\right)^2\\
\notag &=\f12\int_{H_1{(0,\ub)}}\left( u^{i+2}\nab^i\sigmac\right)^2\\
\notag &\quad +\int_{D_{u,\ub}}\left<u^{i+2}\nab^{i}\sigmac,u^{i+2}\left(\nab_3+\f{3+i}2\trchb\right)\nab^{i}\sigmac\right>_{\gamma}\\
\notag &\quad -\f12\int_{D_{u,\ub}}f\bigg( u^{i+2}\nab^i\sigmac\bigg)^2\\
\notag &=\f12\int_{H_1{(0,\ub)}}\bigg( u^{i+2}\nab^i\sigmac\bigg)^2+\int_{D_{u,\ub}}<u^{i+2}\nab^{i}\sigmac,u^{i+2}F_{1,i}>_{\gamma}\\
\notag &\quad -\int_{D_{u,\ub}}<u^{i+2}\nab^{i}\sigmac,u^{i+2}(\div {^{*}}\nab^{i}\betab)>_{\gamma}-\f12\int_{D_{u,\ub}}f\bigg( u^{i+2}\nab^i\sigmac\bigg)^2.
\end{align}
Similarly, using Proposition \ref{intbyparts3} and the equation \eqref{eqn.K}, we get
\begin{align}\label{est.K}
&\quad\ \f12\int_{H_{u}{(0,\ub)}}\left( u^{i+2}\nab^i\left(\K\right)\right)^2\\
\notag &=\f12\int_{H_1{(0,\ub)}}\left( u^{i+2}\nab^i(K-1)\right)^2\\
\notag &\quad +\int_{D_{u,\ub}}\left<u^{i+2}\nab^{i}\left(\K\right),u^{i+2}F_{2,i}\right>_{\gamma}\\
\notag &\quad +\int_{D_{u,\ub}}\left<u^{i+2}\nab^{i}\left(\K\right),u^{i+2}(\div \nab^{i}\betab)\right>_{\gamma}\\
\notag &\quad -\f12\int_{D_{u,\ub}}f\bigg( u^{i+2}\nab^i\left(\K\right)\bigg)^2.
\end{align}
Now, we can integrate by parts on the spheres $\S$ using Proposition \ref{intbypartssph} to show that the sum of the terms with highest order angular derivatives in~\eqref{est.betab}, \eqref{est.sigma} and \eqref{est.K} cancel up to a lower order error term:
\begin{align}\label{est.ang}
&\quad\ \int_{D_{u,\ub}}\left<u^{i+2}\nab^{i}\beb,u^{i+2}\left(\nab \nab^i\left(\K\right)+{^{*}\nab}\nab^{i}\sigmac\right)\right>_{\gamma}\\
\notag &\quad -\int_{D_{u,\ub}}<u^{i+2}\nab^{i}\sigmac,u^{i+2}(\div {^{*}}\nab^{i}\betab)>_{\gamma}\\
\notag &\quad +\int_{D_{u,\ub}}\left<u^{i+2}\nab^{i}\left(\K\right),u^{i+2}(\div \nab^{i}\betab)\right>_{\gamma}\\
\notag &\ls  \left\|u^{2i+4}\nab^i\left(\K,\sigmac\right)\q\nab^i\betab\right\|_{L^1_uL^1_{\ub}L^1(\S)}\\
\notag &\ls \frac{\de^{\f32}\at b^{\f14}}{|u|^{\f32}}\!\left\|u^{i+2}\nab^i\!\left(\!\K,\sigmac\!\right)\right\|_{L^{\infty}_uL^2_{\ub}L^2(\S)}\!\!\!\|u^{i+2}\nab^i\betab\|_{L^{\infty}_{\ub}L^2_uL^2(\S)},\hspace{-1em}
\end{align}
where in the last line we have used the bootstrap assumption \eqref{BA.1}.
Therefore, adding the identities \eqref{est.betab}, \eqref{est.sigma}, \eqref{est.K}, using \eqref{est.ang} and the bound for $f$ in Proposition \ref{intbyparts3}, we obtain

\begin{align}\label{main.ee.0}
&\quad\ \left\|u^{i+2}\nab^{i}\left(\K,\sigmac\right)\right\|_{L^2_{\ub}L^2(\S)}^2+\|u^{i+2}\nab^{i}\beb\|_{L^2_u L^2(\S)}^2\\
\notag &\ls \left\|\nab^{i}\left(\K,\sigmac\right)\right\|_{L^2_{\ub}L^2(S_{0,\ub})}^2+\|u^{2i+4} \nab^i\sigmac F_{1,i}\|_{L^1_uL^1_{\ub}L^1(\S)}^2\\
\notag &\quad +\left\|u^{2i+4}\nab^i \left(\K\right) F_{2,i}\right\|_{L^1_uL^1_{\ub}L^1(\S)}\\
\notag &\quad +\|u^{2i+4}\nab^i\beb F_{3,i}\|_{L^1_{\ub}L^1_uL^1(\S)}\\
\notag &\quad +\left\|u^{2i+4} \f{\de\at b^{\f14}}{|u|^2}\nab^i\left(\K,\sigmac\right)\nab^i\left(\K,\sigmac\right)\right\|_{L^1_uL^1_{\ub}L^1(\S)}\\
\notag &\quad +\frac{\de^{\f32}\at b^{\f14}}{|u|^{\f32}}\left\|u^{i+2}\nab^i\left(\K,\sigmac\right)\right\|_{L^{\infty}_uL^2_{\ub}L^2(\S)}\\
\notag &\qquad\times\|u^{i+2}\nab^i\betab\|_{L^{\infty}_{\ub}L^2_uL^2(\S)}
\end{align}
For the second, third and fourth terms, we can apply Cauchy-Schwarz in either the $H$ or the $\Hb$ hypersurface so that the terms 
$$\left\|u^{i+2}\nab^{i}\left(\K,\sigmac\right)\right\|_{L^2_{\ub}L^2(\S)}$$
and $\|u^{i+2}\nab^{i}\beb\|_{L^2_u L^2(\S)}$ can be absorbed to the left and we only need to bound the weighted norms of $F_{1,i}$, $F_{2,i}$ and $F_{3,i}$. For the fifth term, noticing that 
$$\left\|\f{\de\at b^{\f14}}{|u|^2}\right\|_{L^1_u}\ls \f{1}{b^{\f34}},$$
we see that it can be controlled by Gronwall's inequality. For the final term, since
$$\left\|\frac{\de^{\f32}\at b^{\f14}}{|u|^{\f32}}\right\|_{L^{\infty}_u}\ls \frac{1}{b^{\f54}},$$
the term can be absorbed to the left hand side after using Schwarz's inequality. Therefore, \eqref{main.ee.0} implies that
\begin{equation*}
\begin{split}
&\quad\ \left\|u^{i+2}\nab^{i}\left(\K,\sigmac\right)\right\|_{L^2_{\ub}L^2(\S)}+\|u^{i+2}\nab^{i}\beb\|_{L^2_u L^2(\S)}\\
&\ls \left\|\nab^{i}\left(\K,\sigmac\right)\right\|_{L^2_{\ub}L^2(S_{0,\ub})}+\|u^{i+2} F_{1,i}\|_{L^1_uL^2_{\ub}L^2(\S)}\\
&\quad +\|u^{i+2} F_{2,i}\|_{L^1_uL^2_{\ub}L^2(\S)}+\sum_{i\leq 4}\|u^{i+2} F_{3,i}\|_{L^1_{\ub}L^2_uL^2(\S)}.
\end{split}
\end{equation*}
Summing over $1\leq i\leq 4$ and using the fact that
$$\left\|\nab^{i}(\K,\sigmac)\right\|_{L^2_{\ub}L^2(S_{0,\ub})}\ls \de^{\f12}\at,$$ 
we obtain
\begin{align}\label{main.ee.1}
&\ \ \sum_{1\leq i\leq 4}\left(\left\|u^{i+2}\nab^{i}\left(\K,\sigmac\right)\right\|_{L^2_{\ub}L^2(\S)}+\|u^{i+2}\nab^{i}\beb\|_{L^2_u L^2(\S)}\right)\hspace{-1ex}\\
\notag &\ls \de^{\f12}\at+\sum_{1\leq i\leq 4}(\|u^{i+2} F_{1,i}\|_{L^1_uL^2_{\ub}L^2(\S)}+\|u^{i+2} F_{2,i}\|_{L^1_uL^2_{\ub}L^2(\S)})\\
\notag &\quad +\sum_{1\leq i\leq 4}\|u^{i+2} F_{3,i}\|_{L^1_{\ub}L^2_uL^2(\S)}.
\end{align}

We will estimate the right hand side of \eqref{main.ee.1} term by term.\footnote{We now remark on a convention that we will use in this proof. While it is important that we only have $1\leq i\leq 4$ in the sum in some of the error terms to take advantage of the improved estimates, we will simply write $i\leq 4$ in the terms where this restriction is not necessary.} We first estimate the term $F_{1,i}$. For the first term in $F_{1,i}$, we can assume without loss of generality that $i_1\leq i_3$. We bound separately the contributions where there are at most $4$ derivatives falling on any of the Ricci coefficients, where $5$ derivatives fall on $(\trchb,\chibh,\omb)$ and where $5$ derivatives fall on $(\eta,\etab)$. More precisely, we have
\begin{align}\label{F1.1}
&\quad\ \sum_{i\leq 4}\left\|u^{i+2}\sum_{i_1+i_2+i_3=i+1}\nab^{i_1}\q^{i_2+1}\nab^{i_3}\q\right\|_{L^1_uL^2_{\ub}L^2(\S)}\\
\notag &\ls \sum_{\substack{i_1+i_2\leq 5\\ i_1\leq 2}}\|u^{i_1+i_2+2}\nab^{i_1}\q^{i_2+1} \|_{L^\infty_u L^\infty_{\ub} L^\infty(\S)}\\
\notag &\quad \times\Bigg(\de^{\f12}\sum_{i_3\leq 4}\|u^{i_3+1}\nab^{i_3}\q \|_{L^\infty_uL^\infty_{\ub}L^2(\S)}\|u^{-2} \|_{L^1_u}\\
\notag &\qquad+\de^{\f12}\|u^5\nab^5(\trchb,\chibh,\omb)\|_{L^\infty_{\ub}L^2_uL^2(\S)}\|u^{-1}\|_{L^2_u}\\
\notag &\qquad+\|u^6\nab^5(\eta,\etab)\|_{L^\infty_u L^2_{\ub}L^2(\S)}\|u^{-2}\|_{L^1_u}\Bigg)\\
\notag &\ls \de\at b^{\f14}\left(\frac{\de^{\f32}\at b^{\f14}}{|u|}+\frac{\de^{\f32} a^{\f34}}{|u|}\right)\ls \f{\de^{\f52} a^{\f54} b^{\f14}}{|u|},
\end{align}
where we have used the bootstrap assumptions \eqref{BA.1} and \eqref{BA.3}.

For the remaining contributions in $F_{1,i}$, we will prove the slightly more general bound where we allow $(K-\frac 1{|u|^2}, \sigmac)$ in place of $\sigmac$. Using Sobolev embedding in Proposition \ref{Sobolev}, we have
\begin{align}\label{F1.2}
&\quad\ \sum_{i\leq 4}\left\|u^{i+2}\sum_{i_1+i_2+i_3=i}\nab^{i_1}\q^{i_2+1}\nab^{i_3}\left(\K,\sigmac\right)\right\|_{L^1_uL^2_{\ub}L^2(\S)}\\
\notag &\ls \sum_{\substack{i_1+i_2\leq 4\\ i_1\leq 2}}\|u^{i_1+i_2+1}\nab^{i_1}\q^{i_2+1} \|_{L^\infty_u L^\infty_{\ub} L^2(\S)}\\
\notag &\quad \times\sum_{i_3\leq 4}\left\|u^{i_3+2}\nab^{i_3}\left(\K,\sigmac\right)\right\|_{L^\infty_uL^2_{\ub}L^2(\S)}\|u^{-2} \|_{L^1_u}\\
\notag &\ls \de\at b^{\f14} \delta^{\f32} a^{\f34} b^{\f14}\frac{1}{|u|}\ls \f{\delta^{\f52} a^{\f54} b^{\f12}}{|u|}.
\end{align}
where we have used Proposition \ref{product} and the bootstrap assumption \eqref{BA.3}.

The final term in $F_{1,i}$ can be controlled in a similar fashion as \eqref{F1.2}:
\begin{align}\label{F1.3}
&\quad\ \sum_{i\leq 4}\left\|u^{i+1}\sum_{i_1+i_2+i_3=i-1}\nab^{i_1}\q^{i_2+1}\nab^{i_3}\left(\K,\sigmac\right)\right\|_{L^1_uL^2_{\ub}L^2(\S)}\\
\notag &\ls \sum_{\substack{i_1+i_2\leq 3\\ i_1\leq 2}}\|u^{i_1+i_2+1}\nab^{i_1}\q^{i_2+1} \|_{L^\infty_u L^\infty_{\ub} L^2(\S)}\\
\notag &\quad \times\sum_{i_3\leq 3}\left\|u^{i_3+2}\nab^{i_3}\left(\K,\sigmac\right) \right\|_{L^\infty_uL^2_{\ub}L^2(\S)}\|u^{-2} \|_{L^1_u}\\
\notag &\ls \de\at b^{\f14} \delta^{\f32} a^{\f34} b^{\f14}\frac{1}{|u|}\ls \f{\delta^{\f52} a^{\f54} b^{\f12}}{|u|}.
\end{align}

We now move to the estimates for $F_{2,i}$. Notice that the first, sixth and seventh terms are already estimated above in \eqref{F1.1}, \eqref{F1.2} and \eqref{F1.3}. For the second term, we need to use the improved estimates for $\nab^i\mu$ derived in Proposition \ref{eta.5.bd}. In particular, we need to use the fact that $i\geq 1$.
\begin{equation*}
\begin{split}
&\quad\ \sum_{1\leq i\leq 4}\|u^{i+1}\nab^{i}\mu\|_{L^1_uL^2_{\ub}L^2(\S)}\\
&\ls \delta^{\f12}\sum_{1\leq i\leq 4}\|u^{i+3}\nab^{i}\mu \|_{L^\infty_uL^\infty_{\ub}L^2(\S)}\|u^{-2} \|_{L^1_u}\\
&\ls \frac{\de^{\f52} a^{\f54}b^{\f14}}{|u|}.
\end{split}
\end{equation*}
For the third term in $F_{2,i}$, we have
\begin{equation*}
\begin{split}
&\quad\ \sum_{i\leq 4}\|u^{i+1}\sum_{i_1+i_2+i_3=i}\nab^{i_1}\q^{i_2+1}\nab^{i_3}\q\|_{L^1_uL^2_{\ub}L^2(\S)}\\
&\ls \de^{\frac 12}\sum_{\substack{i_1+i_2\leq 4\\ i_1\leq 2}}\|u^{i_1+i_2+2}\nab^{i_1}\q^{i_2+1} \|_{L^\infty_u L^\infty_{\ub} L^{\infty}(\S)}\\
&\quad\times\sum_{i_3\leq 4}\|u^{i_3+1}\nab^{i_3}\q \|_{L^\infty_uL^\infty_{\ub}L^2(\S)}\|u^{-2} \|_{L^1_u}\\
&\ls \frac{\de^{\f52}a b^{\f12}}{|u|}
\end{split}
\end{equation*}
where we have used Proposition \ref{product} and the bootstrap assumption \eqref{BA.1}.

For the fourth term in $F_{2,i}$, we need to use $i\geq 1$ and apply the improvement in the bounds for $\nab^i(\trchb+\f{2}{|u|})$ from Proposition \ref{trchb.bd}. More precisely, we have
\begin{equation*}
\begin{split}
&\quad\ \sum_{1\leq i\leq 4}\left\|u^{i}\sum_{i_1+i_2+i_3=i}\nab^{i_1}\q^{i_2}\nab^{i_3}\left(\trchb+\f{2}{|u|}\right)\right\|_{L^1_uL^2_{\ub}L^2(\S)}\\
&\ls \de^{\frac 12}\sum_{1\leq i\leq 4}\left\|u^{i+\f32}\nab^i \left(\trchb+\f{2}{|u|}\right)\right\|_{L^\infty_u L^\infty_{\ub} L^2(\S)}\|u^{-\f32} \|_{L^1_u}\\
&\quad +\de^{\frac 12}\sum_{i_1+i_2\leq 2}\|u^{i_1+i_2+2}\nab^{i_1}\q^{i_2+1} \|_{L^\infty_u L^\infty_{\ub} L^{\infty}(\S)}\\
&\qquad\times\sum_{i_3\leq 4}\|u^{i_3+1}\nab^{i_3}\q \|_{L^{\infty}_uL^{\infty}_{\ub}L^2(\S)}\|u^{-2} \|_{L^1_u}\\
&\ls \frac{\de^2a^{\f34}}{|u|^{\f12}}+\frac{\de^{\f52}ab^{\f12}}{|u|},
\end{split}
\end{equation*}
where in addition to using Proposition \ref{product}, we have used Proposition \ref{trchb.bd} for the first term and the bootstrap assumptions \eqref{BA.1} for the second term.

For the fifth term in $F_{2,i}$, we also need to use $i\geq 1$ and apply the improvement in the bounds for $\nab^i(1-\Om^{-1})$ from Proposition \ref{Om.bd}. More precisely, we have
\begin{equation*}
\begin{split}
&\quad\ \sum_{1\leq i\leq 4}\left\|u^{i-1}\sum_{i_1+i_2+i_3=i}\nab^{i_1}\q^{i_2}\nab^{i_3}(1-\Om^{-1})\right\|_{L^1_uL^2_{\ub}L^2(\S)}\\
&\ls \de^{\f12}\sum_{1\leq i\leq 4}\|u^{i+\f12}\nab^i \log\Om\|_{L^\infty_u L^\infty_{\ub} L^2(\S)}\|u^{-\f32} \|_{L^1_u}\\
&\quad+\de^{\f12}\sum_{i_1+i_2\leq 2}\|u^{i_1+i_2+2}\nab^{i_1}\q^{i_2+1} \|_{L^\infty_u L^\infty_{\ub} L^{\infty}(\S)}\\
&\qquad\times\sum_{i_3\leq 3}\|u^{i_3+1}\nab^{i_3}\q \|_{L^{\infty}_uL^{\infty}_{\ub}L^2(\S)}\|u^{-2} \|_{L^1_u}\\
&\quad +\de^{\f12}\sum_{i_1+i_2\leq 4}\|u^{i_1+i_2+1}\nab^{i_1}\q^{i_2+1}\|_{L^\infty_u L^\infty_{\ub} L^2(\S)}\\
&\qquad\times\|u(1-\Om^{-1})\|_{L^\infty_u L^\infty_{\ub}L^\infty(\S)}\|u^{-2}\|_{L^1_u}\\
&\ls \frac{\de^2a^{\f34}}{|u|^{\f12}}+\frac{\de^{\f52}ab^{\f12}}{|u|}.
\end{split}
\end{equation*}
Here, we have used Propositions \ref{Omega}, \ref{product} and \ref{Om.bd}.

We now estimate the contributions from $F_{3,i}$. For the first term, we have
\begin{equation*}
\begin{split}
&\quad\ \sum_{i\leq 4}\left\|u^{i+2}\sum_{i_1+i_2+i_3+i_4=i+1}\nab^{i_1}\q^{i_2}\nab^{i_3}\left(\trchb+\f2{|u|},\chibh,\omb\right)\nab^{i_4}\p\right\|_{L^1_{\ub}L^2_{u}L^2(\S)}\\
&\ls \delta\sum_{\substack{i_1+i_2\leq 5\\ i_1\leq 2}}\|u^{i_1+i_2+2}\nab^{i_1}\q^{i_2+1} \|_{L^\infty_u L^\infty_{\ub} L^\infty(\S)}\\
&\quad\times\sum_{i_3\leq 4}\|u^{i_3}\nab^{i_3}\p\|_{L^\infty_uL^\infty_{\ub}L^2(\S)}\|u^{-1}\|_{L^2_u}\\
&\quad +\delta^{\f12}\|u^2\q \|_{L^\infty_u L^\infty_{\ub} L^\infty(\S)}\|u^5\nab^5\p\|_{L^\infty_uL^2_{\ub}L^2(\S)}\|u^{-1}\|_{L^2_u}\\
&\quad +\de\sum_{\substack{i_1+i_2\leq 5\\i_1\leq 4}}\|u^{i_1+i_2+1}\nab^{i_1}\q^{i_2+1} \|_{L^\infty_uL^\infty_{\ub}L^2(\S)}\\
&\qquad\times\sum_{i_3\leq 2}\|u^{i_3+1}\nab^{i_3}\p\|_{L^\infty_uL^\infty_{\ub}L^\infty(\S)}\|u^{-1} \|_{L^2_u}\\
&\quad +\de\left\|u^5\nab^5\left(\trchb+\f2{|u|},\chibh,\omb\right)\right\|_{L^\infty_{\ub}L^2_uL^2(\S)}\|u\p\|_{L^\infty_u L^\infty_{\ub}L^\infty(\S)}\\
&\ls \frac{\de^2 a b^{\f14}}{|u|^{\f12}},
\end{split}
\end{equation*}
where we have used Proposition \ref{product}, the bootstrap assumption \eqref{BA.3}, as well as the bound for $\nab^i\p$ derived in Proposition \ref{p.bd}.

For the second term, we use Sobolev embedding (Theorem \ref{Sobolev}) to get
\begin{equation*}
\begin{split}
&\quad\ \sum_{i\leq 4}\left\|u^{i+2}\sum_{i_1+i_2+i_3=i}\nab^{i_1}\q^{i_2+1}\nab^{i_3}(K,\sigmac)\right\|_{L^1_{\ub}L^2_{u}L^2(\S)}\\
&\ls \delta^{\f12}\sum_{i_1+i_2\leq 4}\|u^{i_1+i_2+1}\nab^{i_1}\q^{i_2+1} \|_{L^\infty_u L^\infty_{\ub} L^2(\S)}\\
&\quad\times\sum_{i_3\leq 4}\left\|u^{i_3+2}\nab^{i_3}\left(\K,\sigmac\right)\right\|_{L^\infty_uL^2_{\ub}L^2(\S)}\|u^{-2}\|_{L^2_u}\\
&\quad +\delta\sum_{i_1+i_2\leq 4}\|u^{i_1+i_2+1}\nab^{i_1}\q^{i_2+1} \|_{L^\infty_u L^\infty_{\ub} L^2(\S)}\|u^{-1}\|_{L^2_u}\\
&\ls \frac{\de^3 a^{\f54} b^{\f12}}{|u|^{\f32}}+\f{\de^2\at b^{\f14}}{|u|^{\f12}},
\end{split}
\end{equation*}
where we have used Proposition \ref{product} and the bootstrap assumption \eqref{BA.3}.

For the third term of $F_{i,3}$, we need to use the improved bounds for $\nab^i\trch$ for $i\geq 1$ given by Propositions \ref{trch.bd} and \ref{chi.5.bd}. More precisely, we have the estimate
\begin{equation*}
\begin{split}
&\quad\ \sum_{i\leq 4}\left\|u^{i+1}\sum_{i_1+i_2+i_3=i+1}\nab^{i_1}\q^{i_2}\nab^{i_3}\trch\right\|_{L^1_{\ub}L^2_{u}L^2(\S)}\\
&\ls \left(\delta^{\f12}\|u^6\nab^5\trch \|_{L^\infty_u L^2_{\ub} L^2(\S)}+\delta\sum_{i\leq 4}\|u^{i+1}\nab^{i}\trch \|_{L^\infty_u L^\infty_{\ub} L^2(\S)}\right)\| u^{-1}\|_{L^2_u}\\
&\quad +\delta\sum_{\substack{i_1+i_2\leq 5\\ i_1\leq 2}}\|u^{i_1+i_2+2}\nab^{i_1}\q^{i_2+1} \|_{L^\infty_u L^\infty_{\ub} L^\infty(\S)}\\
&\qquad\times\sum_{i_3\leq 4}\|u^{i_3}\nab^{i_3}\p\|_{L^\infty_uL^\infty_{\ub}L^2(\S)}\|u^{-1}\|_{L^2_u}\\
&\quad +\de\sum_{\substack{i_1+i_2\leq 5\\i_1\leq 4}}\|u^{i_1+i_2+1}\nab^{i_1}\q^{i_2+1} \|_{L^\infty_uL^\infty_{\ub}L^2(\S)}\\
&\qquad\times\sum_{i_3\leq 2}\|u^{i_3+1}\nab^{i_3}\p\|_{L^\infty_uL^\infty_{\ub}L^\infty(\S)}\|u^{-1} \|_{L^2_u}\\
&\ls \frac{\de^2 a b^{\f14}}{|u|^{\f12}},
\end{split}
\end{equation*}
where we have used Propositions \ref{trch.bd}, \ref{p.bd} and \ref{chi.5.bd}, as well as applied Proposition \ref{product} to control the product of $\q$.

Returning to \eqref{main.ee.1}, collecting all the above estimates and using the condition $|u|\geq \delta\at b$, we get
$$\sum_{i\leq 4}\left(\left\|u^{i+2}\nab^{i}\left(\K,\sigmac\right)\right\|_{L^2_{\ub}L^2(\S)}+\|u^{i+2}\nab^{i}\beb\|_{L^2_u L^2(\S)}\right)\ls \delta^{\f32}a^{\f34}.\vspace{-1em}$$

\end{proof}
We now turn to the remaining energy estimates, i.e., we derive the bounds for $\|u^{i+1}\nab^i\beta\|_{L^{\infty}_uL^{2}_{\ub}L^2(\S)}$ and $\|u^{i+1}\nab^i(\K,\sigmac)\|_{L^{\infty}_{\ub}L^{2}_{u}L^2(\S)}$ for $i\leq 4$:
\begin{proposition} \label{EE.1}
Under the assumptions of Theorem \ref{main.thm} and the bootstrap assumptions \eqref{BA.1}, \eqref{BA.2}, \eqref{BA.3} and \eqref{BA.4}, we have
\begin{equation*}
\begin{split}
\sum_{i\leq 4}\!\left(\!\|u^{i+1}\nab^i\beta\|_{L^{\infty}_uL^{2}_{\ub}L^2(\S)}\!+\!\left\|u^{i+1}\nab^i\!\left(\K,\sigmac\right)\right\|_{L^{\infty}_{\ub}L^{2}_{u}L^2(\S)}\right)\!\ls \de^{\f12}\at.
\end{split}
\end{equation*}

\end{proposition}

\begin{proof}
As in the proof of Proposition \ref{EE.2}, we derive the energy estimates from the Bianchi equations. We consider the three schematic Bianchi equations below. First, the equation for $\nab_3\beta$:
\begin{align}\label{eqn.beta}
&\quad\ \nab_3\beta-\nab\left(\K\right)-^*\nab\sigmac+\tr\chib\beta\\
\notag &=\q(K,\sigmac)+\sum_{i_1+i_2=1}\q^{i_1+1}\nab^{i_2}\p+\p\nab(\trchb,\chibh)+\f1{|u|}\q\p+\f{1}{|u|}\nab\trch.
\end{align}
We also have the equation for $\nab_4\sigmac$
\begin{equation}\label{eqn.sigma.2}
\nab_4\sigmac+\div^*\beta=\p\sigmac+\p\nab\etb+\sum_{i_1+i_2=1}\q^{i_1+1}\nab^{i_2}\p
\end{equation}
and the equation for $\nab_4(\K)$
\begin{align}\label{eqn.K.2}
\nab_4\left(\K\right)+\div\beta&=\p\left(\K,\sigmac\right)+\p\nab\etb\\
\notag &\quad +\sum_{i_1+i_2=1}\q^{i_1+1}\nab^{i_2}\p+\f1{|u|^2}\p.
\end{align}
We commute the above equations with $\nab^i$. From \eqref{eqn.beta}, we obtain
\begin{equation}\label{eqn.G1}
\nab_3\nab^i\beta-\nab\nab^i\left(\K\right)-^*\nab\nab^i\sigmac+\f{2+i}{2}\tr\chib\nab^i\beta=G_{1,i}
\end{equation}
where $G_{1,i}$ is given by
\begin{equation*}
\begin{split}
G_{1,i}&=\q\nab^5\p+\f{1}{|u|}\nab^5\trch+\p\nab^5(\chibh,\tr\chib)\\
&\quad +\sum_{i_1+i_2+i_3=i}\nab^{i_1}\q^{i_2+1}\nab^{i_3}\left(\K,\sigmac\right)
+\sum_{\substack{i_1+i_2+i_3=i+1\\i_1, i_3\leq i}}\nab^{i_1}\q^{i_2+1}\nab^{i_3}\p\\
&\quad +\f{1}{|u|}\sum_{i_1+i_2+i_3=i}\nab^{i_1}\q^{i_2+1}\nab^{i_3}\p
+\f{1}{|u|^2}\sum_{i_1+i_2=i}\nab^{i_1}\q^{i_2+1}.
\end{split}
\end{equation*}
Using \eqref{eqn.sigma.2}, we get
\begin{equation}\label{eqn.G2}
\nab_4\nab^i\sigmac+\div^*\nab^i\beta=G_{2,i},
\end{equation}
where
\begin{equation*}
\begin{split}
G_{2,i}&=\q\nab^5\p+\p\nab^5\etb+\sum_{i_1+i_2+i_3+i_4=i}\nab^{i_1}\q^{i_2}\nab^{i_3}\p\nab^{i_4}\sigmac\\
&\quad +\sum_{\substack{i_1+i_2+i_3=i+1\\i_1, i_3\leq i}}\nab^{i_1}\q^{i_2+1}\nab^{i_3}\p.
\end{split}
\end{equation*}
Also, by \eqref{eqn.K.2}, we have
\begin{equation}\label{eqn.G3}
\nab_4\nab^i\left(\K\right)+\div\nab^i\beta=G_{3,i},
\end{equation}
where
\begin{equation*}
\begin{split}
G_{3,i}&=\q\nab^5\p+\sum_{i_1+i_2+i_3+i_4=i}\nab^{i_1}\q^{i_2}\nab^{i_3}\p\nab^{i_4}\left(\K,\sigmac\right)\\
&\quad +\sum_{\substack{i_1+i_2+i_3=i+1\\i_1, i_3\leq i}}\nab^{i_1}\q^{i_2+1}\nab^{i_3}\p+\f{1}{|u|^2}\sum_{i_1+i_2=i}\nab^{i_1}\q^{i_2+1}.
\end{split}
\end{equation*}

As in the proof of Proposition 88, we use equations \eqref{eqn.G1}, \eqref{eqn.G2},\eqref{eqn.G3} and apply Propositions \ref{intbyparts34}, \ref{intbypartssph} and \ref{intbyparts3} to obtain
\begin{align}\label{main.ee.2}
&\quad\ \sum_{i\leq 4}\bigg(\|u^{i+1}\nab^{i}\beta\|^2_{L^2_{\ub}L^2(\S)}+\|u^{i+1}\nab^{i}\left(\K\right)\|^2_{L^2(\Hb)}\\
\notag &\qquad\quad\ +\|u^{i+1}\nab^{i}\sigmac\|^2_{L^2_uL^2(\S)}\bigg)\\
\notag &\ls \sum_{i\leq 4}\|u^{i+1}\nab^{i}\beta\|^2_{L^2_{\ub}L^2(S_{0,\ub})}\\
\notag &\quad  +\sum_{i\leq4}\left(\|u^{i+1}G_{1,i}\|_{L^1_{u}L^2_{\ub}L^2(S_{u,\ub})}+\|u^{i+1}G_{2,i}\|_{L^1_{\ub}L^2_{u}L^2(S_{u,\ub})}\right)\\
\notag &\quad +\sum_{i\leq 4}\|u^{i+1}G_{3,i}\|_{L^1_{\ub}L^2_{u}L^2(S_{u,\ub})}.
\end{align}
Notice that the weight $u^{i+1}$ is dictated by term $\f{2+i}{2}\tr\chib\nab^i\beta$ in \eqref{eqn.G1}. The proof of \eqref{main.ee.2} is otherwise analogous to that of \eqref{main.ee.1} in Proposition 88 and is omitted.

We now estimate each of the terms on the right hand side \eqref{main.ee.2}. We begin with the term $G_{1,i}$. 
Among the terms in $G_{1,i}$, we first bound the contributions where 5 derivatives fall on one of the Ricci coefficients. More precisely, these are the terms
\begin{equation}\label{terms.5.ee.2}
\q\nab^5\p,\quad \f{1}{|u|}\nab^5\trch,\quad \p\nab^5(\chibh,\trchb).
\end{equation}
For the first term in \eqref{terms.5.ee.2}, we have the estimate
\begin{equation*}
\begin{split}
&\quad\ \|u^5\q\nab^5\p\|_{L^1_{u}L^2_{\ub}L^2(S_{u,\ub})}\\
&\ls \|u^5\nab^5\p\|_{L^{\infty}_{u}L^2_{\ub}L^2(S_{u,\ub})}\|\q\|_{L^1_{u}L^{\infty}_{\ub}L^{\infty}(\S)}\\
&\ls \delta^{\f12}a^{\f12}\f{\delta a^{\f12}}{|u|}b^{\f12}\ls \f{\de^{\f32} a b^{\f12}}{|u|},
\end{split}
\end{equation*}
where we have used bootstrap assumptions \eqref{BA.1} and \eqref{BA.3}.
The second term of \eqref{terms.5.ee.2} can be controlled by
\begin{equation*}
\begin{split}
&\quad\ \|u^4\nab^5\trch\|_{L^1_{u}L^2_{\ub}L^2(S_{u,\ub})}\\
&\ls \|u^6\nab^5\trch\|_{L^{\infty}_{u}L^2_{\ub}L^2(S_{u,\ub})}\|\f{1}{|u|^2}\|_{L^1_{u}}\\
&\ls \f{\delta^{\f32}ab^{\f14}}{|u|}.
\end{split}
\end{equation*}
Here, we have used the bound for $\nab^5\trch$ in Proposition \ref{chi.5.bd}. Notice that it is important that Proposition \ref{chi.5.bd} gives a better bound for $\nab^5\trch$ than other $\nab^5\p$ components. Then, turning to the third term in \eqref{terms.5.ee.2}, we use Proposition \ref{p.bd} and the bootstrap assumption \eqref{BA.3} to obtain the following estimate:
\begin{equation*}
\begin{split}
&\quad\ \|u^5\p\nab^5(\chibh,\tr\chib)\|_{L^1_{u}L^2_{\ub}L^2(S_{u,\ub})}\\
&\ls \|u^5\nab^5(\chibh,\tr\chib)\|_{L^{\infty}_{\ub}L^2_{u}L^2(S_{u,\ub})}\|\p\|_{L^2_{u}L^2_{\ub}L^{\infty}(S_{u,\ub})}\\
&\ls \f{\delta a^{\f12}}{|u|^{\f12}}b^{\f14}\f{\delta^{\f12}a^{\f12}}{|u|^{\f12}}\ls \f{\delta^{\f32}a b^{\f14}}{|u|}.
\end{split}
\end{equation*}
After bounding all the highest derivative Ricci coefficient terms, we now move to the term containing $(\K,\sigmac)$. Using Sobolev embedding in Proposition \ref{Sobolev} together with Proposition \ref{product} and the bootstrap assumption \eqref{BA.3}, we have
\begin{equation*}
\begin{split}
&\quad\ \sum_{i\leq 4}\left\|\sum_{i_1+i_2+i_3=i}u^{i+1}\nab^{i_1}\q^{i_2+1}\nab^{i_3}\left(\K,\sigmac\right)\right\|_{L^1_{u}L^2_{\ub}L^2(S_{u,\ub})}\\
&\ls \sum_{i_1+i_2\leq 2}\|u^{i_1+i_2}\nab^{i_1}\q^{i_2+1}\|_{L^2_{\ub}L^2_{u}L^{\infty}(S_{u,\ub})}\\
&\quad\times\sum_{i_3\leq 4}\left\|u^{i_3+1}\nab^{i_3}\left(\K,\sigmac\right)\right\|_{L^{\infty}_{\ub}L^2_{u}L^2(S_{u,\ub})}\\
&\quad +\sum_{i_1+i_2\leq 4}\|u^{i_1+i_2-1}\nab^{i_1}\q^{i_2+1}\|_{L^2_{\ub}L^2_{u}L^{2}(S_{u,\ub})}\\
&\qquad\times\sum_{i_3\leq 2}\left\|u^{i_3+2}\nab^{i_3}\left(\K,\sigmac\right)\right\|_{L^{\infty}_{\ub}L^2_{u}L^{\infty}(S_{u,\ub})}\\
&\ls \delta^{\f12}a^{\f12}b^{\f14}\f{\delta a^{\f12}}{|u|^2}|u|^{\f12}\delta^{\f12}b^{\f14}\ls \f{\delta^2 a b^{\f12}}{|u|^{\f32}}.
\end{split}
\end{equation*}
We then control the Ricci coefficient terms where there are at most 4 derivatives. For each $i\leq 4$, there are three terms to be bounded, namely
\begin{equation}\label{terms.lower.G1}
\begin{split}
&\sum_{\substack{i_1+i_2+i_3=i+1\\i_1, i_3\leq i}}\nab^{i_1}\q^{i_2+1}\nab^{i_3}\p,\\
&\sum_{i_1+i_2+i_3=i}\f{1}{|u|}\nab^{i_1}\q^{i_2+1}\nab^{i_3}\p,\quad \sum_{i_1+i_2=i}\f1{|u|^2}\nab^{i_1}\q^{i_2+1}.
\end{split}
\end{equation}
The first terms in \eqref{terms.lower.G1} can be estimated as follows
\begin{equation*}
\begin{split}
&\quad\ \sum_{i\leq4}\Bigg\|\sum_{\substack{i_1+i_2+i_3=i+1\\i_1, i_3\leq i}}u^{i+1}\nab^{i_1}\q^{i_2+1}\nab^{i_3}\p\Bigg\|_{L^1_{u}L^2_{\ub}L^2(S_{u,\ub})}\\
&\ls \sum_{i_1+i_2\leq 2}\|u^{i_1+i_2}\nab^{i_1}\q^{i_2+1}\|_{L^1_{u}L^{\infty}_{\ub}L^{\infty}(\S)}\sum_{i_3\leq 4}\|u^{i_3}\nab^{i_3}\p\|_{L^{\infty}_{u}L^2_{\ub}L^2(\S)}\\
&\quad +\sum_{i_1+i_2\leq 4}\|u^{i_1+i_2-1}\nab^{i_1}\q^{i_2+1}\|_{L^1_{u}L^{\infty}_{\ub}L^{2}(\S)}\sum_{i_3\leq 2}\|u^{i_3+1}\nab^{i_3}\p\|_{L^{\infty}_{u}L^2_{\ub}L^{\infty}(\S)}\\
&\ls \f{\delta a^{\f12}b^{\f14}}{|u|}\delta^{\f12}a^{\f12}\ls \f{\delta^{\f32} a b^{\f14}}{|u|},
\end{split}
\end{equation*}
where we have used Propositions \ref{product} and \ref{p.bd} and the bootstrap assumption~\eqref{BA.1}.

The second term in \eqref{terms.lower.G1} can be controlled in a very similar fashion:
\begin{equation*}
\begin{split}
&\quad\ \sum_{i\leq4}\left\|\sum_{i_1+i_2+i_3=i}u^{i}\nab^{i_1}\q^{i_2+1}\nab^{i_3}\p\right\|_{L^1_{u}L^2_{\ub}L^2(S_{u,\ub})}\\
&\ls \sum_{i_1+i_2\leq 2}\|u^{i_1+i_2}\nab^{i_1}\q^{i_2+1}\|_{L^1_{u}L^{\infty}_{\ub}L^{\infty}(\S)}\sum_{i_3\leq 4}\|u^{i_3}\nab^{i_3}\p\|_{L^{\infty}_{u}L^2_{\ub}L^2(\S)}\\
&\quad +\sum_{i_1+i_2\leq 4}\|u^{i_1+i_2-1}\nab^{i_1}\q^{i_2+1}\|_{L^1_{u}L^{\infty}_{\ub}L^{2}(\S)}\sum_{i_3\leq 2}\|u^{i_3+1}\nab^{i_3}\p\|_{L^{\infty}_{u}L^2_{\ub}L^{\infty}(\S)}\\
&\ls \f{\delta a^{\f12}b^{\f14}}{|u|}\delta^{\f12}a^{\f12}\ls \f{\delta^{\f32} a b^{\f14}}{|u|},
\end{split}
\end{equation*}
where we have again used Proposition \ref{product} and \ref{p.bd} and the bootstrap assumption \eqref{BA.1}.

We now turn to the last term of \eqref{terms.lower.G1}, for which we have the following bound:
\begin{equation*}
\begin{split}
&\quad\ \sum_{i\leq4}\left\|\sum_{i_1+i_2=i}u^{i-1}\nab^{i_1}\q^{i_2+1}\right\|_{L^1_{u}L^2_{\ub}L^2(S_{u,\ub})}\\
&\ls \sum_{i_1+i_2\leq 2}\|u^{i_1+i_2-1}\nab^{i_1}\q^{i_2}\|_{L^2_{u}L^{\infty}_{\ub}L^{\infty}(\S)}\sum_{i_3\leq 4}\|u^{i_3}\nab^{i_3}\q\|_{L^{2}_{u}L^2_{\ub}L^2(\S)}\\
&\ls \f{\delta^{\f32} \at b^{\f14}}{|u|},
\end{split}
\end{equation*}
after using Proposition \ref{product} and bootstrap assumption \eqref{BA.1}.

This concludes the estimates for $G_{1,i}$ We now move to the estimates for $G_{2,i}$ and $G_{3,i}$. It is easy to observe that all the terms of $G_{2,i}$ are in fact contained in the expression for $G_{3,i}$. Hence, it suffices to control the terms in $G_{3,i}$. As in the estimates for $G_{1,i}$, we first bound the contributions where there are $5$ derivatives on one of the Ricci coefficients. There are two terms of this type, namely,
$$\q\nab^5\p,\quad \p\nab^5\etab.$$
For the first term, we have
\begin{equation*}
\begin{split}
&\quad\ \|u^5\q\nab^5\p\|_{L^1_{\ub}L^2_{u}L^2(S_{u,\ub})}\\
&\ls \|u^5\nab^5\p\|_{L^{\infty}_{u}L^2_{\ub}L^2(S_{u,\ub})}\|\q\|_{L^{2}_{\ub}L^2_{u}L^{\infty}(S_{u,\ub})}\\
&\ls \delta^{\f12}a^{\f12}b^{\f14}\f{\delta^{\f32}a^{\f12}b^{\f14}}{|u|^{\f32}}\ls \f{\delta^2 a b^{\f12}}{|u|^{\f32}},
\end{split}
\end{equation*}
where we have used the bootstrap assumptions \eqref{BA.1} and \eqref{BA.3}.
The second term can be controlled after using Proposition \ref{p.bd} and the bootstrap assumption \eqref{BA.3}:
\begin{equation*}
\begin{split}
&\quad\ \|u^5\p\nab^5\etb\|_{L^1_{\ub}L^2_{u}L^2(S_{u,\ub})}\\
&\ls\|u^6\nab^5\etb\|_{L^{\infty}_{u}L^2_{\ub}L^2(S_{u,\ub})}\left\|\f{1}{|u|}\p\right\|_{L^2_{u}L^2_{\ub}L^{\infty}(S_{u,\ub})}\\
&\ls \delta^{\f32}a^{\f34}b^{\f14}\f{\delta^{\f12}a^{\f12}}{|u|^{\f32}}\ls \f{\de^2 a^{\f54} b^{\f14}}{|u|^{\f32}}.
\end{split}
\end{equation*}
We now move to the curvature term containing $(\K,\sigmac)$. We have
\begin{equation*}
\begin{split}
&\quad\ \sum_{i\leq4}\left\|\sum_{i_1+i_2+i_3+i_4=i}u^{i+1}\nab^{i_1}\q^{i_2}\nab^{i_3}\p\nab^{i_4}\left(\K,\sigmac\right)\right\|_{L^1_{\ub}L^2_{u}L^2(S_{u,\ub})}\\
&\ls \de\sum_{i_1+i_2\leq 2}\|u^{i_1+i_2-1}\nab^{i_1}\q^{i_2}\|_{L^{\infty}_uL^{\infty}_{\ub}L^{\infty}(\S)}\sum_{i_3\leq 2}\|\nab^{i_3+1}\p\|_{L^{\infty}_{\ub}L^{\infty}_{u}L^{\infty}(S_{u,\ub})}\\
&\quad \times \sum_{i_4\leq 4}\left\|u^{i_4+1}\nab^{i_4}\left(\K,\sigmac\right)\right\|_{L^{\infty}_{\ub}L^2_{u}L^2(S_{u,\ub})}\\
&\quad +\de\sum_{i_1+i_2\leq 2}\|u^{i_1+i_2-1}\nab^{i_1}\q^{i_2}\|_{L^{\infty}_uL^{\infty}_{\ub}L^{\infty}(\S)}\sum_{i_3\leq 4}\|\nab^{i_3}\p\|_{L^{\infty}_{\ub}L^{\infty}_{u}L^{2}(S_{u,\ub})}\\
&\qquad \times \sum_{i_4\leq 2}\left\|u^{i_4+2}\nab^{i_4}\left(\K,\sigmac\right)\right\|_{L^{\infty}_{\ub}L^2_{u}L^{\infty}(S_{u,\ub})}\\
&\quad +\de\sum_{i_1+i_2\leq 4}\|u^{i_1+i_2-2}\nab^{i_1}\q^{i_2}\|_{L^{\infty}_uL^{\infty}_{\ub}L^2(\S)}\sum_{i_3\leq 2}\|\nab^{i_3+1}\p\|_{L^{\infty}_{\ub}L^{\infty}_{u}L^{2}(S_{u,\ub})}\\
&\qquad\times \sum_{i_4\leq 2}\left\|u^{i_4+2}\nab^{i_4}\left(\K,\sigmac\right)\right\|_{L^{\infty}_{\ub}L^2_{u}L^{\infty}(S_{u,\ub})}\\
&\ls \f{\delta}{|u|}\at\cdot\de^{\f12}\at b^{\f14}\ls \f{\delta^{\f32}a b^{\f14}}{|u|},
\end{split}
\end{equation*}
where we have used Sobolev embedding in Proposition \ref{Sobolev} and also Propositions \ref{product} and \ref{p.bd} and the bootstrap assumptions \eqref{BA.1} and \eqref{BA.3}.
We finally turn to the remaining two terms, i.e., the Ricci coefficient terms with at most 4 derivatives. These are the terms
$$\sum_{\substack{i_1+i_2+i_3=i+1\\i_1, i_3\leq i}}\nab^{i_1}\q^{i_2+1}\nab^{i_3}\p,\quad \sum_{i_1+i_2+i_3=i}\f{1}{|u|^2}\nab^{i_1}\q^{i_2}\nab^{i_3}\p$$
for $i\leq 4$.
For the first term, we have
\begin{equation*}
\begin{split}
&\quad\ \sum_{i\leq4}\Bigg\|\sum_{\substack{i_1+i_2+i_3=i+1\\i_1, i_3\leq i}}u^{i+1}\nab^{i_1}\q^{i_2+1}\nab^{i_3}\p\Bigg\|_{L^1_{\ub}L^2_{u}L^2(S_{u,\ub})}\\
&\ls \sum_{i_1+i_2\leq 2}\|u^{i_1+i_2}\nab^{i_1}\q^{i_2+1}\|_{L^1_{\ub}L^{2}_{u}L^{\infty}(\S)}\sum_{i_3\leq 4}\|u^{i_3}\nab^{i_3}\p\|_{L^{\infty}_{\ub}L^{\infty}_{u}L^2(\S)}\\
&\quad +\sum_{i_1+i_2\leq 4}\|u^{i_1+i_2-1}\nab^{i_1}\q^{i_2+1}\|_{L^1_{\ub}L^{2}_{u}L^{2}(\S)}\sum_{i_3\leq 2}\|u^{i_3+1}\nab^{i_3}\p\|_{L^{\infty}_{u}L^{\infty}_{\ub}L^{\infty}(\S)}\\
&\ls \f{\delta^{\f32} a^{\f12}b^{\f14}}{|u|^{\f32}}\delta^{\f12}a^{\f12}\ls \f{\delta^2 ab^{\f14}}{|u|^{\f32}},
\end{split}
\end{equation*}
where we have used Propositions \ref{product} and \ref{p.bd} and the bootstrap assumption \eqref{BA.1}.
Finally, using again Propositions \ref{product} and \ref{p.bd} and the bootstrap assumption \eqref{BA.1}, we obtain
\begin{equation*}
\begin{split}
&\quad\ \sum_{i\leq4}\left\|\sum_{i_1+i_2+i_3=i}u^{i-1}\nab^{i_1}\q^{i_2}\nab^{i_3}\p\right\|_{L^1_{\ub}L^2_{u}L^2(S_{u,\ub})}\\
&\ls \sum_{i_1+i_2\leq 2}\|u^{i_1+i_2-1}\nab^{i_1}\q^{i_2}\|_{L^1_{\ub}L^{2}_{u}L^{\infty}(\S)}\sum_{i_3\leq 4}\|u^{i_3}\nab^{i_3}\p\|_{L^{\infty}_{\ub}L^{\infty}_{u}L^2(\S)}\\
&\quad +\sum_{i_1+i_2\leq 4}\|u^{i_1+i_2-2}\nab^{i_1}\q^{i_2}\|_{L^1_{\ub}L^{2}_{u}L^{2}(\S)}\sum_{i_3\leq 2}\|u^{i_3+1}\nab^{i_3}\p\|_{L^{\infty}_{\ub}L^{\infty}_{u}L^{\infty}(\S)}\\
&\ls \f{\delta^{\f32} a^{\f12}b^{\f14}}{|u|^{\f32}}\delta^{\f12}a^{\f12}\ls \f{\delta^2 ab^{\f14}}{|u|^{\f32}}.
\end{split}
\end{equation*}
We have thus estimated all of the error terms. Notice that since $|u|\geq \de \at b$, all the error terms can be controlled by 
$$\ls \de^{\f12}\at.$$
Returning to \eqref{main.ee.2}, we therefore conclude the proof of the proposition.
\end{proof}

With Propositions \ref{EE.2} and \ref{EE.1}, the only remaining energy estimate is that for $\K$ in the case $i=0$. To obtain the desired bounds, we will simply integrate the estimates we obtained in Proposition \ref{K.bd} to get the following proposition:
\begin{proposition}\label{EE.lower}
Under the assumptions of Theorem \ref{main.thm} and the bootstrap assumptions \eqref{BA.1}, \eqref{BA.2}, \eqref{BA.3} and \eqref{BA.4}, we have
\begin{equation*}
\begin{split}
\left\|u^{2}\left(\K\right)\right\|_{L^{\infty}_uL^{2}_{\ub}L^2(S)}
\ls  \de^{\f32} a^{\f34}.
\end{split}
\end{equation*}
\end{proposition}
\begin{proof}
This follows immediately from Proposition \ref{K.bd}:
$$\left\|u^{2}\left(\K\right)\right\|_{L^{\infty}_uL^{2}_{\ub}L^2(S)}\ls \de^{\f12}\cdot\delta\at (1+\tilde{\M O}_{5,2}+\M R)$$
and the bootstrap assumption \eqref{BA.3}.
\end{proof}

Combining Propositions \ref{EE.2}, \ref{EE.1} and \ref{EE.lower}, we have thus obtained
\begin{equation}\label{R.final}
\mathcal R\ls 1.
\end{equation}
Substituting the bound \eqref{R.final} into Proposition \ref{O52.bd}, we have
\begin{equation}\label{O52.final}
\tilde{\mathcal O}_{5,2}\ls 1.
\end{equation}
\eqref{R.final} and \eqref{O52.final} together improve over the bootstrap assumption \eqref{BA.3} for $b$ sufficiently large. Now, using \eqref{R.final} and \eqref{O52.final} together with Propositions~\ref{p.bd} and \ref{q.bd} and the Sobolev embedding in Proposition \ref{Sobolev}, we obtain
\begin{equation*}
\sum_{i\leq 4}\frac{1}{\delta a^{\frac 12}}\|u^{i+1}\nab^i\q\|_{L^2(S_{u,\ub})}+\sum_{i\leq 2}\frac{1}{\delta a^{\frac 12}}\|u^{i+2}\nab^i\q\|_{L^{\infty}(S_{u,\ub})} \ls 1
\end{equation*}
and
\begin{equation*}
\sum_{i\leq 4}\frac{1}{a^{\frac 12}}\|u^{i}\nab^i\p\|_{L^2(S_{u,\ub})}+\sum_{i\leq 2}\frac{1}{a^{\frac 12}}\|u^{i+1}\nab^i\p\|_{L^{\infty}(S_{u,\ub})} \ls 1,
\end{equation*}
which improve over \eqref{BA.1} and \eqref{BA.2} respectively, after choosing $b$ to be sufficiently large. Finally, recall from Proposition \ref{K.bd} that we have
\[
 \sum_{i\leq 3}\left\|u^{i+1}\nab^{i}\left(\K\right)\right\|_{L^\infty_uL^\infty_{\ub}L^2(\S)} \ls \f{1}{b^{\f34}},
\]
which improves over \eqref{BA.4} for $b$ sufficiently large.
We have thus recovered all the bootstrap assumptions and showed that the bound
$$\mathcal O, \tilde{\mathcal O}_{5,2}, \mathcal R\ls 1$$
holds.

This concludes the proof of Theorem \ref{main.thm}.

\section{Formation of trapped surfaces}\label{sec.formation}

In the section, we will prove Theorem \ref{trapped.thm}, thus concluding the proof of Theorem \ref{main.thm.intro}. The proof will make use of the bounds in Theorem \ref{main.thm} and follows from an ODE argument as in \cite{Chr:book}. We first need to obtain some refined estimates compared to that in the proof Theorem \ref{main.thm}.

First, using the bound for $\om$ proved in Proposition \ref{om.bd} instead of applying the bootstrap assumption, we can revisit the proof of Proposition \ref{Omega} to obtain the following improved estimate for $\Omega$:
\begin{proposition}\label{Om.trapped.improved}
Under the assumptions of Theorem \ref{main.thm}, we have
$$\|\Om^{-1}-1\|_{L^\infty(\S)}\ls \f{\de\at}{|u|}.$$
\end{proposition}
\begin{proof}
Using \eqref{Omegatransport}, we have
$$\|\Om^{-1}-1\|_{L^\infty(\S)}\ls \int_0^{\ub}\|\om\|_{L^\infty(S_{u,\ub'})}d\ub'\ls \f{\delta\at}{|u|},$$
where in the last step we have used the estimate for $\om$ derived in Theorem~\ref{main.thm}.
\end{proof}
We then also need the following improved estimate for $\trch$. More precisely, we show that while the bound in Proposition \ref{trch.bd} is in general sharp, the main contribution comes from the $|\chih|^2_{\gamma}$ term. This will allow us to obtain an upper bound for $|\trch(u,\ub,\theta^1,\theta^2)|$ depending on the integral of $|\chih|_{\gamma}^2$ along the characteristic in the $(\theta^1,\theta^2)$ direction. This improvement in the bound for $\trch$ will in turn enable us to prove a lower bound for the integral of $|\chih|_{\gamma}^2$ in the proof of Theorem \ref{trapped.thm}.
\begin{proposition}\label{trch.trapped.improved}
Under the assumptions of Theorem \ref{main.thm}, we have
$$|\trch(u,\ub,\theta^1,\theta^2)|\ls \f{1}{|u|}+\int_0^{\de} |\chih|_{\gamma}^2(u,\ub',\theta^1,\theta^2) d\ub'$$
for every $(\theta^1,\theta^2)$.
\end{proposition}
\begin{proof}
Fix $(\theta^1,\theta^2)$. Recall that 
$$\nab_4\trch+|\chih|^2=\p\trch.$$
Integrating this equation and recalling that initially $\trch=\f2{|u|}$, we obtain
\begin{align*}
|\trch(u,\ub,\theta^1,\theta^2)| &\ls \f{1}{|u|}+\f{\delta\at}{|u|}\sup_{\ub'\leq \ub}|\trch(u,\ub',\theta^1,\theta^2)|\\
&\quad +\int_0^{\de} |\chih|_{\gamma}^2(u,\ub',\theta^1,\theta^2) d\ub'.
\end{align*}
Since $\f{\delta\at}{|u|}\leq \f{1}{b}$, we can choose $b$ to be sufficiently large to obtain
$$|\trch(u,\ub,\theta^1,\theta^2)|\ls \f{1}{|u|}+\int_0^{\de} |\chih|_{\gamma}^2(u,\ub',\theta^1,\theta^2) d\ub'.\vspace{-2em}$$
\end{proof}
With these estimates, we are now ready to prove Theorem \ref{trapped.thm}.
\begin{proof}[Proof of Theorem \ref{trapped.thm}]

In order to show that $S_{\delta\at b,\delta}$ is a trapped surface, we need to prove that $\trch<0$ and $\trchb<0$ everywhere on $S_{\delta\at b,\delta}$. The fact that $\trchb<0$ follows directly from the estimates we have proved. More precisely, recall from Theorem \ref{main.thm} that we have the following bound for $\trchb$ in the region $\delta\at b\leq |u|\leq 1$:
\begin{equation*}
\left\|\tr\chib+\f{2}{|u|}\right\|_{L^{\infty}(S_{u,\ub})}\ls \f{\delta a^{\f12}}{|u|^2}.
\end{equation*}
In particular,
\begin{equation}\label{trchb.neg}
\trchb<0
\end{equation}
everywhere on the sphere $S_{b\delta a^{\f12},\delta}$. 

It now remains to show that $\trch<0$ pointwise on $S_{b\delta a^{\f12},\delta}$. Fix a point $p=(\theta^1,\theta^2)\in S_{1,0}$. We will show that $\trch(u=\delta\at b,\ub = \delta,\theta^1,\theta^2)<0$. To achieve this, we make the following bootstrap assumption for every $u\in [\delta\at b, 1]$:
\begin{equation}\label{BA.trapped}
\int_0^{\de}u^2|\chih|_{\gamma}^2(u,\ub,\theta^1,\theta^2)d\ub\leq 2\int_0^{\de}|\chih|_{\gamma}^2(1,\ub,\theta^1,\theta^2)d\ub.
\end{equation}
Our strategy is as follows. The bootstrap assumption \eqref{BA.trapped} will allow us to have an upper bound for $|\trch|$. This will in turn give us both upper and lower bounds for the integral $\int_0^{\de} u^2|\chih|_{\gamma}^2(u,\ub',\theta^1,\theta^2) d\ub'$. In particular, we can close the bootstrap assumption \eqref{BA.trapped}. We will then use the lower bound for $|\chih|_{\gamma}^2(u=\delta\at b,\ub,\theta^1,\theta^2)$ to get an upper bound for $\trch(u=\delta\at b,\ub=\delta,\theta^1,\theta^2)$, which in particular guarantees that $\trch(u=\delta\at b,\ub=\delta,\theta^1,\theta^2)<0$.

We now turn to the details. First, as an immediate consequence of the bootstrap assumption \eqref{BA.trapped} and Proposition \ref{trch.trapped.improved}, we have 
\begin{equation}\label{trch.trapped.improved.2}
|\trch(u,\ub,\theta^1,\theta^2)|\ls \f{1}{|u|}+\f1{|u|^2}\int_0^{\de} |\chih|_{\gamma}^2(1,\ub',\theta^1,\theta^2) d\ub'.
\end{equation}
Now consider the null structure equation
\begin{equation*}
 \nab_3\chih+\frac 12 \tr\chib \chih=\nab\widehat{\otimes} \eta+2\omegab \chih-\frac 12 \tr\chi \chibh +\eta\widehat{\otimes} \eta.
\end{equation*}
Contracting this equation with $\chih$, we have
\begin{equation*}
\f12 \nab_3|\chih|^2_{\gamma}+\f12\tr\chib |\chih|^2_{\gamma}-2\omb|\chih|^2_{\gamma}=\chih(\nab\widehat{\otimes}\eta-\f12\tr\chi\chibh+\eta\widehat{\otimes}\eta).
\end{equation*}
In the coordinate system introduced in Section \ref{coordinates}, we have
\begin{equation*}
\begin{split} 
&\quad\ -\f{1}{2\Omega}\left(\f{\partial}{\partial u}+d^{A}\f{\partial}{\partial \theta^{A}}\right)|\chih|^2_{\gamma}-\f{1}{\Omega|u|}|\chih|^2_{\gamma}+\f12\left(\tr\chib+\f{2}{\Omega|u|}\right)|\chih|^2_{\gamma}-2\omb|\chih|^2_{\gamma}\\
&=\chih(\nab\widehat{\otimes}\eta-\f12\tr\chi\chibh+\eta\widehat{\otimes}\eta).
\end{split}
\end{equation*}
Using $\omb=-\f12\nab_3(\log \Omega)$, we can rewrite this as
\begin{align}\label{main.eqn.trapped}
u^{-2}\f{\partial}{\partial u}(u^2\Omega^{-2}|\chih|^2_{\gamma})
 &=\f{\partial}{\partial u}(\Omega^{-2}|\chih|^2_{\gamma})+\f{2}{\Omega^{2}|u|}|\chih|^2_{\gamma}\\
\notag &=-2\Om^{-1}\chih\left(\nab\widehat{\otimes}\eta-\f12\tr\chi\chibh+\eta\widehat{\otimes}\eta\right)\\
\notag &\quad +\Om^{-1}\left(\tr\chib+\f{2}{\Omega|u|}\right)|\chih|^2_{\gamma}+d^A\f{\partial}{\partial \theta^A}(\Omega^{-2}|\chih|^2).
\end{align}
We now obtain pointwise bounds for the terms on the right hand side. After bounding $\Om^{-1}$ in $L^\infty$ using Proposition \ref{Omega}, the first term can be re-written schematically as 
$$\sum_{i_1+i_2=1}\p\q^{i_1}\nab^{i_2}\q+\p\q\trch.$$
Using the bounds in Theorem \ref{main.thm} together with \eqref{trch.trapped.improved.2}, we have
\begin{align}\label{trapped.1}
&\quad\ \left\|\sum_{i_1+i_2=1}\p\q^{i_1}\nab^{i_2}\q\right\|_{L^\infty(\S)}+\|\p\q\trch\|_{L^{\infty}(\S)}\\
\notag &\ls \f{\de a}{|u|^4}+\f{\delta a}{|u|^5}\int_0^{\de} |\chih|_{\gamma}^2(1,\ub',\theta^1,\theta^2) d\ub'.
\end{align}
For the second term on the right hand side of \eqref{main.eqn.trapped}, note that
\begin{equation*}
\begin{split}
\left\|\tr\chib+\f{2}{\Omega|u|}\right\|_{L^\infty(\S)}&=\left\|\tr\chib+\f{2}{|u|}\right\|_{L^\infty(\S)}+\left\|\f{2}{|u|}(\Omega^{-1}-1)\right\|_{L^\infty(\S)}\\
&\ls \f{\de \at}{|u|^2},
\end{split}
\end{equation*}
where we have used the bounds proved in Theorem \ref{main.thm} together with Proposition \ref{Om.trapped.improved}.
Therefore, for every $u\in[0,\delta\at b]$, $\ub\in [0,\de]$, we have
\begin{equation}\label{trapped.2}
\left|\Om^{-1}\left(\tr\chib+\f{2}{\Omega|u|}\right)\right|\chih|^2_{\gamma}|(u,\ub,\theta^1,\theta^2)\ls \f{\de \at}{|u|^2}|\chih|^2_{\gamma}(u,\ub,\theta^1,\theta^2).
\end{equation}
For the third term on the right hand side of \eqref{main.eqn.trapped}, we have
\begin{align}\label{trapped.3.0}
&\quad\ \left\|d^A\f{\partial}{\partial \theta^A}(\Omega^{-2}|\chih|^2)\right\|_{L^\infty(\S)}\\
\notag &\ls \|d\|_{L^\infty(\S)}\|\nab\chih\|_{L^\infty(\S)}\|\p\|_{L^\infty(\S)}\!+\!\|d\|_{L^\infty(\S)}\|\q\p\p\|_{L^\infty(\S)}\\
\notag & \ls \f{a}{|u|^3}\|d\|_{L^\infty(\S)}.
\end{align}
To obtain the estimate for $d$, we use the fact that
\begin{equation*}
[L,\Lb]=\f{\partial d^{A}}{\partial \ub}\f{\partial}{\partial \theta^{A}}
\end{equation*}
which implies
\begin{equation*}
\f{\partial d^{A}}{\partial \ub}=-4\Omega^2\zeta^{A}.
\end{equation*}
Integrating this in the $\ub$ direction and applying the bounds in Theorem \ref{main.thm} thus give
$$\|d\|_{L^\infty(\S)}\ls \f{\de^2\at}{|u|^2}.$$
Combining this with \eqref{trapped.3.0}, we thus have
\begin{equation}\label{trapped.3}
\left\|d^A\f{\partial}{\partial \theta^A}(\Omega^{-2}|\chih|^2)\right\|_{L^\infty(\S)}\ls \f{\de^2 a^{\f32}}{|u|^5}.
\end{equation}
Multiplying \eqref{main.eqn.trapped} by $u^2$, integrating in $u$ and using the bounds \eqref{trapped.1}, \eqref{trapped.2} and \eqref{trapped.3}, we thus have
\begin{equation*}
\begin{split}
&\quad\ \left|u^2|\chih|^2_{\gamma}(u,\ub,\theta^1,\theta^2)-|\chih|^2_{\gamma}(1,\ub,\theta^1,\theta^2)\right|\\
&\leq \f{Ca^{\f12}}{b}+\f{C \delta a}{|u|^2}\int_0^{\de} |\chih|_{\gamma}^2(1,\ub',\theta^1,\theta^2) d\ub'\\
&\quad +C\delta\at\int_u^1 |\chih|_{\gamma}^2(u',\ub,\theta^1,\theta^2)du', 
\end{split} 
\end{equation*}
for some universal constant $C$.
Integrating in $\ub$ from $0$ to $\de$, we obtain
\begin{equation*}
\begin{split}
&\quad\ \left|\int_0^{\delta}u^2|\chih|^2_{\gamma}(u,\ub',\theta^1,\theta^2)d\ub'- \int_0^{\delta}|\chih|^2_{\gamma}(1,\ub',\theta^1,\theta^2)d\ub'\right|\\
&\leq \f{C\delta a^{\f12}}{b}+\f{C \delta^2 a}{|u|^2}\int_0^{\de} |\chih|_{\gamma}^2(1,\ub',\theta^1,\theta^2) d\ub'\\
&\quad +C\de \at\int_0^{\de}\int_u^1 |\chih|_{\gamma}^2(u',\ub',\theta^1,\theta^2)du'd\ub'\\
&\leq \f{C\delta a^{\f12}}{b}+\f{C }{b}\int_0^{\de} |\chih|_{\gamma}^2(1,\ub',\theta^1,\theta^2) d\ub'\\
&\quad +C\de \at\left(\int_u^1 \f{1}{|u'|^2}du'\right)\left(\sup_{u'}\int_0^{\de}(u')^2|\chih|_{\gamma}^2(u',\ub,\theta^1,\theta^2)d\ub'\right)\\
&\leq \f{C\delta a^{\f12}}{b}+\f{C }{b}\int_0^{\de} |\chih|_{\gamma}^2(1,\ub',\theta^1,\theta^2) d\ub'+\f{2C\de \at}{|u|}\int_0^{\de}|\chih|_{\gamma}^2(1,\ub',\theta^1,\theta^2)d\ub'\\
&\leq \f{C\delta a^{\f12}}{b}+\f{3C }{b}\int_0^{\de} |\chih|_{\gamma}^2(1,\ub',\theta^1,\theta^2) d\ub',
\end{split}
\end{equation*}
where we have used the bootstrap assumption \eqref{BA.trapped}.
After choosing $b$ sufficiently large so that $\f Cb\leq \f1{15}$, this yields the upper bound
\begin{equation*}
\int_0^{\delta}u^2|\chih|^2_{\gamma}(u,\ub',\theta^1,\theta^2)d\ub'\leq \f65\int_0^{\delta}|\chih|^2_{\gamma}(1,\ub',\theta^1,\theta^2)d\ub'+\f{2C\delta a^{\f12}}{b}.  
\end{equation*}
and the lower bound
\begin{equation*}
\int_0^{\delta}u^2|\chih|^2_{\gamma}(u,\ub',\theta^1,\theta^2)d\ub'\geq \f45\int_0^{\delta}|\chih|^2_{\gamma}(1,\ub',\theta^1,\theta^2)d\ub'-\f{2C\delta a^{\f12}}{b}.  
\end{equation*}
Recalling that the assumption \eqref{main.lower.bd} implies
$$\int_0^{\delta}|\chih_0|^2(1,\ub',\theta^1,\theta^2) d\ub'\geq 4b \delta\at,$$
we can thus choose $b$ to be sufficiently large to guarantee that
\begin{equation}\label{BA.trapped.recover}
\int_0^{\delta}u^2|\chih|^2_{\gamma}(u,\ub',\theta^1,\theta^2)d\ub'\leq \f32\int_0^{\delta}|\chih|^2_{\gamma}(1,\ub',\theta^1,\theta^2)d\ub'.  
\end{equation}
and
\begin{equation}\label{lower.bd}
\begin{split}
\int_0^{\delta}u^2|\chih|^2_{\gamma}(u,\ub',\theta^1,\theta^2)d\ub'> \f34\int_0^{\delta}|\chih|^2_{\gamma}(1,\ub',\theta^1,\theta^2)d\ub'\geq 3b\delta\at.  
\end{split}
\end{equation}
In particular, \eqref{BA.trapped.recover} improves over the bootstrap assumption \eqref{BA.trapped}. Therefore, \eqref{lower.bd} holds.
We now use \eqref{lower.bd} to obtain the desired upper bound for $\trch$. To this end, we combine the transport equation for $\trch$
$$\nab_4 \tr\chi+\f12(\tr\chi)^2=-|\chih|^2-2\omega\tr\chi$$
and $\omega=-\f12\nab_4(\log \Omega)$ to get
\begin{equation*}
\begin{split}
\nab_4(\Omega^{-1} \tr\chi)=-\f{\Omega^{-1}}2(\tr\chi)^2-\Omega^{-1}|\chih|^2.
\end{split}
\end{equation*}
Using the fact $e_4=\Omega^{-1}\f{\partial}{\partial \ub}$, we integrate this equation to obtain
\begin{align}\label{trch.neg}
&\quad\ \Omega^{-1}\tr\chi(b\delta a^{\f12},\delta, \theta^1, \theta^2)\\
\notag &\leq  \Omega^{-1}\tr\chi(b\delta a^{\f12},0, \theta^1, \theta^2)-\int_0^{\delta}|\chih|^2(b\delta a^{\f12},\ub',\theta^1,\theta^2)d\ub'\\
\notag &\leq  \f{2}{b\delta a^{\f12}}-\f{3}{b\delta a^{\f12}}\\
\notag &<  0.
\end{align}
Here, we have used the lower bound \eqref{lower.bd} for the integral of $|\chih|_{\gamma}^2$.

Therefore, by \eqref{trchb.neg} and \eqref{trch.neg}, $S_{b\delta a^{\f12},\delta}$ is a trapped surface.
\end{proof}


\begin{thebibliography}{99} 

\bibitem[1]{An} X. An, \textit{Formation of trapped surfaces from past null infinity}, preprint (2012), {\tt arXiv:1207.5271}.

\bibitem[2]{Chr.1} D. Christodoulou, \textit{The formation of black holes and singularities in spherically symmetric gravitational collapse}, Comm. Pure Appl. Math. {\bf 44} (1991), no.~3, 339--373.

\bibitem[3]{Chr.1.5} D. Christodoulou, \textit{Bounded variation solutions of the spherically symmetric Einstein-scalar field equations}, Comm. Pure Appl. Math. {\bf 46} (1993), no.~8, 1131--1220.

\bibitem[4]{Chr.2} D. Christodoulou, \textit{Examples of naked singularity formation in the gravitational collapse of a scalar field}, Ann. of Math. (2) {\bf 140} (1994), no.~3, 607--653.

\bibitem[5]{Chr.3} D. Christodoulou, \textit{The instability of naked singularities in the gravitational collapse of a scalar field}, Ann. of Math. (2) {\bf 149} (1999), no.~1, 183--217.

\bibitem[6]{Chr:book} D. Christodoulou, \textit{The formation of black holes in general relativity}, Monographs in Mathematics, European Mathematical Soc. (2009). 

\bibitem[7]{Chr-Kl}D. Christodoulou and S. Klainerman, \textit{The global nonlinear stability of the Minkowski space}, Princeton Mathematical Series {\bf 41} (1993).

\bibitem[8]{Dafermos} M. Dafermos, \textit{The formation of black holes in general relativity}, Ast\'erisque {\bf 352} (2013).

\bibitem[9]{DHR} M. Dafermos, G. Holzegel, and I. Rodnianski, \textit{A scattering theory construction of dynamical vacuum black holes}, to appear in J. Diff. Geom. (2017).

\bibitem[10]{Holzegel} G. Holzegel, \textit{Ultimately Schwarzschildean spacetimes and the black hole stability problem}, preprint (2010), {\tt arXiv:1010.3216}.

\bibitem[11]{KNI:book} S. Klainerman and  F. Nicolo, \textit{The evolution problem in General Relativity}, Progress in Mathematical Physics, Birkha\"user (2003).

\bibitem[12]{KLR} S. Klainerman, J. Luk, and I. Rodnianski, \textit{A fully anisotropic mechanism for formation of trapped surfaces in vacuum}, Invent. Math. {\bf 198} (2014), no.1, 1-26.

\bibitem[13]{KR:Scarred} S. Klainerman and I. Rodnianski, 
\textit{On emerging scarred surfaces for the Einstein vacuum equations}, Discrete Contin. Dyn. Syst. {\bf 28} (2010), no.~3, 1007--1031.

\bibitem[14]{KR:Trapped} S. Klainerman and I. Rodnianski, 
\textit{On the the formation of trapped surfaces}, Acta Math. {\bf 208} (2012), no.~2, 211--333.

\bibitem[15]{LY} J. Li and P. Yu, \textit{Construction of Cauchy data of vacuum Einstein field equations evolving to black holes}, Ann. of Math. (2) {\bf 181} (2015), no.2, 699-768.

\bibitem[16]{L:local} J. Luk, \textit{On the local existence for the characteristic initial value problem in general relativity},  Int. Mat. Res. Notices {\bf 20} (2012), 4625--4678.

\bibitem[17]{L} J. Luk, \textit{Weak null singularities in general relativity}, preprint (2013), {\tt arXiv:1311.4970}.

\bibitem[18]{L-R:Propagation} J. Luk and I. Rodnianski, 
\textit{Local propagation of impulsive gravitational waves}, Comm. Pure Appl. Math. {\bf 68} (2015), no.4, 511-624.

\bibitem[19]{L-R:Interaction} J. Luk and I. Rodnianski, 
\textit{Nonlinear interactions of impulsive gravitational waves for the vacuum Einstein equations}, preprint (2013), {\tt arXiv:1301.1072}.
 
\bibitem[20]{Penrose} R. Penrose, \textit{Gravitational collapse and space-time singularities}, Phys. Rev. Lett. {\bf 14} (1965), 57--59. 

\bibitem[21]{R-T} M. Reiterer and E. Trubowitz, \textit{Strongly focused gravitational waves}, Comm. Math. Phys. {\bf 307} (2011), no.~2, 275--313.

\bibitem[22]{Yu1} P. Yu, \textit{Energy estimates and gravitational collapse}, Comm. Math. Phys. {\bf 317} (2013), no.~2, 273--316.

\bibitem[23]{Yu2} P. Yu, 
\textit{Dynamical Formation of black holes due to the condensation of matter field}, preprint (2011), {\tt arXiv:1105.5898}.
\end{thebibliography}
\end{document}